\documentclass[final,3p,times,12pt]{elsarticle}

\usepackage{amssymb}
\usepackage{amsthm}
\usepackage{amsmath}
\usepackage{latexsym}
\usepackage{times}
\usepackage{graphicx}
\usepackage{epsfig}
\usepackage{epstopdf}
\usepackage{mathrsfs}
\usepackage{bm} 

\usepackage[active]{srcltx}
\usepackage{subcaption}

\usepackage{hyperref}  		
\hypersetup{colorlinks=true,citecolor={red},linkcolor={blue},linktocpage} 
\usepackage{color}
\definecolor{armygreen}{rgb}{0.29, 0.33, 0.13}
\definecolor{maroon(x11)}{rgb}{0.69, 0.19, 0.38}
\definecolor{darkorange}{rgb}{1.0, 0.55, 0.0}
\definecolor{auburn}{rgb}{0.43, 0.21, 0.1}

\newtheorem{Theorem}{Theorem}
\newtheorem{Corollary}{Corollary}

\newcommand{\clo}{\overline}
\newcommand{\ra}{\rangle}
\newcommand{\la}{\langle}

\newcommand{\usc}{u_{\rm sc}}

\newcommand{\uinc}{u_{\rm inc}}

\newcommand{\Omm}{\Omega^{-}}

\newcommand{\uij}{u_{i,j}}
\newcommand{\uimj}{u_{i-1,j}}
\newcommand{\uiMj}{u_{i+1,j}}
\newcommand{\uijm}{u_{i,j-1}}
\newcommand{\uijM}{u_{i,j+1}}

\newcommand{\Sph}{\mathbb{S}}

\usepackage{color}
\definecolor{hotcolor}{rgb}{1,0,0}

\journal{ArXiv}

\begin{document}

\begin{frontmatter}




\title{\textbf{High order local absorbing boundary conditions for acoustic \\waves in terms of farfield expansions}}

\author[label1]{Vianey Villamizar\corref{cor1}}
\ead{vianey@mathematics.byu.edu}

\author[label2]{Sebastian Acosta}
\ead{sacosta@bcm.edu}
\ead[url]{sites.google.com/site/acostasebastian01}

\author[label1]{Blake Dastrup}
\ead{blakedast@gmail.com}

\address[label1]{Department of Mathematics, Brigham Young University, Provo, UT}
\address[label2]{Department of Pediatrics - Cardiology, Baylor College of Medicine, Houston, TX}
\cortext[cor1]{Corresponding author}

\begin{abstract}
We devise a new high order local absorbing boundary condition (ABC) for radiating problems and scattering of { time-harmonic acoustic waves} from obstacles of arbitrary shape. By introducing an artificial boundary $S$ enclosing the  scatterer, 
the original unbounded domain $\Omega$ is decomposed into a bounded computational domain $\Omega^{-}$ and an exterior unbounded domain $\Omega^{+}$. Then, we define interface conditions at the artificial boundary $S$, from truncated versions of the well-known Wilcox and Karp farfield expansion representations of the exact solution in the exterior region $\Omega^{+}$. As a result, we obtain a new local absorbing boundary condition (ABC) for a bounded problem on $\Omega^{-}$, which effectively accounts for the outgoing behavior of the scattered field. Contrary to the low order absorbing conditions previously defined, the order of the error induced by this ABC can easily match the order of the numerical method in $\Omega^{-}$. We accomplish this by simply adding as many terms as needed to the truncated farfield expansions of Wilcox or Karp. The convergence of these expansions guarantees that the order of approximation of the new ABC can be increased arbitrarily without having to enlarge the radius of the artificial boundary. We include numerical results in two and three dimensions which demonstrate the improved accuracy and simplicity of this new formulation when compared to other absorbing boundary conditions.
\end{abstract}


\end{frontmatter}


\section{Introduction} \label{Section.Intro}
Equations modeling wave phenomena  in fields such as geophysics, oceanography, and acoustics among others, are normally defined on unbounded domains. Due to the complexity of the corresponding boundary value problems (BVP),  in general, an explicit analytical technique cannot be found. Therefore, they are treated by numerical methods.
 Major challenges appear when numerically solving wave problems defined in these unbounded regions using volume discretization methods. One of them consists of the appropriate definition of absorbing boundary conditions (ABC) on artificial boundaries such that the solution of the new bounded problem approximates to a reasonable degree the solution of the original unbounded problem in their common domain.
That is why the  definition of ABCs for wave propagation problems in unbounded domain plays a key role in computation.

Historically two main approaches were initially followed in the evolution of ABCs, as described by Givoli in \cite{GivoliReview2}.  First, low order local ABCs were constructed. Undoubtedly, one of the most important ABCs in this category was introduced by Bayliss-Gunzburger-Turkel in their celebrated paper \cite{Bayliss01}. This condition is denoted as BGT in the ABC literature. Other well-known conditions in this category were introduced by Engquist-Majda \cite{Engquist01}, Feng \cite{Feng} and Li-Cendes \cite{Li-Cendes}.
Some of them became references for many that followed thereafter. Several years later in the late 1980s and early 1990s, exact non-local ABCs made their appearance. Since their definitions are based on Dirichlet-to-Neumann (DtN) maps, they are called DtN absorbing boundary conditions. The pioneer work in their formulations and implementations was performed by Keller-Givoli \cite{Keller01,Givoli-Keller1990} and Grote-Keller \cite{Grote-Keller01}. The main virtue of the DtN absorbing conditions is that they approximate the field at the artificial boundary almost exactly. Therefore, the accuracy of the numerical computation depends almost entirely on the accuracy of the numerical method employed for the computation at the interior points. 

The BGT absorbing condition consists of a sequence of differential operators applied at the artificial boundary (a circle or a sphere of radius $R$) which progressively annihilate the first terms of a farfield expansion of the outgoing wave valid in the exterior of the artificial boundary. We call the first of these operators BGT$_1$. In three dimensions, it provides an accuracy of $O(1/R^{3})$ and involves a first order normal derivative. The next condition in this sequence, BGT$_2$ has  $O(1/R^{5})$ accuracy and includes a second order normal derivative in its definition. They are called BGT operators of order one and order two, respectively.
The drawback of the BGT and of the other ABCs in the first category is that to increase the order of the approximation at the boundary, the order of the derivatives present in their definitions also needs  to be increased. This leads to  impractical ABCs due to the difficulty found in their implementations beyond the first two orders. There is also a downside for the DtN-ABCs stemming from the fact the computation of the field at any boundary point involves all the other boundary points which leads to partially dense matrices at the final stage of the numerical computation.  

The above disadvantages are overcome by the introduction of {\it high order local} ABCs without high order derivatives. According to \cite {GivoliReview2}, they are sequences of ABCs of increasing accuracy which are also practically implementable for an arbitrarily high order. 
Several ABCs have been formulated within this category in recent years. A detailed description of some of them is found in \cite {GivoliReview2}. A common feature of all these high order local ABCs is the presence of auxiliary variables which avoid the occurrence of high derivatives (beyond order two) in the ABC's formulation. Probably, the best known of all these high order local conditions was formulated by Hagstrom-Hariharan \cite{Hagstrom98} which we denote as HH. They start representing the outgoing solution by an infinite series in inverse powers of $\frac{1}{R}$, where $R$ is the radius of a circular or spherical artificial boundary. 
Analogous to the BGT formulation, the key idea in this formulation is the construction of a sequence of operators that approximately annihilate the residual of the preceding term in the sequence. As a result, a sequence of conditions in the form of recurrence formulas for a set of unknown auxiliary variables is obtained. The expression for the first auxiliary variable coincides with BGT$_1$. Similarly, combining the formulas for the first two auxiliary variables, the HH absorbing condition reduces to BGT$_2$. Actually, Zarmi \cite{ZarmiThesis} proves that HH is equivalent to BGT for all orders. The difference between these two formulations is that HH does not involve high derivatives owing to the use of the auxiliary variables. Thus, HH can be implemented for arbitrarily high order. The three-dimensional (3D) HH can be considered an exact ABC since it is obtained from an exact representation of the solution in the exterior of the artificial boundary. However, the two-dimensional (2D) HH is only asymptotic because it is obtained from an asymptotic expansion of the exact representation of the solution. Recently, Zarmi-Turkel \cite{Zarmi-Turkel} generalized the HH construction of local high order ABCs. They developed an annihilating technique that can be applied to rather general series representation of the solution in the exterior of the computational domain.  As a result, they were able to reobtain HH and derive new high order local ABCs such as a high order extension of Li-Cendes ABC \cite{Li-Cendes}. 

Our construction of high order local ABCs proceeds in the opposite way of the  previous ABCs discussed above. Instead of defining local differential operators which progressively annihilate the first terms of a series representation of the solution in the exterior of the artificial boundary, we use a truncated version of the series representations directly to define the ABC without defining special differential operators at the boundary.  As a consequence, the derivation of the absorbing condition is extremely simple. Moreover, the order of the error induced by this ABC can be easily improved by simply adding as many terms as needed to the truncated farfield expansions 
The series representations employed are Karp's farfield expansion \cite{Karp} in 2D, and Wilcox's farfield expansion \cite{Wil-1956} in 3D. They are exact representations of the outgoing wave outside the circular or spherical artificial boundaries of radius $R$, respectively. Therefore, the resulting ABC which we call 
Karp's double farfield expansion (KDFE) and Wilcox farfield expansion (WFE), respectively, can be considered exact ABCs. The exact character of KDFE represents an improvement over HH in 2D, which is only asymptotically valid. Instead of having unknown auxiliary functions as part of the new condition, we simply consider as unknowns the original angular functions appearing in Wilcox's or Karp's farfield expansions. To determine these angular functions, we use the recurrence formulas derived from Wilcox's or Karp's theorems which do not disturb the local character of the ABC. A relevant feature of the farfield expansions approach is that the coefficient (angular function) of their leading term is the farfield pattern of the propagating wave. Thus, no additional computation is required to obtain an approximation for this important profile. For comparison purposes, we also obtain a farfield expansion ABC from the asymptotic farfield expansion of Karp's exact series. We call it Karp's single farfield expansion (KSFE) absorbing boundary condition. 

An important consideration is that the formulation of these absorbing boundary conditions depends on existing knowledge of an exact or asymptotic series representation for the outgoing waves of the problem being studied. This limits the use of Karp and Wilcox farfield expansions ABCs to problems in the entire plane or space, respectively.  As a consequence, problems involving straight  infinite boundaries as waveguide problems, half-plane, or quarter-plane cannot be modeled by these ABCs. For these type of problems, the most popular method to formulate ABCs is the perfectly matched layer (PML) introduced by Berenger \cite{Berenger01}. However, a class of high order absorbing boundary conditions has also been employed by several authors. For instance, 
Hagstrom, Mar-Or, and Givoli \cite{H-MO-Givoli2008} obtained high order local ABCs for two-dimensional waveguide problems modeled by the wave equation. This ABC was first formulated for the wave equation by Hagstrom-Warburton \cite{H-W} which in turn is based on a modification of the Higdon ABCs \cite{Higdon1987}. More recently,
Rabinovich and et al. \cite{Rab-Giv-Bec-2010} adapted Hagstrom-Warburton ABC to time-harmonic problems in a waveguide and a quarter-plane modeled by the Helmholtz equation. 

The outline of the succeeding sections is as follows. In Section \ref{Section.Formulation}, details about the expansions KDFE, KSFE, and WFE are given. Also, the relationships between lower orders of KSFE absorbing boundary condition, BGT$_1$, and BGT$_2$ are established. Then in Section \ref{Section:NumMethd}, the numerical method is described in the 2D case for KDFE. In particular, the discrete equations at the boundary are carefully derived. This is followed by an analysis of the structure of the matrices defining the ultimate linear systems for KSFE, KDFE, and DtN boundary value problems, respectively. Finally, numerical results for scattering and radiating problems, from circular and complexly shaped obstacles in 2D, and also from spherical obstacles in 3D, employing the novel farfield expansions ABCs, are reported in Section \ref{Section.Numerics2D}.

\section{High order local absorbing conditions from farfield expansions
} \label{Section.Formulation}

We start this Section by considering the scattering problem of a time-harmonic incident wave, $\uinc$, from a single obstacle in two or three dimensions. This scatterer is an impenetrable obstacle that occupies a simply connected bounded region  with boundary $\Gamma$. The open unbounded region in the exterior of $\Gamma$ is denoted as 
$\Omega$. This region $\Omega$ is occupied by a homogeneous and isotropic medium. Both the incident field $\uinc$ and the scattered field $\usc$ satisfy the Helmholtz equation in $\Omega$. For simplicity, we assume a Dirichlet boundary condition (soft obstacle)  on $\Gamma$. However,  the analysis in this article can be easily extended to Neumann or Robin boundary conditions, and to a bounded penetrable scatterer with inhomogeneous and anisotropic properties. Then, $\usc$ solves the following boundary value problem (BVP):
\begin{eqnarray}
&& \Delta \usc + k^2 \usc = f \quad\qquad \text{in $\Omega$}, \label{BVPsc1} \\
&& \usc = - \uinc \qquad\qquad \text{on $\Gamma$,} \label{BVPsc2} \\
&& \lim_{r \rightarrow \infty} r^{(\delta-1)/2} \left( \partial_{r} \usc- \mathrm{i} k \usc \right) = 0.\label{BVPsc3}
\end{eqnarray}
 {The wave number $k$ and the source $f$ may vary in space}. Equation (\ref{BVPsc3}) is known as the Sommerfeld radiation condition where $r = |\textbf{x}|$ and $\delta=2$ or 3 for two or three dimensions, respectively. It implies that $\usc$ is an outgoing wave. This boundary value problem is well-posed under classical and weak formulations \cite{ColtonKress02,Nedelec01,McLean2000}.

As pointed out in the introduction, the unbounded BVP (\ref{BVPsc1})-(\ref{BVPsc3}) needs to be transformed into a bounded BVP before a numerical solution can be sought. 
This is typically done by introducing an artificial boundary $S$ enclosing the obstacle followed by defining an appropriate absorbing boundary condition (ABC) on $S$. We choose a circular or spherical artificial boundary for the two- or three-dimensional scenarios, respectively. As a result, the region $\Omega$ is divided into two open regions. The region 
$\Omega^{-},$ bounded internally by the obstacle boundary $\Gamma$ and externally by the artificial boundary $S$ (a circle or a sphere of radius $r=R$), and the open unbounded connected region $\Omega^{+}= \Omega \setminus \clo{\Omega^{-}}$.  {We assume that the source $f$ has its support in $\Omega^{-}$, and the wave number $k$ is constant in $\Omega^{+}$}. An appropriate ABC should induce no or little spurious reflections from the artificial boundary $S$ in order to maintain a good accuracy for the numerical solution inside $\Omega^{-}$.  

As an intermediate step before constructing our high order local ABC in the next sections, we consider the following equivalent interface problem to the original BVP (\ref{BVPsc1})-(\ref{BVPsc3}) for  $\usc^{-} = \usc |_{\Omega^{-}}$ and $\usc^{+} = \usc |_{\Omega^{+}}$:
\begin{eqnarray}
&& \Delta \usc^{-} + k^2 \usc^{-} = f, \quad\qquad  \text{in $\Omega^{-}$}, \label{BVPInterf1} \\
&& \Delta \usc^{+} + k^2 \usc^{+} = 0, \quad\qquad  \text{in $\Omega^{+}$}, \label{BVPInterf2} \\
&& \usc^{-} = - \uinc, \qquad\qquad\,  \text{on $\Gamma$,} \label{BVPInterf3}
\end{eqnarray}
\noindent with the interface and Sommerfeld conditions:
\begin{eqnarray}
&& \usc^{-} = \usc^{+}, \qquad\qquad\quad  \text{on $S$,}\label{BVPInterf4}  \\
&&\partial_{\nu} \usc^{-} = \partial_{\nu} \usc^{+}  \quad\qquad\,\,\,\,\text{on $S$}, \label{BVPInterf5} \\
&& \lim_{r \rightarrow \infty} r^{(\delta-1)/2}\left( \partial_{r} \usc^{+} - \mathrm{i} k \usc^{+} \right) = 0,  \label{BVPInterf6}
\end{eqnarray}
where $\partial_{\nu}$ denotes the derivative in the outer normal direction on $S$. The original scattering problem (\ref{BVPsc1})-(\ref{BVPsc3}), and the interface problem (\ref{BVPInterf1})-(\ref{BVPInterf6}) are equivalent as shown in \cite[Thm 1]{JCP2010} or \cite[Lemma 4.19]{McLean2000}. 
As a consequence, by simply requiring the Cauchy data to match at the artificial boundary $S$, all higher order derivatives also match at the interface.  {This matching condition at the artificial boundary will ultimately lead to a bounded BVP in $\Omega^{-}$ whose numerical solution approximates to a reasonable degree the solution of the original unbounded problem in $\Omega^{-}$. This bounded BVP is constructed by realizing that there is an analytical representation of the solution $\usc^{+}$ for the portion of the interface problem defined in $\Omega^{+}$. By matching, at the artificial boundary $S$, this analytical solution with the solution $\usc^{-}$ defined in the interior region $\Omega^{-}$, the bounded BVP sought in $\Omega^{-}$ is finally obtained. The numerical solution of this bounded BVP in $\Omega^{-}$ is the main subject of this work. Furthermore, once this numerical solution for $\usc^{-}$ is obtained, the analytical representation for $\usc^{+}$ can be evaluated in $\Omega^{+}$. Details of the derivation of the bounded BVP in $\Omega^{-}$ are given in the sections below.  Moreover, since the problem in $\Omega^{-}$ is to be solved numerically, we can consider a rather general source term $f$ and a variable wave number $k$ inside $\Omega^{-}$. However, for sake of simplicity, from now on we assume $f=0$ and $k$ constant.}

\subsection{Karp's double farfield expansion (KDFE) absorbing boundary condition in 2D} \label{Section.ABC2D}

Here, we consider the outgoing field $\usc^{+}$ satisfying the 2D Helmholtz equation exterior to a circle $r=R$ and the Sommerfeld radiation condition (\ref{BVPsc3}) for $\delta = 2$. Our derivation of the new exact absorbing boundary condition is based on a well-known representation of outgoing solutions of the Helmholtz equation in 2D by two infinite series in powers of  $1/kr$. This representation is provided by the following theorem due to Karp.

\begin{Theorem}[Karp \cite{Karp}] \label{Thm.Karp}
Let $\usc^{+}$ be an outgoing solution of the two-dimensional Helmholtz equation in the exterior region to a circle $r=R$. Then, $\usc^{+}$ can be represented by a convergent expansion
\begin{eqnarray}
\usc^{+} (r,\theta) = H_0(kr) \sum_{l=0}^{\infty} \frac{F_l(\theta)}{(kr)^l} + H_1(kr) \sum_{l=0}^{\infty} \frac{G_l(\theta)}{(kr)^l}, \qquad \mbox{for}\,\, r> R. \label{KarpExp}
\end{eqnarray}
This series is uniformly and absolutely convergent for $r>R$ and can be differentiated term by term with respect to $r$ and $\theta$ any number of times.
\end{Theorem}
\noindent Here, $r$ and $\theta$ are polar coordinates. The functions $H_0$ and $H_1$ are Hankel functions of first kind of order 0 and 1, respectively. Karp also claimed that the terms $F_l$ and $G_l$ ($l=1,2,\dots$) can be computed recursively from $F_0$ and $G_0$. To accomplish this, he suggested the substitution of the expansion (\ref{KarpExp}) into Helmholtz equation in polar coordinates and the use of the identities:
$H_0'(z)=-H_1(z)$ and $H_1'(z)= H_0(z) - \frac{1}{z}H_1(z)$. In fact, by doing this and requiring the coefficients of $H_0$ and $H_1$ to vanish, we derive a recurrence formula for the coefficients $F_l$ and $G_l$ of the expansion (\ref{KarpExp}). This result is stated in the following corollary.
\begin{Corollary} \label{KarpRecurrence}
The coefficients $F_l(\theta)$ and $G_l(\theta)$ ($l>1$) of the expansion (\ref{KarpExp}),  can be determined from $F_0(\theta)$ and $G_0(\theta)$ by the recursion formulas
\begin{align}
& 2 l G_{l}(\theta) = (l-1)^2 F_{l-1}(\theta) + d^2_{\theta}  F_{l-1}(\theta) , \qquad && \text{for $l=1,2, \dots$} \label{Recurrence1}\\
& 2 l F_{l}(\theta) = - l^2 G_{l-1}(\theta) - d^2_{\theta} G_{l-1}(\theta), 
\qquad && \text{for $l=1,2, \dots$}. \label{Recurrence2}
\end{align}
\end{Corollary}

 {As discussed in the previous section, we use the semi-analytical representation of $\usc^{+}$ given by (\ref{KarpExp}) and the matching conditions
 (\ref{BVPInterf4})-(\ref{BVPInterf5}), at the interface $S$, to obtain an approximation $u \approx \usc^{-}$ that satisfies the following bounded BVP in the region $\Omm$:}
 \begin{eqnarray}
&& \Delta u + k^2 u = f, \quad\qquad  \text{in $\Omega^{-}$}, \label{BVPBd1} \\
&& u = - \uinc, \qquad\qquad\,  \text{on $\Gamma$,} \label{BVPBd2} \\
&& u(R,\theta)=H_0(kR) \sum_{l=0}^{L-1} \frac{F_l(\theta)}{(kR)^l} + H_1(kR)\sum_{l=0}^{L-1} \frac{G_l(\theta)}{(kR)^l},\label{BVPBd3} \\
&&\partial_{r} u(R,\theta) = \partial_{r}\left( H_0(kr) \sum_{l=0}^{L-1} \frac{F_l(\theta)}{(kr)^l} + H_1(kr)\sum_{l=0}^{L-1} \frac{G_l(\theta)}{(kr)^l}\right) \bigg|_{r=R},\label{BVPBd4} 
\end{eqnarray}
where $R$ is the radius of the circular artificial boundary $S$. This problem is not complete until enough conditions at the artificial boundary $S$, for the two families of unknown angular functions $F_l$ and $G_l$ of Karp's expansion, are specified. Clearly, extra conditions to determine $F_l$ and $G_l$ for $l=1,\dots L-1$ are provided by the recurrence formulas (\ref{Recurrence1}) and (\ref{Recurrence2}). To apply these recurrence formulas, $F_0$ and $G_0$ need to be known.  {The boundary conditions
$(\ref{BVPBd3})$ and $(\ref{BVPBd4})$ may be used to determine $u$ and $F_0$ at the boundary $S$. Therefore, we are still short by another condition to determine $G_0$ at $S$. 
Now, $\usc$ has a second order partial derivative which is continuous with respect to $r$ at $r=R$. Thus, a natural condition to add at $r=R,$ to our new bounded problem (\ref{BVPBd1})-(\ref{BVPBd4}) supplemented  with (\ref{Recurrence1})-(\ref{Recurrence2}), is $\partial_{\nu}^2 \usc^{-} = \partial_{\nu}^2 \usc^{+} $ which can be fully written in terms of $u$ as}
\begin{eqnarray}
\partial_{r}^2 u(R,\theta) = \partial_{r}^2\left( H_0(kr)\sum_{l=0}^{L-1} \frac{F_l(\theta)}{(kr)^l} + H_1(kr)\sum_{l=0}^{L-1} \frac{G_l(\theta)}{(kr)^l}\right) \bigg|_{r=R}, \label{BVPBd5}
\end{eqnarray}
where the second radial derivative $\partial_{r}^2 u$ may also be expressed in terms of $\partial_{r} u$ and $\partial_{\theta}^2 u$ using the Helmholtz equation itself.

 {Summarizing, we approximate the solution of the interface problem (\ref{BVPInterf1})-(\ref{BVPInterf6}) in the region $\Omega^{-}$ by the solution of the bounded BVP consisting of (\ref{BVPBd1})-(\ref{BVPBd5}) and (\ref{Recurrence1})-(\ref{Recurrence2}). The equations (\ref{BVPBd3})-(\ref{BVPBd5}) for the \textit{double} family of farfield functions $F_l$ and $G_l$, supplemented by the recurrence formulas (\ref{Recurrence1})-(\ref{Recurrence2}),  constitute our novel \textit{Karp's Double Farfield Expansion} $(\text{KDFE}_{L})$ absorbing boundary conditions with $L$ terms.}


\subsection{Karp's single farfield expansion (KSFE) absorbing boundary condition in 2D} \label{Section.ABC2DAsymp}

It is possible to approximate the two-family expansion (\ref{KarpExp}) with a one-family expansion by means of an asymptotic approximation for large values of $kr$. A similar procedure was employed in \cite{Bayliss01}. The Hankel functions $H_{0}(z)$ and $H_{1}(z)$ admit the following approximations \cite[\S 9.2]{HandMathFunct}, 
\begin{align}
& H_{0}(z) = \frac{e^{i z}}{\sqrt{z}} \sum_{l=0}^{L-1} \frac{C_{0,l}}{z^l} + O(|z|^{-L}) \qquad \text{and} \qquad H_{1}(z) = \frac{e^{i z}}{\sqrt{z}} \sum_{l=0}^{L-1} \frac{C_{1,l}}{z^l} + O(|z|^{-L}) \label{Eqn:Hpower}
\end{align}
valid for $z \in \mathbb{C}$ with $| \text{arg}(z) | < \pi$ as $|z| \to \infty$. Therefore, after multiplication of the power series of (\ref{KarpExp}) with these approximations for $H_{0}(kr)$ and $H_{1}(kr)$, re-arranging terms of same powers, and neglecting the terms $O(|kr|^{-L})$, we can combine the two families of angular functions $F_l$ and $G_l$ into one family $f_l$. As a result, a new asymptotic series representation of the outgoing wave (\ref{BVP2D3_Asymp} ) is obtained.
Moreover, the application of the 2D Helmholtz operator to the new asymptotic expansion renders a recursive formula (\ref{BVP2D5_Asymp}) for the functions $f_{l}$.  {Thus, in virtue of the approximation (\ref{Eqn:Hpower}) and the matching at the artificial boundary $S$ described in the previous section,  we obtain a new absorbing boundary condition for the  problem (\ref{BVPBd1})-(\ref{BVPBd2}) given by}
\begin{eqnarray}
&& u(R,\theta)=\frac{e^{ikR}}{\sqrt{kR}}\sum_{l=0}^{L-1} \frac{f_l(\theta)}{(kR)^l}\label{BVP2D3_Asymp} \\
&&\partial_{r} u(R,\theta) =  \frac{e^{i k R}}{\sqrt{k R}} \sum_{l=0}^{L-1} \left( ik - \left( l +\tfrac{1}{2} \right)/R  \right) \frac{f_l(\theta)}{(kR)^l}, \label{BVP2D4_Asymp}\\
&& 2 i l f_{l}(\theta) = \left( l - \tfrac{1}{2} \right)^2 f_{l-1}(\theta) + \partial_{\theta}^{2} f_{l-1}(\theta), \qquad l \geq 1. \label{BVP2D5_Asymp}
\end{eqnarray}
We call the boundary condition defined by (\ref{BVP2D3_Asymp})-(\ref{BVP2D5_Asymp}) with a \textit{single} family of farfield functions $f_{l}$ the \textit{Karp's Single Farfield Expansion} $(\text{KSFE}_{L})$ absorbing boundary condition with $L$ terms. As we see in the numerical results in Section \ref{Section.Numerics2D}, both the KSFE$_{L}$ and KDFE$_{L}$ render similar results  as the number of terms $L$ increases, for moderate to large values of $kR$. But, KSFE$_{L}$ exhibits a slower convergence behavior. However, we warn that (as discussed in \cite{Karp}) the approximations (\ref{Eqn:Hpower}) cannot be convergent for fixed $|z|$ as $L \to \infty$, because the Hankel functions possess a branch cut on the negative real axis which prevents them to be expanded by any Laurent series. Thus the number $L$ should be chosen judiciously, especially for small values of $kR$.


\subsubsection{Relationship between KSFE and BGT absorbing conditions}

First, we consider the relationship between the BVPs corresponding to KSFE$_{1}$ and BGT$_{1}$ (the first order ABC from \cite{Bayliss01}). More precisely, we consider $u_{1}$ solving a BVP corresponding to the 
KSFE$_1$ condition (KSFE$_1$-BVP):
\begin{eqnarray}
&& \Delta u_1+ k^2 u_1= 0, \qquad\qquad\qquad\quad\qquad  \text{in $\Omega^{-}$}, \label{BVP1_1} \\
&& u_1 = - \uinc, \quad\qquad\qquad\qquad\qquad\qquad\,  \text{on $\Gamma$,} \label{BVP1_2} \\
&& u_1(R,\theta)= e^{ikR}\frac{f_0(\theta)}{(kR)^{1/2}}\label{BVP1_3} \\
&&\partial_{r} u_1(R,\theta) = \partial_{r}\left(e^{ikr}\frac{f_0(\theta)}{(kr)^{1/2}}  \right) \bigg|_{r=R}=
 e^{ikR} \frac{f_{0}(\theta)}{(kR)^{1/2}}   \left(  ik  - \frac{1}{2R} \right),\label{BVP1_4} 
\end{eqnarray}
and $U_{1}$ solving a BVP corresponding to the BGT$_1$ condition (BGT$_1$-BVP):
\begin{eqnarray}
&& \Delta U_1+ k^2 U_1= 0, \qquad\qquad\qquad\quad\qquad  \text{in $\Omega^{-}$}, \label{BGT1_1} \\
&& U_1 = - \uinc, \quad\qquad\qquad\qquad\qquad\qquad\,  \text{on $\Gamma$,} \label{BGT1_2} \\
&& \partial_{r}U_1(R,\theta) + \frac{1}{2R}U_1(R,\theta) - ikU_1(R,\theta) = 0.\label{BGT1_3} 
\end{eqnarray}
It is clear from combining (\ref{BVP1_3}) and (\ref{BVP1_4}) that a solution $u_{1}$ of (\ref{BVP1_1})-(\ref{BVP1_4}) also
satisfies the BVP (\ref{BGT1_1})-(\ref{BGT1_3}). Conversely, if $U_1$ is a solution of (\ref{BGT1_1})-(\ref{BGT1_3}), then by defining $f_0(\theta)=U_1(R,\theta)(kR)^{1/2}e^{-ikR}$, we immediately show that $U_1$ is a also a solution of (\ref{BVP1_1})-(\ref{BVP1_4}). Furthermore, 
the BVP (\ref{BGT1_1})-(\ref{BGT1_3}) has a unique solution as shown in \cite{Bayliss01}. As a consequence, the BVPs defined by the BGT$_1$ and KSFE$_1$ conditions have the same unique solution, which we state in the form of a theorem.

\begin{Theorem} \label{Equiv1}
The boundary value problems (\ref{BVP1_1})-(\ref{BVP1_4}) and (\ref{BGT1_1})-(\ref{BGT1_3}) are equivalent and they have a unique solution.
\end{Theorem}

Secondly, we analyze if the BVPs corresponding to KSFE$_2$ and BGT$_2$ are equivalent. The KSFE$_2$-BVP consists of finding a function $u_2$ satisfying Helmholtz equation in $\Omega^{-}$, Dirichlet boundary condition on $\Gamma$, and the following absorbing boundary condition on $S$:
\begin{eqnarray}
&& u_2(R,\theta)= \frac{e^{ikR}}{(kR)^{1/2}}\left( f_0(\theta)+ \frac{f_1(\theta)}{kR}\right)\label{BVP2_3} \\
&&\partial_{r} u_2(R,\theta) = \frac{e^{ikR}}{(kR)^{1/2}} \left(  \left(ik - \frac{1}{2R} \right)f_{0}(\theta) + \left( ik - \frac{3}{2R} \right) \frac{f_{1}(\theta)}{kR}   \right), \label{BVP2_4} \\
&& 2 i f_1(\theta) = \frac{1}{4} f_0(\theta) + f_0''(\theta).\label{BVP2_5}
\end{eqnarray}
Similarly, the BGT$_2$-BVP consists of finding a function $U_2$ satisfying Helmholtz equation in $\Omega^{-}$, Dirichlet boundary condition on $\Gamma$, and the following absorbing boundary condition on $S$:
\begin{equation}
 \partial_{r}U_2 = \frac{( 2(kR)^2 + 3ikR - 3/4) U_2 + \partial^{2}_{\theta}U_2}{2R(1 - ikR)} 
,\label{BGT2ABC} 
\end{equation}
Next, we will prove the following statement about the relationship between the BVPs corresponding to KSFE$_2$, KSFE$_3$, and BGT$_2$. 

\hspace{1cm}
\begin{Theorem}\label{Non-Equiv}
\hfill
\begin{enumerate}
\item
A solution $u_2$ of KSFE$_2$-BVP satisfies BGT$_2$-BVP only up to $O( R^{-7/2})$ at the artificial boundary $S$.
\item
A solution $u_3$ of KSFE$_3$-BVP satisfies BGT$_2$-BVP up to $O(R^{-9/2})$ at the artificial boundary $S$.
\end{enumerate}
\end{Theorem}
\begin{proof}
We will prove statement (a) by showing that when $U_{2}$ is replaced by $u_{2}$ in (\ref{BGT2ABC}), then the left hand side (lhs) of (\ref{BGT2ABC}) is equal to its right hand side (rhs) up to $O(R^{-3/2})$.
To obtain the expression for the lhs, we replace $\partial_r U_2$ in (\ref{BGT2ABC}) with $\partial_r u_2$ and use (\ref{BVP2_4}). This leads to
\begin{equation}
\text{lhs} = \left[ \left( ik - \frac{1}{2R} \right) f_{0} + \left( ik - \frac{3}{2R} \right) \frac{f_{1}}{kR}  \right] \frac{e^{ikR}}{(kR)^{1/2}}. \label{lhs2}
\end{equation}
On the other hand, replacing $U_2$ by $u_2$ defined by (\ref{BVP2_3}) into the rhs of (\ref{BGT2ABC}), we obtain,
\begin{equation}
\text{rhs} = \frac{1}{2R \left( 1 - ikR \right)} \left[ \left( 2(kR)^2 + 3ikR - \frac{3}{4} \right)\left( f_{0} + \frac{f_{1}}{kR} \right) + f''_{0} + \frac{f''_{1}}{kR} \right] \frac{e^{ikR}}{(kR)^{1/2}}. \label{rhs2}
\end{equation}
Now, using the recurrence formula (\ref{BVP2_5}) in (\ref{lhs2})-(\ref{rhs2}), we obtain,
\begin{equation}
\left( 1- ikR \right) \left( \text{lhs} - \text{rhs} \right) =  \frac{ik e^{ikR}}{(kR)^{5/2}} \left( \frac{9}{16} f_{0} + \frac{5}{2} f''_{0} + f''''_{0} \right). \label{diff}
\end{equation}
Hence, division by $(1-ikR)$ renders the statement (a).

A similar procedure leads to the proof of statement (b). First, we consider BVP defining the absorbing condition KSFE$_3$ which consists of finding a function $u_3$ satisfying Helmholtz equation in $\Omega^{-}$, Dirichlet boundary condition on $\Gamma$, and the following absorbing boundary condition on $S$:
\begin{eqnarray}
&& u_3(R,\theta)= \frac{e^{ikR}}{(kR)^{1/2}}\left(f_0(\theta) + \frac{f_1(\theta)}{kR} +  \frac{f_2(\theta)}{(kR)^{2}}\right), \label{BVP3_3} \\
&& \partial_{r} u_3(R,\theta) = \frac{e^{ikR}}{(kR)^{1/2}} \left( \left( ik - \frac{1}{2R} \right) f_{0}(\theta) + \left( ik - \frac{3}{2R} \right) \frac{f_{1}(\theta)}{kR} + \left( ik - \frac{5}{2R} \right) \frac{f_{2}(\theta)}{(kR)^2} \right),\label{BVP3_4} \\
&& 2i f_1(\theta) = \frac{1}{4}f_0(\theta) + f_0''(\theta), \label{BVP3_5} \\
&& 4i f_2(\theta) = \frac{9}{4}f_1(\theta) + f_1''(\theta). \label{BVP3_6}
\end{eqnarray}
When replacing $U_3$ with $u_3$, then lhs of (\ref{BGT2ABC}) becomes equal to
\begin{equation}
\text{lhs} = \frac{e^{ikR}}{(kR)^{1/2}} \left( \left( ik - \frac{1}{2R} \right) f_{0}(\theta) + \left( ik - \frac{3}{2R} \right) \frac{f_{1}(\theta)}{kR} + \left( ik - \frac{5}{2R} \right) \frac{f_{2}(\theta)}{(kR)^2} \right). \label{lhs3}
\end{equation}

Similarly, substituting $u_3$ into the rhs of (\ref{BGT2ABC}) leads to
\begin{equation}
2R (1 - ikR) \, \text{rhs} = \left( 2 (kR)^2 + 3ikR - 3/4 \right) \left( f_0 + \frac{f_1}{kR} + \frac{f_{2}}{(kR)^2} \right) + f''_0 + \frac{f''_1}{kR} + \frac{f''_2}{(kR)^2}.  \label{rhs3}
\end{equation}
Then, using the recurrence formulas (\ref{BVP3_5})-(\ref{BVP3_6}), we obtain that $(1 - ikR)\left( \text{lhs} - \text{rhs} \right) = O (R^{-7/2})$.
Finally, the statement (b) is proved by dividing both sides by $(1 - ikR)$.
\end{proof}

It was shown in \cite{Bayliss01} that a solution $U_2$ of the BGT$_2$-BVP approximates the exact solution of (\ref{BVPsc1})-(\ref{BVPsc3}) to $O (R^{-9/2} )$ when $R\rightarrow\infty$. From our previous results, 
we conclude that BGT$_2$-BVP and KSFE$_2$-BVP are not equivalent. Since a solution of KSFE$_2$-BVP satisfies BGT$_2$  to $O(R^{-7/2})$, a solution of KSFE$_2$-BVP will be a poorer approximation to the exact solution than $U_2$. However, the solution of KSFE$_3$-BVP satisfies BGT$_2$ to $O (R^{-9/2})$ also. It means that the solutions of BGT$_2$-BVP and KSFE$_3$-BVP approximate the exact solution at a comparable rate. This behavior is confirmed in our numerical experiments in Section \ref{Section.Numerics2D}.


\subsection{Wilcox's farfield expansion absorbing boundary condition in 3D} \label{Section.ABC3D}
For the 3D case ($\delta=3$), we also use a representation of outgoing waves by an infinite series in powers of $1/{kr}$. This representation is provided by a well-known theorem due to Atkinson and Wilcox, which is stated here for completeness.

\begin{Theorem}[Atkinson-Wilcox \cite{Wil-1956}]\label{Thm.Wilcox}
Let $\usc^{+}$ be an outgoing solution of the three-dimensional Helmholtz equation in the exterior region to a sphere of radius $r=R$. Then, $\usc^{+}$ can be represented by a convergent expansion
\begin{eqnarray}
\usc^{+}(r,\theta,\phi) = \frac{e^{ikr}}{kr}\sum_{l=0}^{\infty} \frac{F_l(\theta,\phi)}{(kr)^l} \qquad\qquad \text{for $r> R$}.\label{WilcoxExp}
\end{eqnarray}
This series is uniformly and absolutely convergent for $r>R$, $\theta$, and $\phi$. It can be differentiated  term by term with respect to $r$, $\theta$, and $\phi$ any number of times and the resulting series all converge absolutely and uniformly. Moreover the coefficients $F_l$ ($l \ge 1$) can be determined by the recursion formula,
\begin{eqnarray}
2 i l F_l(\theta,\phi) = l(l-1) F_{l-1}(\theta,\phi) + \Delta_{\Sph} F_{l-1}(\theta,\phi), \qquad l \geq 1. \label{AW-Recursive}
\end{eqnarray}

\end{Theorem}
\noindent Here, $r$, $\theta$, and $\phi$ are spherical coordinates and $\Delta_{\Sph}$ is the Laplace-Beltrami operator in the angular coordinates $\theta$ and $\phi$. See \cite{Bayliss01}.

 {Following an analogous procedure to the one employed in Section \ref{Section.ABC2D} for the 2D case, we use a truncated version of the series (\ref{WilcoxExp}) defined in $\Omega^{+}$ to match the solution in $\Omega^{-}$ through the interface conditions (\ref{BVPInterf4})-(\ref{BVPInterf5}).} This yields an approximation $u \approx \usc^{-}$ that is defined to be the solution of the following BVP in the region $\Omega^{-}$:
\begin{eqnarray}
&& \Delta u + k^2 u = 0, \quad\qquad  \text{in $\Omega^{-}$}, \label{BVP3D1} \\
&& u = - \uinc, \qquad\qquad\,  \text{on $\Gamma$,} \label{BVP3D2} \\
&& u(R,\theta,\phi)=\frac{e^{ikR}}{kR}\sum_{l=0}^{L-1} \frac{F_l(\theta,\phi)}{(kR)^l}\label{BVP3D3} \\
&&\partial_{r} u(R,\theta,\phi) = \frac{e^{i k R}}{k R} \sum_{l=0}^{L-1} \left( ik - \frac{l+1}{R} \right) \frac{F_{l}(\theta,\phi)}{(kR)^l}, \label{BVP3D4}\\
&& 2ilF_l(\theta,\phi) = l(l-1) F_{l-1}(\theta,\phi) + \Delta_{\Sph} F_{l-1}(\theta,\phi), \qquad l \geq 1. \label{BVP3D5}
\end{eqnarray}
The equations (\ref{BVP3D3})-(\ref{BVP3D5}) form the absorbing boundary condition with $L$ terms which we call {\it Wilcox farfield expansion absorbing boundary condition} and denote WFE$_L$. We also denote the BVP (\ref{BVP3D1})-(\ref{BVP3D5}) as WFE$_L$-BVP.
 {Contrary to the 2D case, there is only one family of unknown angular functions $F_l$ in this case. Hence, we only need the interface conditions (\ref{BVPInterf4})-(\ref{BVPInterf5}) plus the recurrence formula (\ref{AW-Recursive}) to define the new farfield expansion ABC at the artificial boundary $S$.}

The WFE-BVP (\ref{BVP3D1})-(\ref{BVP3D5}) can also be posed in weak form which is essential for the finite element methods. First we define the following (affine) spaces to deal with the Dirichlet boundary conditions (\ref{BVP3D2}),
\begin{eqnarray*}
&& H^{1}_{\rm \Gamma , Dir}(\Omega^{-}) = \left\{ \text{$u \in H^{1}(\Omega^{-})$ : $u = - \uinc$ on $\Gamma$} \right\}, \\
&& H^{1}_{\rm \Gamma,0}(\Omega^{-}) = \left\{ \text{$u \in H^{1}(\Omega^{-})$ : $u = 0$ on $\Gamma$} \right\}.
\end{eqnarray*}
We require the solution $(u,F_{0},F_{1},...,F_{L-1})$ to satisfy $u \in H^{1}_{\rm \Gamma , Dir}(\Omega^{-})$, $F_{l} \in H^{1}(S)$ for $l=0,...,L-2$, $F_{L-1} \in H^{0}(S)$, and
\begin{eqnarray}
&& - \la \nabla u , \nabla v \ra_{\Omega} + k^2 \la u , v \ra_{\Omega} + \frac{e^{ikR}}{kR}\sum_{l=0}^{L-1} \frac{ik - (l+1)/R}{(kR)^l} \la F_{l}, v \ra_{S} = 0, \quad \text{for all $v \in H^{1}_{\rm \Gamma,0}(\Omega^{-})$},  \label{BVP3Dvar1} \\
&& \la u , v_0 \ra_{S} = \frac{e^{ikR}}{kR} \sum_{l=0}^{L-1} \frac{1}{(kR)^l} \la F_{l} , v_0 \ra_{S},  \qquad \text{for all $v_0 \in H^{0}(S)$}, \label{BVP3Dvar2} \\
&& 2 i l \la F_{l} , v_l \ra_{S} = l(l-1) \la F_{l-1} , v_l \ra_{S} - \la \nabla_{S} F_{l-1}, \nabla_{S} v_l\ra_{S},  \qquad \text{for all $v_l \in H^{1}(S)$, $l \geq 1$}. \label{BVP3Dvar3}
\end{eqnarray}
where the symbol $\nabla_{S}$ represent the gradient in the geometry of the sphere $S$, and the functions $F_{l}$, originally defined on the unit-sphere, can be seen as defined on the sphere $S$ of radius $R$ by writing the argument as $\hat{x} = x/R$ where $x \in S$ and $R$ is fixed.

\section{Numerical method} \label{Section:NumMethd}

We start this section describing how to obtain a numerical approximation of the solution for the acoustic scattering of a plane wave from a circular shaped obstacle of radius $r=r_0$ using Karp farfield expansions as ABC. As discussed in previous sections, our approach consists of numerically solving the KDFE$_L$-BVP defined by  (\ref{BVPBd1})-(\ref{BVPBd2}) with the ABC given by (\ref{BVPBd3})-(\ref{BVPBd5}) supplemented by the recurrence formulas (\ref{Recurrence1})-(\ref{Recurrence2}).
For this particular circular scatterer, polar coordinates $(r,\theta)$ is the natural choice of coordinate system.
However, we will extend the discussion to more general obstacle shapes in generalized curvilinear coordinates in the next section. The numerical method chosen is based on a centered second  order finite difference. The number of grid points in the radial direction is $N$ and in the angular direction is $m+1.$
 Therefore, the step sizes in the radial and angular directions are $\Delta r =(R-r_0)/(N-1)$ and $\Delta\theta = 2\pi/m,$ respectively. Also, $r_i = (i-1)\Delta r,$ $\theta_j=(j-1)\Delta \theta$ and $u_{i,j}=u(r_i,\theta_j),$ where $i=1,\dots N$ and $j=1,\dots, m+1.$  Since the pairs $(r_i,\theta_1)$ and $(r_i,\theta_{m+1})$ represent the same physical point, $u_{i,1}=u_{i,m+1},$ for $i=1,\dots N.$
The discretization of the governing equations varies according to the location of the grid points. The areas of interest within the numerical domain and its boundaries are:  the obstacle boundary $\Gamma$, the interior of the domain $\Omega^{-},$ and the artificial boundary $S$.

At the obstacle boundary $u=-\uinc$ holds. Then, we start constructing the corresponding linear system $A{\bf U}={\bf b}$ by simply including the negative value of $\uinc$ at each boundary grid point in the forcing vector ${\bf b}$, and let the corresponding entry in the coefficient matrix $A$ equal unity. This results in an identity matrix of size $m\times m$ in the upper left-hand corner of the matrix $A$.

At the interior points $(r_i,\theta_j)$ ($i=2,\dots N-2$, $j=1,\dots m$) in $\Omega^{-},$ we discretize Helmholtz equation to obtain
\begin{eqnarray}
&&\alpha_i^{+}\uiMj + \alpha_i^{-}\uimj + \alpha_i\uij + \beta_i\uijM + \beta_i\uijm =0,\label{intequ}\\
&&\alpha_i^{+} = \tfrac{1}{\Delta r^2} + \tfrac{1}{(2\Delta r) r_i},\qquad
\alpha_i^{-} = \tfrac{1}{\Delta r^2} - \tfrac{1}{(2\Delta r) r_i},\nonumber\\
&&\alpha_i = k^2 - \tfrac{2}{\Delta r^2} - \tfrac{2}{\Delta \theta^2 r_i^2},\qquad
\beta_i = \tfrac{1}{\Delta \theta^2 r_i^2}.\nonumber
\end{eqnarray}
This discrete equation renders  $(N-3) m$ new rows to the sparse matrix  $A$ with a total of $5(N-3)m$ non-zero entries.

At the interior points $(r_{N-1},\theta_j)$ with $j=1,\dots m,$ we replace the 
$u_{N,j}$ term by $H_0(kR) \sum_{l=0}^{L-1} \frac{F_{l,j}}{(kR)^l} + H_1(kR)\sum_{l=0}^{L-1} \frac{G_{l,j}}{(kR)^l}$ from the farfield absorbing condition.  This leads to the following discrete equation:
\begin{eqnarray}
&&\alpha_{N-1}^{+}\sum_{l=0}^{L-1}\frac{H_0(kR)}{(kR)^l}  F_{l,j} + \alpha_{N-1}^{+}\sum_{l=0}^{L-1}\frac{H_1(kR)}{(kR)^l}  G_{l,j}\nonumber\\
 &&\quad+ \alpha_{N-1}^{-}u_{N-2,j} + \alpha_{N-1}u_{N-1,j}+ \beta_{N-1}u_{N-1,j+1} + \beta_{N-1}u_{N-1,j-1}=0.\label{interioreq}
\end{eqnarray}
This equation adds $ m$ new rows to the matrix $A$ with a total of $2(L+2)m$ nonzero entries.

Also at the artificial boundary points $(r_{N},\theta_j)$ with $j=1,\dots m,$ the discrete equation  (\ref{intequ}) is written as
 \begin{eqnarray}
&&\alpha_N^{+} u_{N+1,j}+ \alpha_N^{-} u_{N-1,j} + \alpha_N u_{N,j}+ \beta_{N} u_{N,j+1}+ \beta_N u_{N,j-1} =0.
\label{AB}
\end{eqnarray}

Now, consider the discretization of equation (\ref{BVPBd4}) using centered finite difference, {\it i.e.}
\begin{equation}
u_{N+1,j} = u_{N-1,j}  
-2\Delta r \sum_{l=0}^{L-1}A_l(kR)F_{l,j} - 
2\Delta r \sum_{l=0}^{L-1}B_l(kR)G_{l,j}, \label{1deriv}
\end{equation}
where 
\begin{equation}
A_l(kR)=\frac{kH_1(kR)}{(kR)^l} + \frac{klH_0(kR)}{(kR)^{l+1}}\quad\mbox{and}\quad 
B_l(kR)=\frac{k(l+1)H_1(kR)}{(kR)^{l+1}} - \frac{kH_0(kR)}{(kR)^{l}}. \nonumber
\end{equation}

Substitution of $u_{N+1,j}$, $u_{N,j+1}$, and $u_{N,j-1}$ into (\ref{AB}) using the previous expression and Karp's expansion, respectively, leads to the following set of $m$ equations for $j=1,\dots m$,
\begin{eqnarray}
&&(\alpha_N^{+} +\alpha_N^{-} )u_{N-1,j}  + \sum_{l=0}^{L-1}C_l(kR) F_{l,j} +\sum_{l=0}^{L-1}D_l(kR) G_{l,j}+ \label{uneq}\\
&&\beta_{N}\sum_{l=0}^{L-1}\frac{H_0(kR)}{(kR)^l}  F_{l,j+1} + \beta_{N}\sum_{l=0}^{L-1}\frac{H_1(kR)}{(kR)^l}  G_{l,j+1}
+ \beta_{N}\sum_{l=0}^{L-1}\frac{H_0(kR)}{(kR)^l}  F_{l,j-1} + \beta_{N}\sum_{l=0}^{L-1}\frac{H_1(kR)}{(kR)^l}  G_{l,j-1}=0,
\nonumber
\end{eqnarray}
where the coefficients $C_l$ and $D_l$ are given by
\begin{eqnarray}
C_l(kR) =  -\alpha_N^{+}2\Delta r k \,A_l(kR)
 + \alpha_N \frac{H_0(kR)}{(kR)^l} \quad \mbox{and}\quad
D_l(kR) =  -\alpha_N^{+} 2\Delta r k \,B_l(kR) + \alpha_N \frac{H_1(kR)}{(kR)^l}.\nonumber
\end{eqnarray}
The number of non-zero entries for these set of equation is $nz = (6L+1)m$.

Another set of $m$ equations is obtained from the discretization of the continuity condition on the second derivative combined with (\ref{1deriv}), and Karp farfield expansion,
\begin{eqnarray}
\frac{2}{\Delta r^2}u_{N-1,j} + \sum_{l=0}^{L-1}M_l(kR)F_{l,j} + \sum_{l=0}^{L-1}N_l(kR)G_{l,j}=0,
\end{eqnarray}
where 
\begin{eqnarray}
&&M_l(kR) =- \frac{2}{\Delta r} k\, A_l(kR) +
k^2\left(\frac{H_0(kR)}{(kR)^l} - (2l+1) \frac{H_1(kR)}{(kR)^{l+1}}- \frac{l(l+1)H_0(kR)}
{(kR)^{l+2}}\right) - \frac{2H_0(kR)}{\Delta r^2(kR)^l} \nonumber\\
&& N_l(kR) = - \frac{2}{\Delta r} k \,B_l(kR)+
k^2\left(-(2l+1)\frac{H_0(kR)}{(kR)^{l+1}} +\frac{H_1(kR)}{(kR)^{l}}- \frac{(l+1)(l+2)H_1(kR)}
{(kR)^{l+2}}\right) - \frac{2H_1(kR)}{\Delta r^2(kR)^l} \nonumber
\end{eqnarray}
The total number of nonzero entries for these equations is $(2L+1)m$.

Finally, each one of the recurrence formulas (\ref{Recurrence1})-(\ref{Recurrence2}) contribute with $(L-1)m$  new equations. They are given by
\begin{eqnarray}
 &&2lG_{l,j} = \left((l-1)^2 -\frac{2}{\Delta \theta^2}\right)F_{l-1,j} + \frac{1}{\Delta \theta^2}F_{l-1,j+1} + \frac{1}{\Delta \theta^2}F_{l-1,j-1}, \label{RecurrGF1}  \\
 &&2lF_{l,j} = \left(-l^2 +\frac{2}{\Delta \theta^2}\right)G_{l-1,j} - \frac{1}{\Delta \theta^2}G_{l-1,j+1} - \frac{1}{\Delta \theta^2}G_{l-1,j-1}  \label{RecurrGF2}
\end{eqnarray}
for $l=1,\dots L-1$ and $j=1,\dots m$. The number of nonzero entries is four for each $j$ and for each recurrence formula.

The above discrete equations are written as a linear system of equations $A{\bf U} = {\bf b}$. The matrix $A$ structure depends on how the unknown vector ${\bf U}$ is ordered. We chose ${\bf U}$ as follows:
\begin{eqnarray}
&&{\bf U}=
\Big[ \overbrace{u_{1,1} ...u_{1,m}}^{\text{at boundary }} ~ ~
\overbrace{u_{2,1} ...u_{2,m}...
u_{N-1,1}...u_{N-1,m}}^{\text{at interior grid points}} ~ ~\nonumber\\
&&\qquad\qquad\qquad\overbrace{F_{0,1}...F_{0,m}\, G_{0,1}...G_{0,m}...\,F_{L-1,1}...F_{L-1,m}\,G_{L-1,1}...G_{L-1,m}}^{\text{at artificial boundary}}
 \Big]^{T} \label{VectorU}
\end{eqnarray}
From the previous discrete equations (\ref{intequ})-(\ref{VectorU}), (\ref{uneq})-(\ref{RecurrGF2}), it can be seen that the matrix $A$ has dimension $(N-1+2L)m\times(N-1+2L)m$. Furthermore, adding the $m$ non-zero entries corresponding to the upper left-hand corner subdiagonal matrix of $A$ to the non-zero entries of the discrete equations  (\ref{intequ})-(\ref{VectorU}) and (\ref{uneq})-(\ref{RecurrGF2}), it can be shown that the non-zero entries of $A$ are $nz=(5N-16)m+18Lm$. 

A completely analogous work can be performed for the discretization of KSFE-BVPs.  However, the BVP defined by (\ref{BVPBd1})-(\ref{BVPBd2}) with the KSFE$_L$ condition (\ref{BVP2D3_Asymp})-(\ref{BVP2D5_Asymp}) has only  one family of unknown farfield angular coefficients $f_l(\theta)$ ($l=0,\dots L-1$). As a consequence, the matrix $A$ corresponding to its discrete equations has dimension $(N-1+L)m\times(N-1+L)m$. Moreover, it can be shown that its number of non-zero entries is $nz=(5N-13)m+8Lm$. For purpose of comparison, we also consider the discretization of the Dirichlet-to-Neumann  boundary value problem (DtN-BVP) derived by Keller and Givoli \cite{Keller01}.  For this BVP the matrix $A$, obtained by employing a second order centered finite difference method, has dimension $Nm\times Nm$ and its non-zero entries are $nz=(5N-8)m+m^2$.
A relevant feature of the matrices $A$ for the KDFE-BVP and KSFE-BVP is that they do not have full blocks as found in the case of DtN-BVP. In fact, the number of non-zero entries for the DtN-BVP matrix is $O(m^2)$ against $O(Lm)$ for the KDFE-BVP and KSFE-BVP matrices, respectively. Now, the number L of terms in the farfield expansion is always much smaller than m (nodes in the angular direction). As a consequence, the non-zeros of the matrices associated to the KDFE and KSFE boundary value problems are considerable less than those of the matrix corresponding to the DtN-BVP for the same problem.
This is a key property for the computational efficiency of the numerical technique proposed in this work. Furthermore, this is why higher order local ABCs are preferred over global exact ABCs such as DtN.

\section{Applications of farfield ABCs}\label{Problems}

\subsection{Scattering from a circular obstacle}\label{ScattCircle}
To illustrate the computational advantage of the exact farfield expansions ABCs over the DtN-ABC, we consider the acoustic scattering of a plane wave propagating along the positive $x$-axis from a circular obstacle of radius $r_0=1$. We place the artificial boundary at $R=2$ and select a frequency $k=2\pi$  for the incident wave.  Then, we apply the centered finite difference scheme described in Section \ref{Section:NumMethd} for the KDFE-BVP. For purpose of comparison, we also apply it with its respective modifications to KSFE-BVP and DtN-BVP.  The points per wavelength in each case is $PPW=20$. The number of terms employed for KDFE$_L$ is $L=3$ and for the KSFE$_L$ is $L=8$. These choices of $L$ made possible that the three numerical solutions approximate the exact solution at the artificial boundary with about the same relative error of $3.8 \times 10^{-3}$ in the L$^2$-norm. In Fig.~\ref{SpyGraph}, the structure of their respective matrices are depicted. Although the matrix corresponding to the DtN-ABC has the smallest dimension, it has more than one and a half times as many non-zero entries as the farfield expansions ABCs. As the number of point per wavelength increases, this difference is even bigger since the number of points for the DtN-ABC is $O(m^2)$ while for the farfield ABCs is only $O(Lm)$. 

\begin{figure}[h!]
\begin{center}
\begin{subfigure}{0.32\textwidth}
\includegraphics[width=\textwidth]{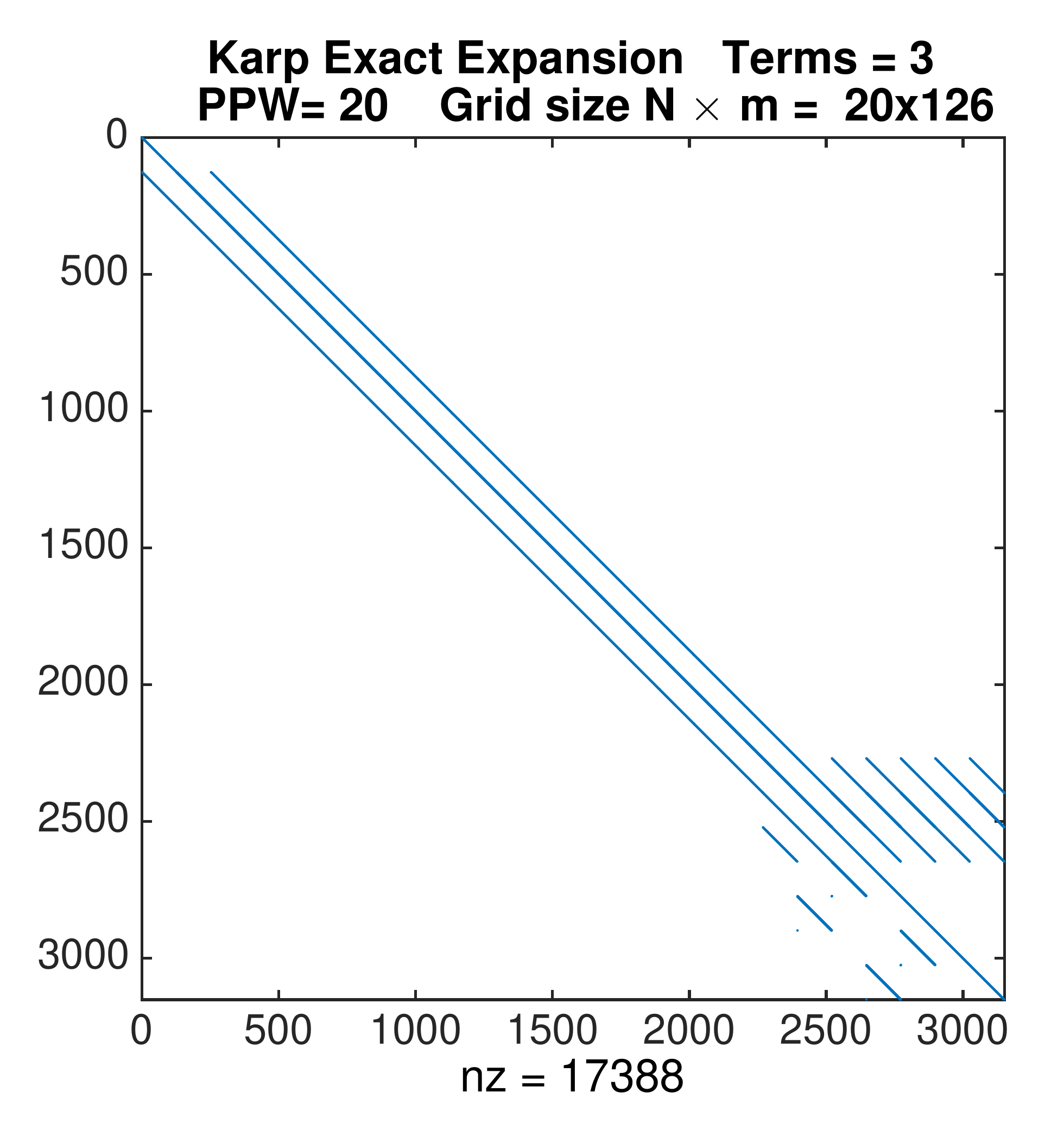}
\caption{}
\label{SpyKDFE20}
\end{subfigure}
\begin{subfigure}{0.32\textwidth}
\includegraphics[width=\textwidth]{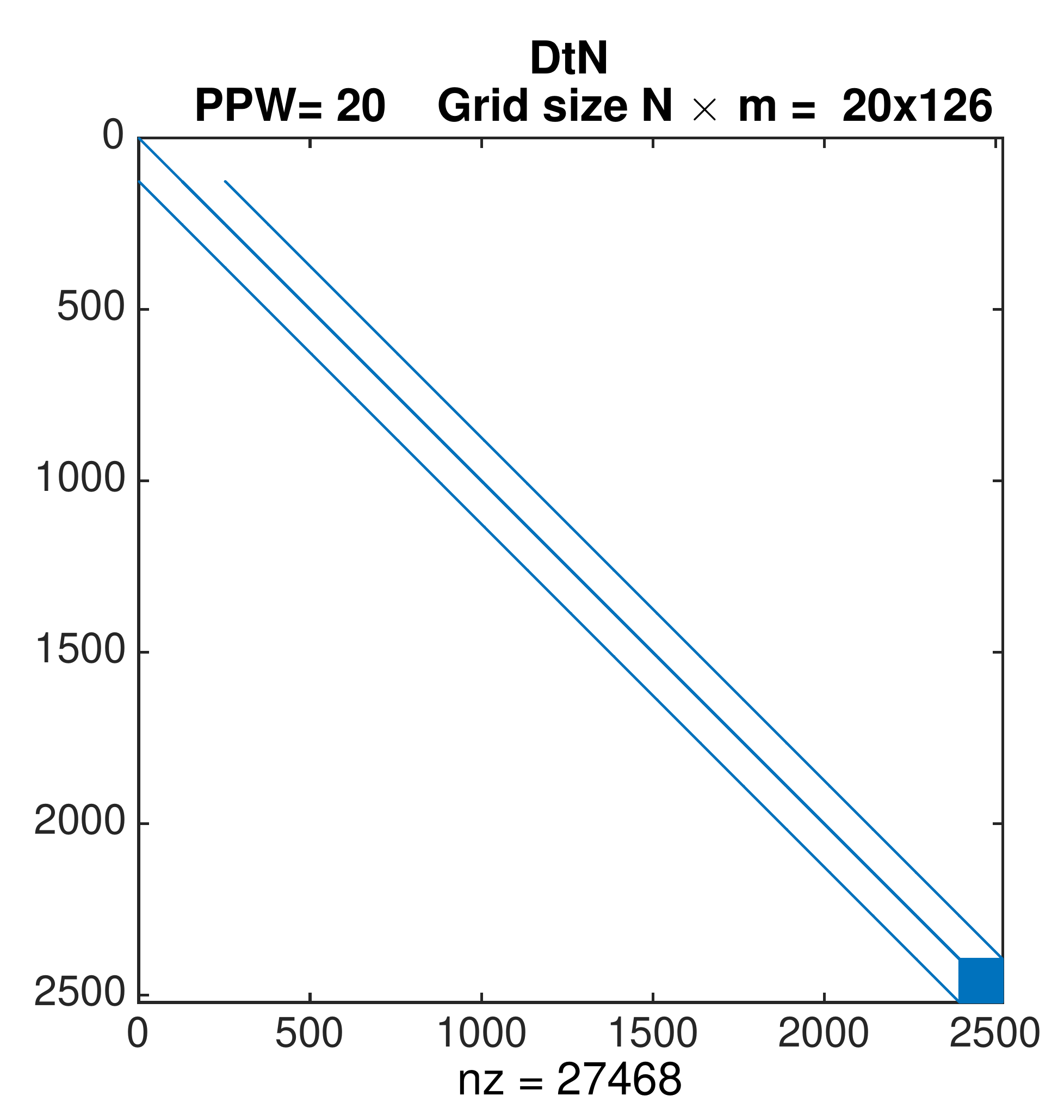} 
\caption{}
\label{SpyDtN20}
\end{subfigure}
\begin{subfigure}{0.32\textwidth}
\includegraphics[width=\textwidth]{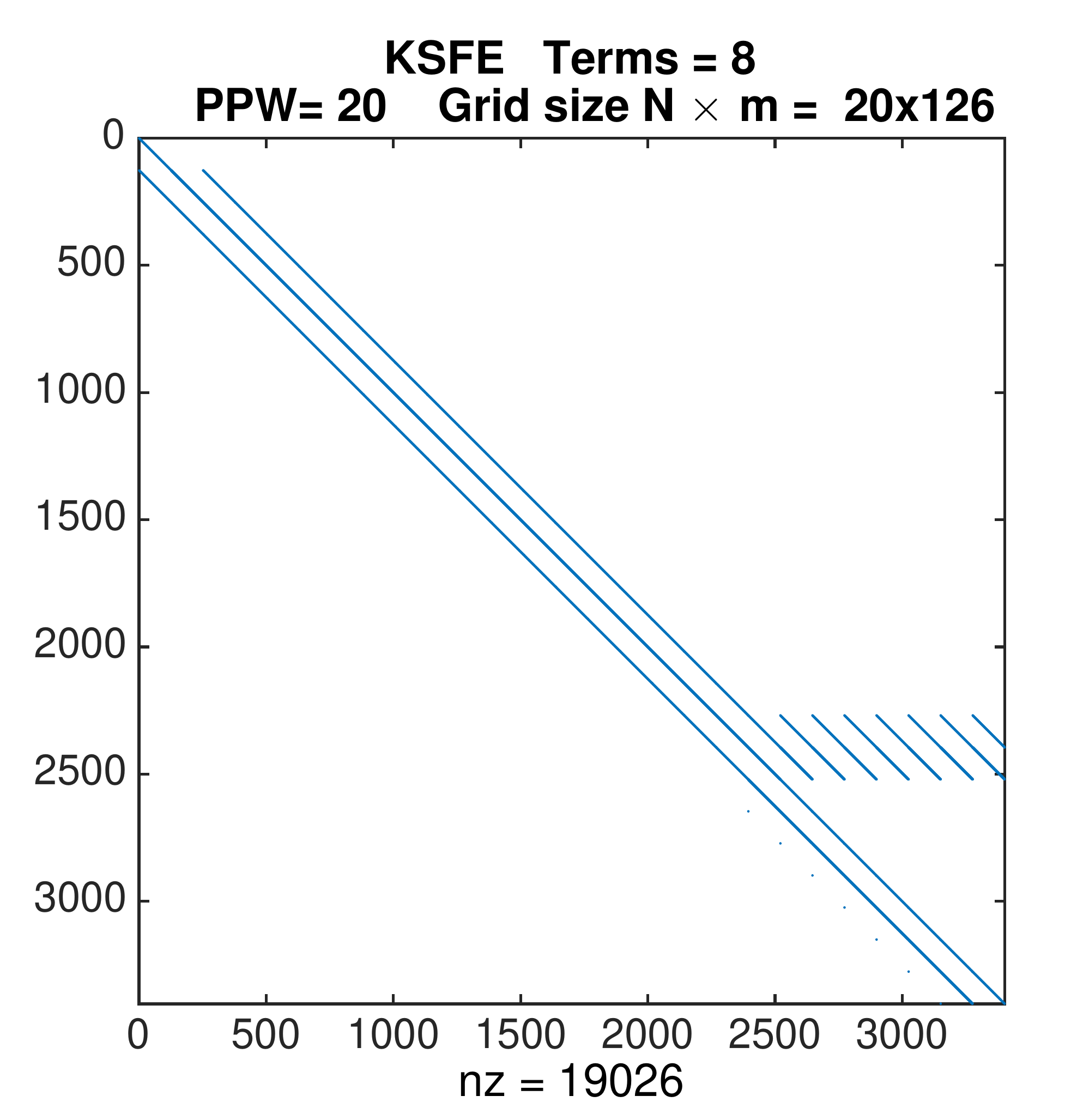}
\caption{}
\label{SpyKSFE20}
\end{subfigure}
\caption{Comparison of the matrix structure for: a) KDFE$_3$, b) DtN, and c) KSFE$_8$ with $R=2$ and $PPW =20$.}
\label{SpyGraph}
\end{center}
\end{figure}

It is timely to comment on the numerical difficulties that can be faced when solving three-dimensional problems modeled by the farfield ABCs. Using our finite difference technique will lead to sparse but very large matrices at the discrete level. Therefore, a direct solver may not be a feasible choice as the mesh is refined. Iterative methods become an imperative choice. Among such methods are Krylov subspace iterative methods, multigrid and domain decomposition methods.
However, their applications to the resulting sparse matrices experience difficulties because
these matrices are known to be non-Hermitian and poorly conditioned. Efforts have been made to develop good preconditioners and parallelizable methods tailored to these wave scattering problems modeled by the Helmholtz equation \cite{Kechroud-Soulaimani2004, Erlanga2008}.  We intend to explore some of these new techniques along with the application of the farfield absorbing boundary conditions to complex 3D problems in future work. 

\subsection{Scattering from a spherical obstacle. Axisymmetric case}
\label{ScattSphere}

In this section, we formulate the BVP corresponding to scattering from a spherical obstacle of an incident plane wave $\uinc = e^{ikz}$ propagating along the positive z-axis. The mathematical model including our novel Wilcox farfield ABC (WFE-ABC) consists of the BVP (\ref{BVP3D1})-(\ref{BVP3D5}) in spherical coordinates $(r,\theta,\phi)$.
This problem is axisymmetric about the $z-$axis. Therefore, the governing Helmholtz equation for the approximation $u$ of the scattered field $\usc$ is independent of the angle $\phi$. As a consequence, it reduces to
\begin{equation}
\frac{\partial ^2u}{\partial r^2} + \frac{2}{r}\frac{\partial u}{\partial r} + 
\frac{1}{r^2\sin \theta}\frac{\partial}{\partial\theta}\left( \sin\theta\frac{\partial u}{\partial \theta} \right) + k^2 u =0, \quad \mbox{in}\,\, \Omega^{-}. \label{Helmholtzsphsymm}
\end{equation}
Obviously, this equation is singular at the poles when $\theta = 0,\pi .$ However, there is not such singularity at these angular values for Helmholtz equation in cartesian coordinates. The singularity arises by the introduction of spherical coordinates. It can be shown \cite{SadikuBook} that equation (\ref{Helmholtzsphsymm}) reduces to $\frac{\partial u}{\partial \theta}(r,\theta)=0,$ when $\theta=0,\pi$. The angular coefficients $F_l$ of the Wilcox farfield expansion are also independent of $\phi$. As in the two-dimensional case, we employ
 a second order centered finite difference scheme as our numerical method to obtain the approximate solution to this scattering problem. 
 Due to the analogy between the KSFE and the WFE absorbing boundary conditions for this axisymmetric case, the discretization of the equations and the structure of the matrix obtained after applying a centered finite difference approximation to the equations defining this BVP
 are similar to those of KSFE-BVP.
 In Section \ref{orderconverg3D}, numerical results for this problem are presented.

\subsection{ Radiation and scattering from complexly shaped obstacles in two-dimensions}\label{ScattComplex}
Since most real applications deal with obstacle of arbitrary shape, in this section, we consider scattering problems for arbitrary shaped scatterers using the farfield absorbing boundary conditions. In order to do this,  we introduce
generalized curvilinear coordinates such that the physical
scatterer boundaries correspond to coordinate lines. These type of
coordinates, called boundary conforming coordinates
\cite{Steinberg}, are generated by invertible transformations $
T:\thinspace \mathcal{D^{\prime}}\rightarrow \mathcal{D}$, from a
rectangular computational domain ${\cal D^{\prime}}$ with
coordinates $(\xi,\eta)$ to the physical domain ${\cal D}$ with
coordinates $(x,y)=(x(\xi,\eta),y(\xi,\eta))$. A common practice in elliptic grid generation is
to implicitly define the transformation $T$ as the numerical
solution to a Dirichlet boundary value problem governed by a system of quasi-linear elliptic equations
for the physical coordinates $x$ and
$y.$ Following this approach, the authors Acosta and Villamizar \cite{JCP2010} introduced the {\it elliptic-polar grids}
as the solution to the following quasi-linear  elliptic system of equations:
\begin{eqnarray}
&& \alpha x_{\xi \xi }-2\beta x_{\xi \eta}+\gamma x_{\eta \eta } +
\frac{1}{2} \alpha_{\xi}x_{\xi} + \frac{1}{2}
\gamma_{\eta}x_{\eta} = 0, \label{Elliptic3} \\
&& \alpha y_{\xi \xi }-2\beta y_{\xi \eta}+\gamma y_{\eta \eta } +
\frac{1}{2} \alpha_{\xi}y_{\xi} + \frac{1}{2} \gamma_{\eta}y_{\eta} = 0.
\label{Elliptic4}
\end{eqnarray}
The symbols $\alpha$, $\beta$, and $\gamma$,
represent the scale metric factors of the
coordinates transformation $T$, respectively. These are defined as
$$\alpha =x_{\eta }^{2}+y_{\eta }^{2}, \qquad
\beta =x_{\xi }x_{\eta }+y_{\xi }y_{\eta }, \qquad \gamma=x_{\xi
}^{2}+y_{\xi }^{2}.
$$
In this work, we adopt the elliptic-polar coordinates in the presence of complexly shaped obstacles. Before we attempt a numerical solution to our BVP with the farfield expansions ABCs in 
these coordinates, we express the governing equations in terms of them. For instance, the two-dimensional 
Helmholtz equation transforms into
\begin{eqnarray} \frac{1}{J^2} \Bigg[\alpha u_{\xi \xi
}-2\beta u_{\xi \eta}+\gamma u_{\eta \eta } +
\frac{1}{2}
\Big(\alpha_{\xi}\,u_{\xi}+\gamma_{\eta}\,u_{\eta}\Big)\Bigg]
+ k^2 u = 0, \label{HelmElliptic}
\end{eqnarray}
where the symbol $J$ corresponds to the jacobian of the transformation $T$. Once the farfield expansions ABC equations are also expressed in terms of elliptic-polar coordinates, we transform all of these continuous equations into discrete ones using centered second order finite difference schemes. This process is described in detail in \cite{JCP2010}. Then, the corresponding linear system is derived in much the same way as we did above for polar coordinates.  Numerical results for several complexly shaped obstacles are discussed in Section \ref{Section.Numerics2D}.

\section{Farfield Pattern definition and its accurate numerical computation}
\label{NumericalFFP}
In scattering problems, an important property  to be determined is the scattered field far from the obstacles. The geometry and physical properties of the scatterers are closely related to it.
In Section 4.2.1 of \cite{MartinBook}, Martin defines the farfield pattern (FFP) as the angular function present in the dominant term of the asymptotic expansions for  the scattered wave when $r\rightarrow\infty$. For instance in 2D, the farfield pattern is the coefficient $f_0(\theta)$ of KSFE,
\begin{equation}
u(r,\theta)=
\frac{e^{ikr}}{(kr)^{1/2}}f_0(\theta)+O\left(1/(kr)^{3/2}\right). \label{KSFE}
\end{equation}
Following Bruno and Hyde \cite{McKaySIAM}, we now describe how the FFP can be efficiently 
calculated from the  approximation of the scattered wave at 
the artificial boundary.

If $r>R$, where $R$ is the
radius of the artificial circular boundary enclosing the
obstacle, then, the scattered wave can be represented as the following complex Fourier series,
\begin{equation}
u(r,\theta)= \sum_{q=-\infty}^{\infty} c_q(r) e^{i
q\theta}=\sum_{q=-\infty}^{\infty} b_q H_q^{(1)}(kr)e^{i q\theta}, \qquad \text{where} \quad b_q=\dfrac{c_q(r)}{H_q^{(1)}(kr)}.
\label{FourierSeries}
\end{equation}
Using the asymptotic
expansion of the Hankel function $H_q^{(1)}(kr)$ when
$r\rightarrow\infty$, equation
(\ref{FourierSeries}) transforms into
\begin{equation}
u(r,\theta) = 
\frac{e^{ikr}}{(kr)^{1/2}}\left(\sqrt{\frac{2}{\pi
}}e^{-i\pi/4}\sum_{q=-\infty}^{\infty} b_q (-i)^q e^{i
q\theta}\right) + \mathcal{O}\left(1/{(kr)^{3/2}}\right). \label{Asymp}
\end{equation}
  By comparing (\ref{Asymp}) with (\ref{KSFE}) the
following expression for $f_0(\theta)$ is derived
\begin{equation}
f_0(\theta)=\sqrt{\frac{2}{\pi}}e^{-i\pi/4}\sum_{q=-\infty}^{\infty} b_q (-i)^q e^{i q\theta}.
\end{equation}
Thus, the FFP can be determined once the coefficients $b_q$
have been calculated. But as pointed out above, the coefficients 
$b_q$ can be determined from the coefficients $c_q(r)$ for $r$ fixed. 
Likewise, approximated
values of $c_q(R)$ can be obtained from the scattered field approximation at the 
artificial boundary
$r=R$, {\it i.e.,} $u_{N,j}$ for $j=1,\dots m$. 
 In fact as stated by Kress \cite{Kress},
 approximations ${\hat
c}_q$ to the coefficients $c_q(R)$, at the fictitious infinite
boundary can be obtained by considering the discrete Fourier
transform vector ${\hat c}_q$ ($q=-m/2 ,\dots m/2 -1$) of the vector $u_{N,j}$, 
 interpolating the points
$\left(\theta_j,u_{N,j}\right)$ for $j=1,\dots
m\, (m\,\, \mbox{even})$. 
 More precisely,
\begin{equation}
{\hat c}_q=\frac{1}{m}\sum_{j=1}^{m-1} u_{N,j}
e^{-i q\theta_j},\qquad\mbox{ for}\quad q=-m/2,\dots m/2-1.
\label{FCoeffs}
\end{equation}
These finite series can be directly evaluated,  or a FFT algorithm
can be used to compute them. The importance of the above
derivation is that  a semi-analytical formula
\begin{equation}
f_0(\theta)=\sqrt{\frac{2}{\pi
}}e^{-i\pi/4}\sum_{q=-m/2}^{m/2-1} {\hat b}_q (-i)^q e^{i
q\theta}, \label{AnalyticalSCS}
\end{equation}
approximating the FFP for arbitrary shaped obstacles, can be
obtained from the numerical approximation of the scattered far
field, where ${\hat b}_q={\hat c}_q/H^{(1)}_q(kR)$. This
formula is extremely accurate as shown in \cite{McKaySIAM}. The
error in the approximation of $f_0(\theta)$ using
(\ref{AnalyticalSCS}) depends almost entirely upon the error made
in the approximation of the coefficients ${\hat b}_q$.


\section{Numerical Results} \label{Section.Numerics2D}
In this Section,   {we present numerical evidences of the advantages of using the exact farfield expansions  ABCs, when dealing with acoustic scattering and radiating problems, compared with other commonly used ABCs.}
First, we numerically solve bounded problems with farfield expansions ABCs as defined in Sections \ref{Section.ABC2D}-\ref{Section.ABC3D}. Then, we show that these numerical solutions indeed converge to the exact solutions of the original unbounded BVPs. As described in Section \ref{Section:NumMethd}, the numerical method employed consists of familiar second order centered finite difference discretizations for Helmholtz equation in polar, spherical, and generalized curvilinear coordinates. This numerical method is completed with the discrete equations of the farfield expansions ABCs on the artificial boundary $S$.  
Our numerical results contain two sources of error. The first one is the error introduced by the finite difference scheme employed to discretize the Helmholtz equation in the computational domain $\Omega^{-}$. This error can be diminished by refining the finite difference mesh as we increase the number of points per wavelength. The second source of error is due to the truncation of the farfield expansion series. This error can be diminished by increasing the number $L$ of terms in the absorbing conditions KSFE$_L$, KDFE$_L$, or WFE$_L$. For example, if a finite difference scheme for a two-dimensional problem in polar coordinates leads to a second order convergence, then the order of the total error introduced by combining the finite difference scheme with the proposed absorbing boundary conditions is given by
\begin{align}
\text{error} = O \left( h^2 \right) + O\left( (kR)^{-L} \right), \label{TotalError}
\end{align}
where $h = r_{0} \Delta \theta = \Delta r$ is the mesh refinement parameter and $L$ is the number of terms in the farfield expansions absorbing boundary conditions. Then for the total error to exhibit second order convergence with mesh refinement, it is necessary to choose $L = O \left( \log( 1/h) \right)$. Therefore, there is no need of large increments of  $L$ to improve the order of convergence. In practice, we can expect that a moderately large (but fixed) choice of $L$ will be sufficient to reveal the order of convergence of the finite difference method for a reasonable range of mesh refinements. We construct our numerical experiments, reported later in this section, according to this fact.

First, we present numerical results for the scattering of a plane wave $\uinc= e^{ikx}$ from a circular obstacle in 2D. As a reference point, we also display the results from the use of the DtN nonreflecting condition \cite{Keller01,Givoli05,Grote-Keller01,GivoliReview}. Since this latter condition is considered \textit{exact}, it serves as a reference point to gauge the error introduced by the finite difference scheme alone. For all ABCs, we use a second order centered finite difference scheme in the interior of the domain $\Omega^{-}.$  In our experiments for the circular scatterer, we are able to obtain second order convergence of the numerical solution to the exact solution. In fact, we found that the numerical solution obtained from KDFE$_L$ and KSFE$_{L}$, for appropriate number of terms $L$, are comparable to the approximation obtained from the DtN absorbing boundary condition. However, the advantage over the DtN-ABC formulation is that the farfield expansions ABCs are local while the former is not. 

Secondly, we numerically solve the scattering from complexly shaped scatterers, using the exact farfield expansions ABCs. As a result, the farfield patterns (FFP) for several obstacles of arbitrary shape are obtained. Then we present our results, with the new ABCs, for an exterior radiating problem obtained from two sources conveniently located inside a domain bounded by complexly shaped curves.   {For these specials radiating problems, it is possible to obtain analytical solutions.} Then, by comparing the numerical approximations and the exact solutions, we determine the order of convergence for several non-separable geometries, as we do in the circular case.

Finally, results for a spherical scatterer are presented.  The numerical method is analogous to the one employed in the two-dimensional case for the KDFE$_L$-BVP: a centered finite difference of second order in  $\Omega^{-}$ and WFE-ABC on the artificial boundary. Again, a second order convergence is reached by using few terms in the WFE farfield expansions. 

\subsection{Scattering from a circular obstacle. Comparison against exact solution and order of convergence}
\label{orderconverg2D}
First, we point out that approximated solutions of the scattering problem obtained for the BVP corresponding to KSFE$_1$ are identical to the numerical solutions obtained for the BVP corresponding to BGT$_1$. This is a numerical evidence of the equivalence between these two problems, as proved in Theorem \ref{Equiv1}. 

Another important result is the convergence of the numerical solutions of KDFE$_L$ and KSFE$_L$ boundary value problems to the exact solution as the number  $L$ of terms  in the farfield expansion is increased for a sufficiently small $h$. In particular, this is shown in Fig. {\ref{fig:ScattFields}}. However, the KSFE$_L$ numerical solutions only converge asymptotically as $kR \to \infty$ when $L$ is fixed. Indeed when $kR$ is fixed, {\it e.g.} $kR=\pi/2$, and $L$ grows then the numerical solution unavoidably diverges as explained at the end of subsection \ref{Section.ABC2DAsymp}. This fact is also discussed in more detail in the Conclusions  Section \ref{Section.Conclusions}. Actually, Fig. {\ref{fig:ScattFields}} (left panel) shows the appearance of unphysical oscillations in the farfield pattern for KSFE$_{9}$ which become larger as $L$ increases.

\begin{figure}[h!]
\begin{subfigure}{0.5\textwidth}
\includegraphics[width=\textwidth]{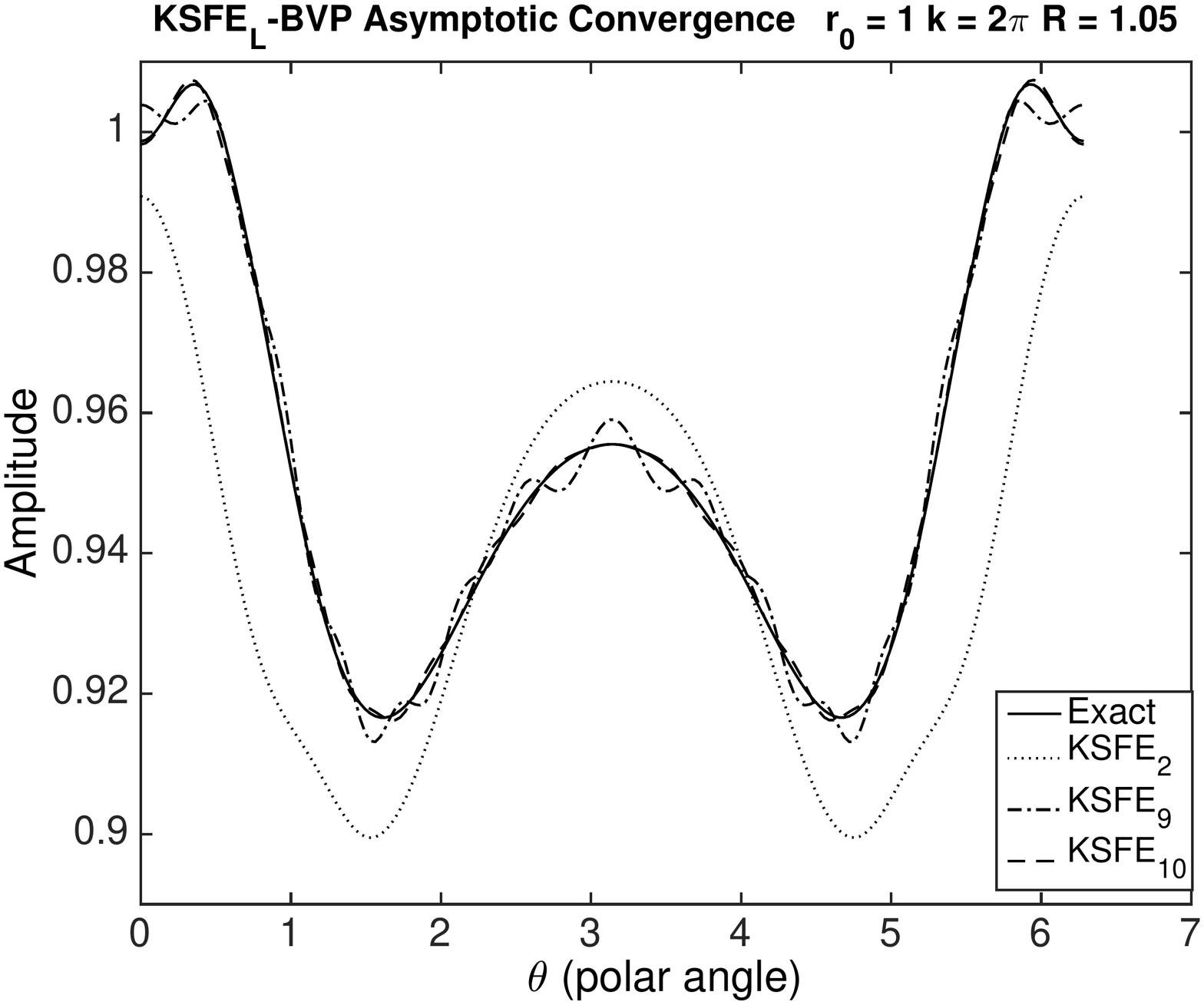} 
\end{subfigure}
\begin{subfigure}{0.5\textwidth}
\includegraphics[width=\textwidth]{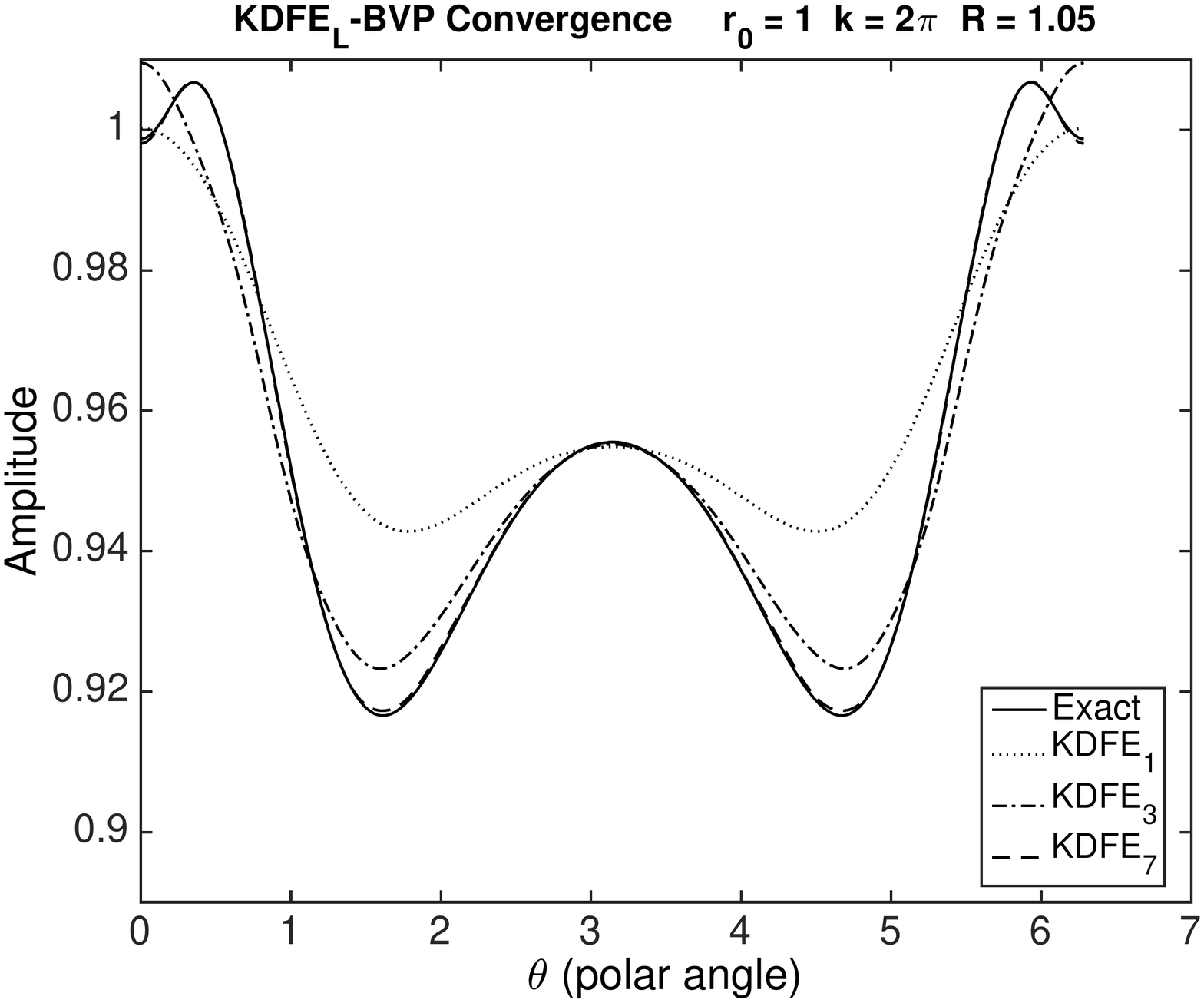}
\end{subfigure}
\caption{Convergence of solutions of KSFE$_L$- and KDFE$_L$-BVPs ($h$ fixed and $L$ increasing) to the exact solution of scattering of a plane wave from a circular scatterer of radius $r_0=1$ 
along the artificial boundary with radius $R=1.05.$}
\label{fig:ScattFields}
\end{figure}

The relevant data used in these problems is the following: wavenumber $k=2\pi$, radius of the circular obstacle is $r_0 =1$, and radius of artificial boundary $R=1.05$. We define the grid such that the number of points per wavelength in all experiments is PPW = 30 in the angular direction and $N=21$ points in the radial direction. This is an extreme problem where the artificial boundary radius has been chosen almost equal to the radius of the circular scatterer. So, the domain of computation is very small. Even in this extreme situation, it is observed how well the numerical solution of KDFE$_7$-BVP approximates the exact solution at the artificial boundary with a $L_2$-norm relative error equal to $3.44 \times 10^{-4}$ with only seven terms in the farfield expansion. Similarly, the numerical solution of KSFE$_{11}$-BVP at the artificial boundary also approximates the exact solution with a $L_2$-norm relative error equal to $3.73 \times 10^{-4}$ with eleven terms in the expansion. This illustrates the slower convergence of the numerical solutions of KSFE$_L$-BVP when compared with the sequence of solutions obtained from KDFE$_L$-BVP.

\begin{figure}[!h]
\begin{subfigure}{0.5\textwidth}
\includegraphics[width=\textwidth]{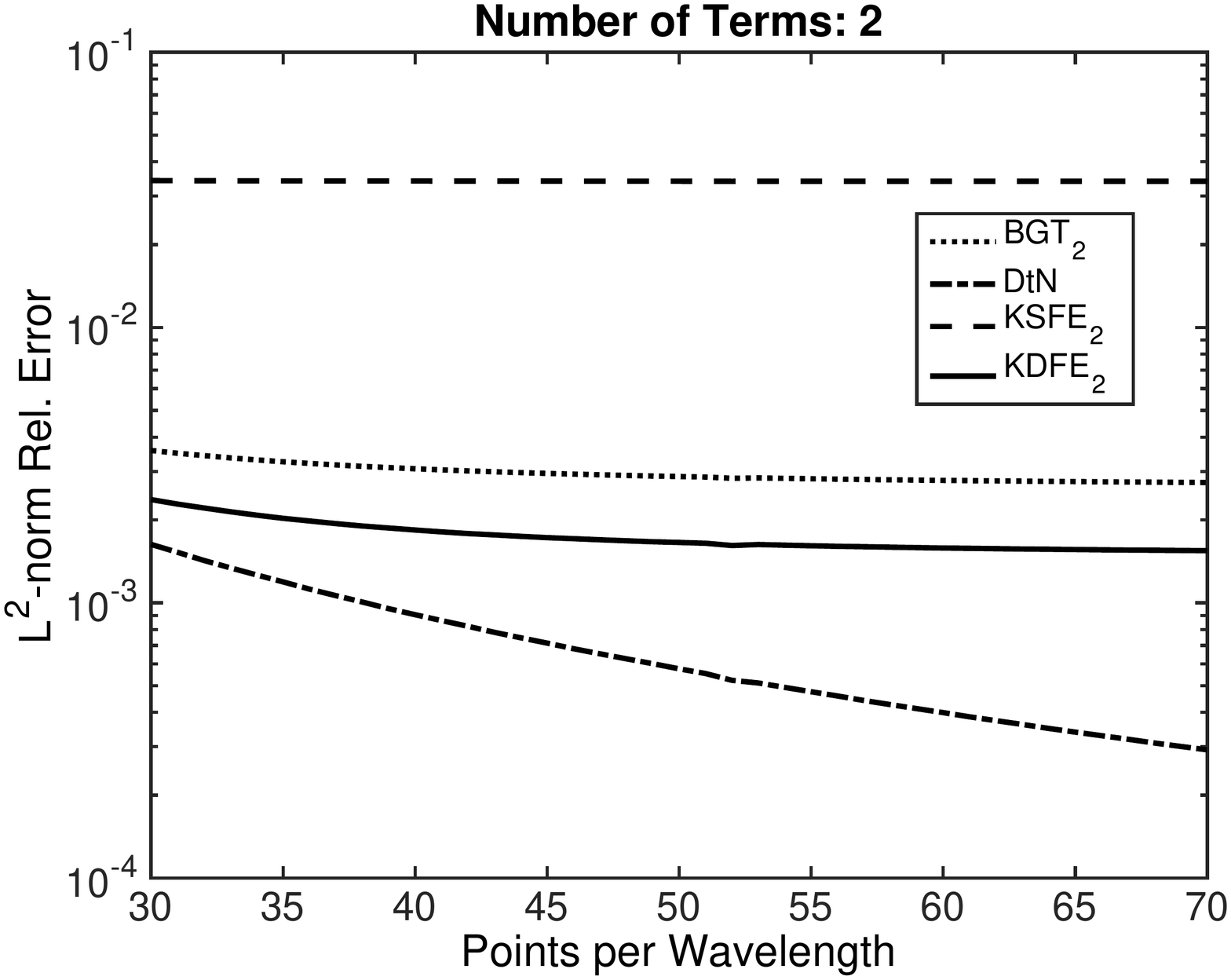} 
\end{subfigure}
\begin{subfigure}{0.5\textwidth}
\includegraphics[width=\textwidth]{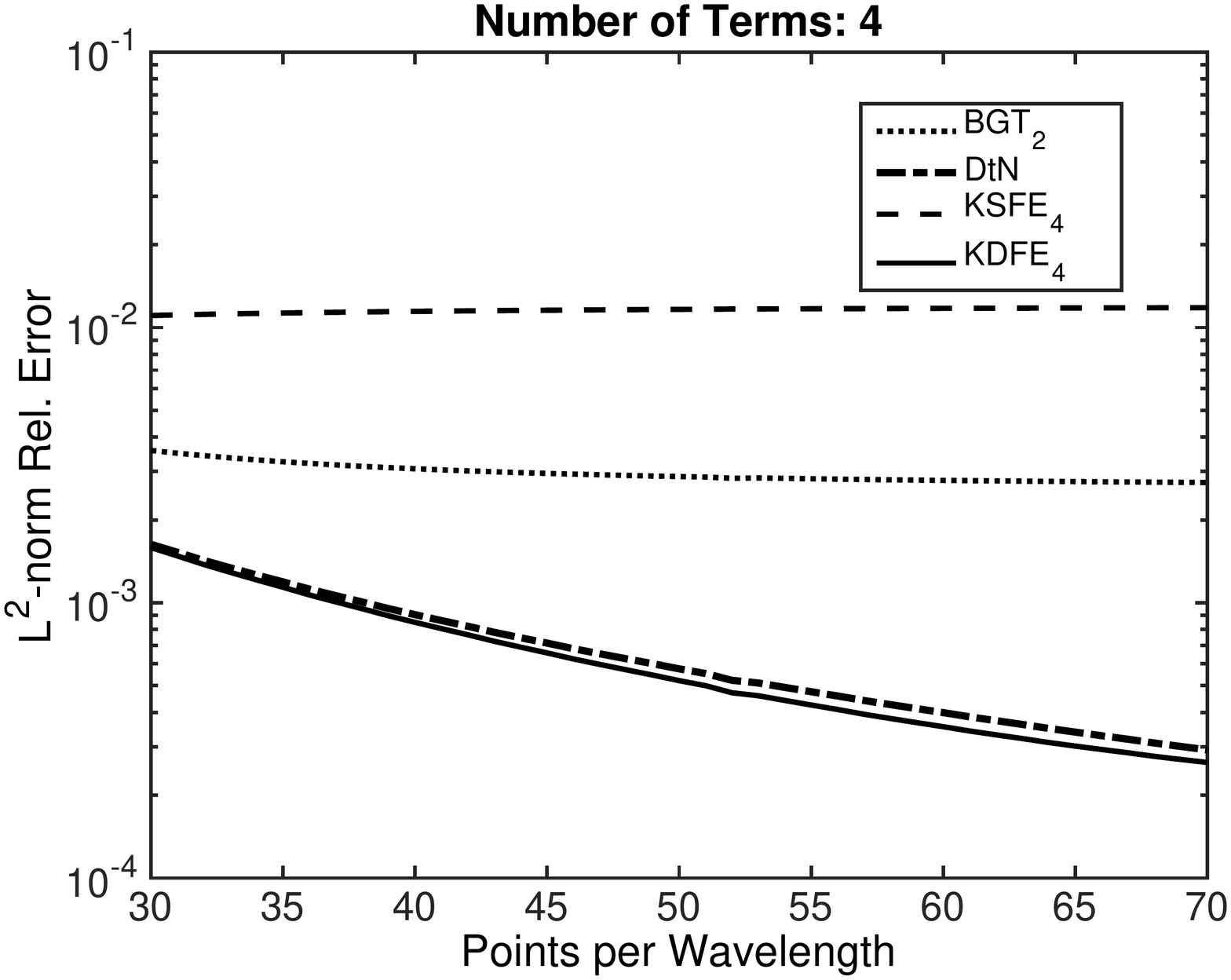}
\end{subfigure}
\begin{subfigure}{0.5\textwidth}
\includegraphics[width=\textwidth]{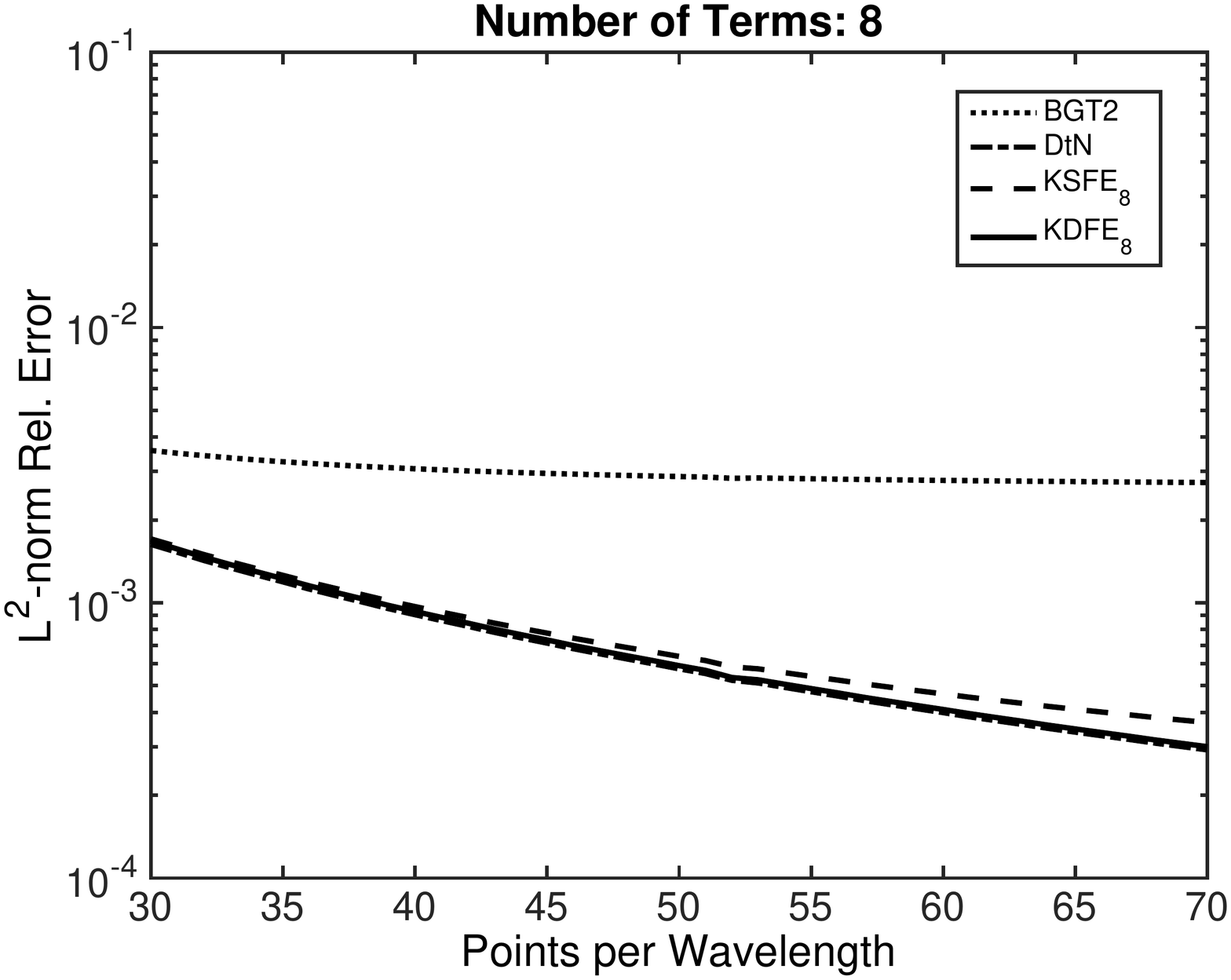}
\end{subfigure}
\begin{subfigure}{0.5\textwidth}
\includegraphics[width=\textwidth]{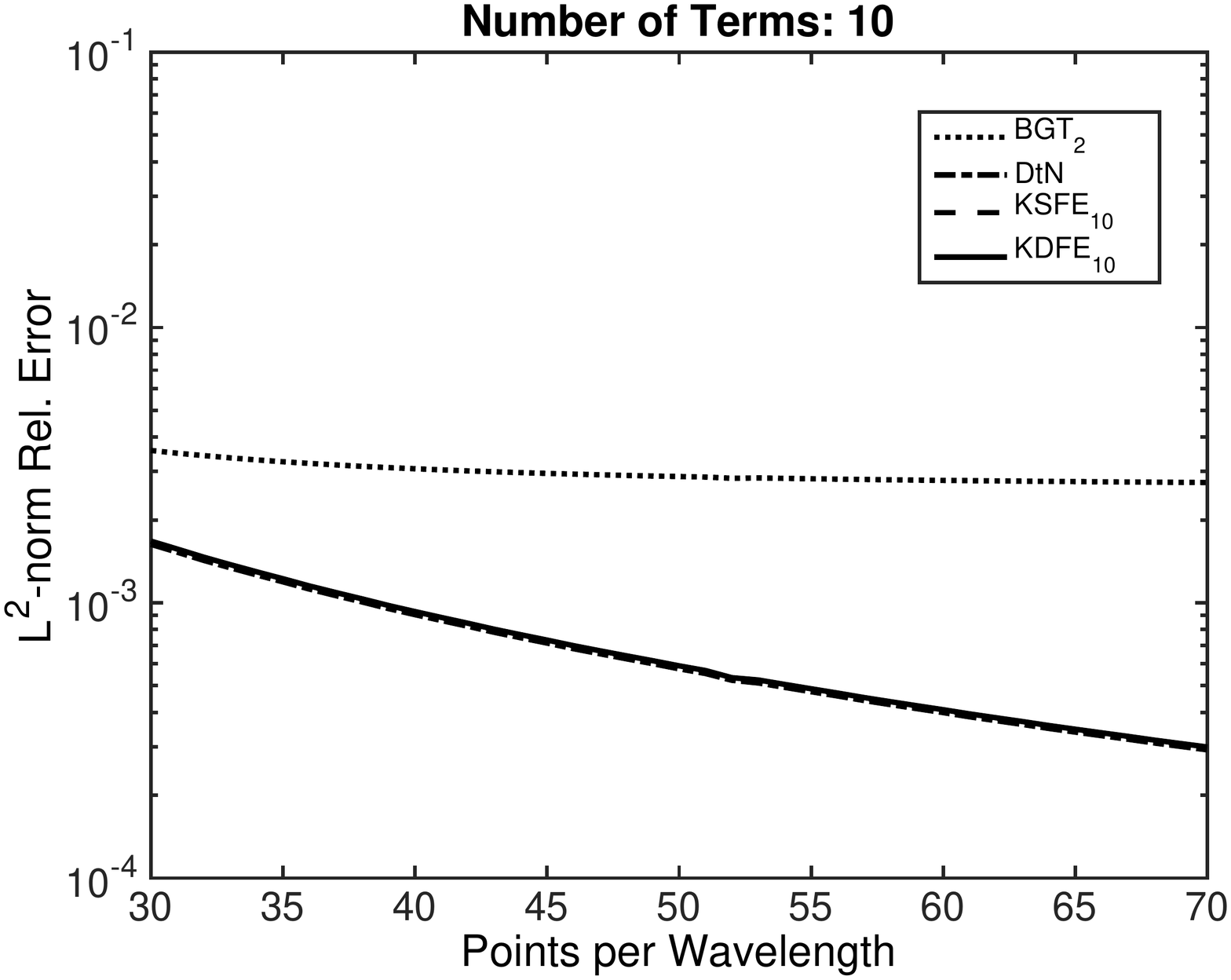}
\end{subfigure}
\caption{Comparison of L$_2$-norm relative error of the Farfield Pattern among DtN, BGT$_2$, KSFE$_L$, and KDFE$_L$ for $L=2,4,8,10$. The data in use is $r_0=1$, $R=2$, and $k=2\pi$.} 
\label{fig:ConvergencePPW}
\end{figure}

In our next set of numerical experiments, we analyze the performance of the second order finite difference method for the scattering from a circular scatterer using the following ABC: BGT$_2$, DtN, KSFE$_L$, and KDFE$_L$ ($L= 2,4,8,10$). By comparing the numerical farfield pattern (FFP) with the one obtained from the exact solution, we obtain the L$^2$-norm relative error. The formula employed to compute the FFP, for all types of ABCs from the numerical solution of the scattered field at the artificial boundary, is the formula (\ref{AnalyticalSCS}) described in Section \ref{Section:NumMethd}. 
The results of these experiments are illustrated in Fig. \ref{fig:ConvergencePPW}. The common data in these numerical 
simulations is the following: frequency $k=2\pi$, radius of the circular obstacle is $r_0 =1$, and radius of artificial boundary $R=2$. In all our experiments, the error reported  is the $L^2$-norm relative error. The grid is systematically refined, as $L$ is kept fixed,  to discover the rate of convergence. For $L=2$ (top left corner subgraph), it is observed that the rate of convergence for three of the four types of ABC is close to zero, while the approximation to the exact solution of the numerical solution corresponding to DtN improves as the grid is refined. The subgraph at the top right corner reveals that the numerical solution of KDFE$_4$-BVP has almost the same rate of convergence than the one corresponding to DtN-BVP. From the subgraph in the lower row left corner, we conclude that the numerical solution of KSFE$_8$-BVP also converges at almost the same rate as the one for KDFE$_4$ and DtN boundary value problems. Finally for ten terms in both farfield expansions, the rate of convergence for the ABC: KSFE, KDFE, and DtN is basically the same. 

The previous discussion illustrated in Fig. \ref{fig:ConvergencePPW} is appropriately summarized by a single graph depicted in Fig. \ref{fig:OrderConv}. This figure clearly shows the second order convergence of the three methods using: DtN, KDFE$_L$ ($L\ge 5$) and KSFE$_{L}$, ($L\ge 9$) while BGT$_2$-BVP order of convergence is around $3.8 \times 10^{-1}$. The set of grids employed to obtain Fig. \ref{fig:OrderConv} consist of PPW = 30, 40, 50, 60, and 70, respectively. 
As a particular case of the quadratic convergence  of KDFE$_L$-BVP ($L\ge 5$), we show the convergence of the numerical solution of KDFE$_5$-BVP in Table \ref{table:1}. The grids are ordered from less to more refine. Furthermore,
Fig. \ref{LeastSqKDFE5} shows the line obtained from the least squares approximation of the orders between progressively finer grids. The slope of this line is 1.99948, which confirms the quadratic order of convergence for the numerical solution of KFDE$_5$-BVP using the technique proposed in this work.

\begin{figure}[h]
\begin{center}
\includegraphics[width=0.6\linewidth]{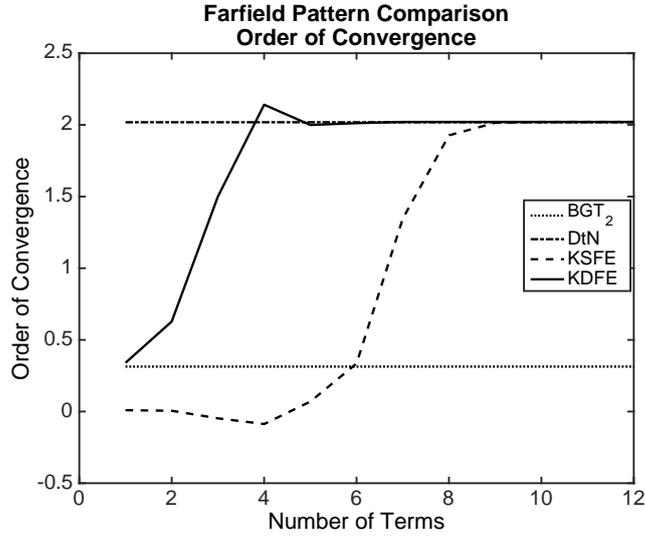} 
\end{center}
\caption{Comparison of order of convergence of FFP approximation for various ABCs versus the number of terms in the farfield expansion. The data in use is $r_0=1$, $k=2\pi$, $R=2$, and PPW = 30, 40, 50, 60, 70.}
\label{fig:OrderConv}
\end{figure}

\begin{table}[h!]
\centering
\begin{tabular}{||c c c c c||}
\hline
   PPW    & Grid size & $h= r_{0} \Delta \theta = \Delta r$   & $L^2$-norm Rel. Error   & Observed order \\ [0.5ex]
   \hline\hline
    $30$   & $30\times 190$  & $0.03324$  & $1.64\times 10^{-3}$  & $ $    \\
 $40$   & $40\times 253$  & $0.02493$  & $9.19\times 10^{-4}$  & $ 2.02$    \\
 $50$   & $50\times 316$  & $0.01995$  & $5.87\times 10^{-4}$  & $ 2.00$    \\
 $60$   & $60\times 378$  & $0.01667$  & $4.10\times 10^{-4}$  & $ 1.99$    \\
$70$   & $70\times 441$  & $0.01428$  & $3.04\times 10^{-4}$  & $ 1.95$    \\
[1ex]
\hline
\end{tabular}
\caption{Order of convergence of FFP approximation using KDFE$_5$-BVP}
\label{table:1}
\end{table}

\begin{figure}[h!]
\begin{center}
\includegraphics[width=0.5\linewidth]{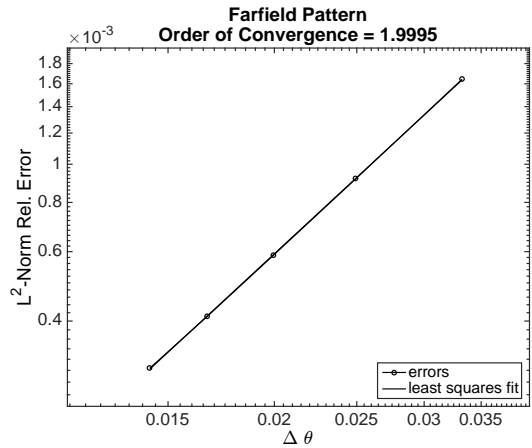} 
\end{center}
\caption{Least squares fitting line for the data in Table \ref{table:1}}
\label{LeastSqKDFE5}
\end{figure}
We observe that the numerical solution of KSFE-BVP also exhibits a second order convergence to the exact solution, although it requires more terms than the solution of KDFE-BVP to converge at the same rate. In fact as shown in Fig. \ref{fig:OrderConv}, nine terms or more are required in the farfield expansion of KSFE-ABC to reach second order convergence while only four or more terms are required in the farfield expansion of KDFE-ABC. Moreover, these numerical experiments provide numerical evidence of the non-equivalence between KSFE$_2$-ABC and BGT$_2$ as established in Theorem \ref{Non-Equiv}.

\subsection{Scattering and radiation from complexly shaped obstacles}
\label{ComplexObst}

Our results in Section \ref{orderconverg2D} for a circular shaped scatterer reveals the high precision that can be achieved by using the farfield expansions as ABCs with the appropriate number of terms and reasonable set of grids. As pointed out above, the accuracy of the overall numerical method is limited by the accuracy of the numerical method employed in the interior of the domain for relatively small number of terms, L, of the farfield ABCs. 
\begin{figure}[!h]
\begin{subfigure}{0.5\textwidth}
\includegraphics[height=0.7\textwidth]{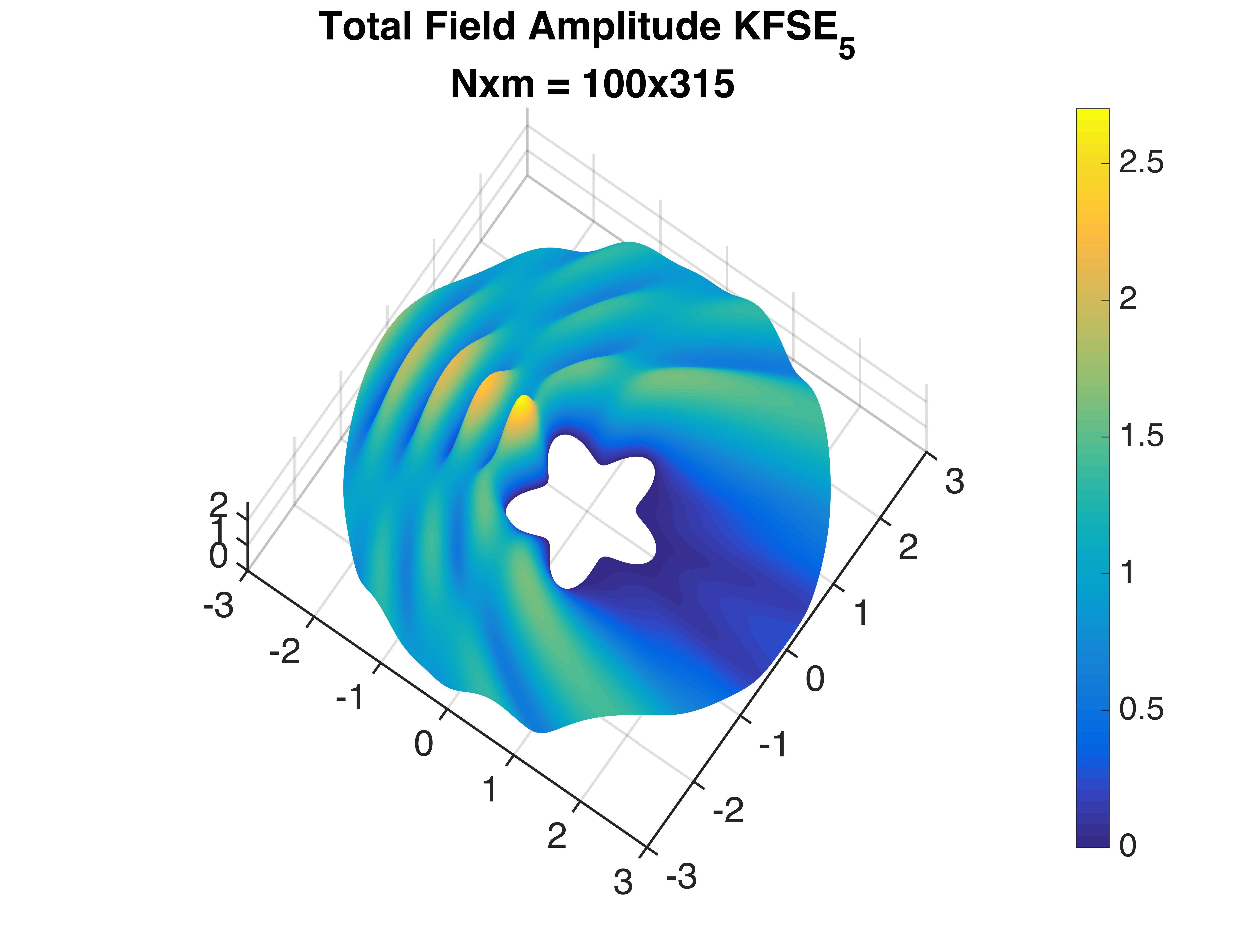} 
\end{subfigure}
\begin{subfigure}{0.5\textwidth}
\includegraphics[height=0.7\textwidth]{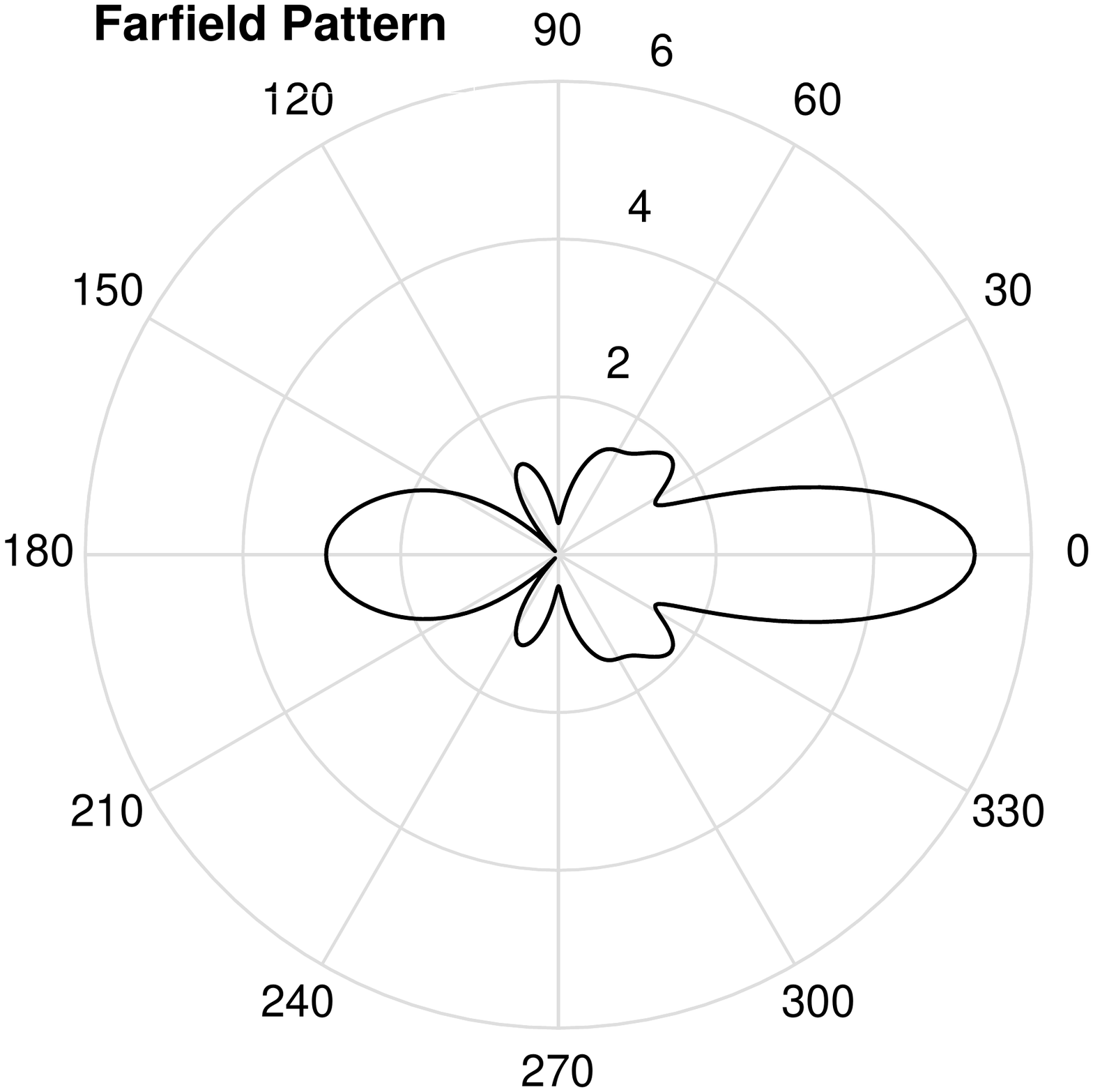}
\end{subfigure}
\begin{subfigure}{0.5\textwidth}
\includegraphics[height=0.7\textwidth]{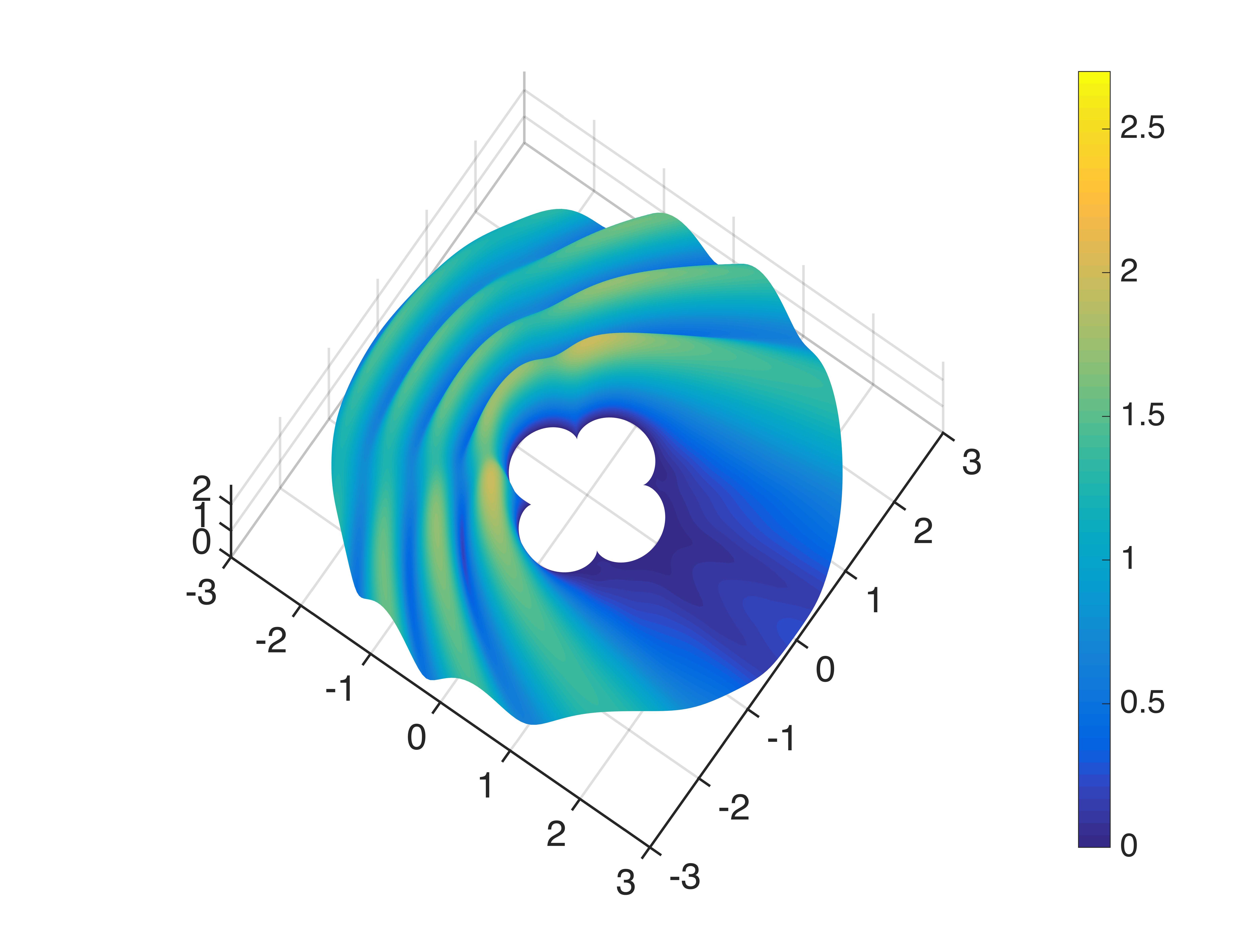}
\end{subfigure}
\begin{subfigure}{0.5\textwidth}
\includegraphics[height=0.7\textwidth]{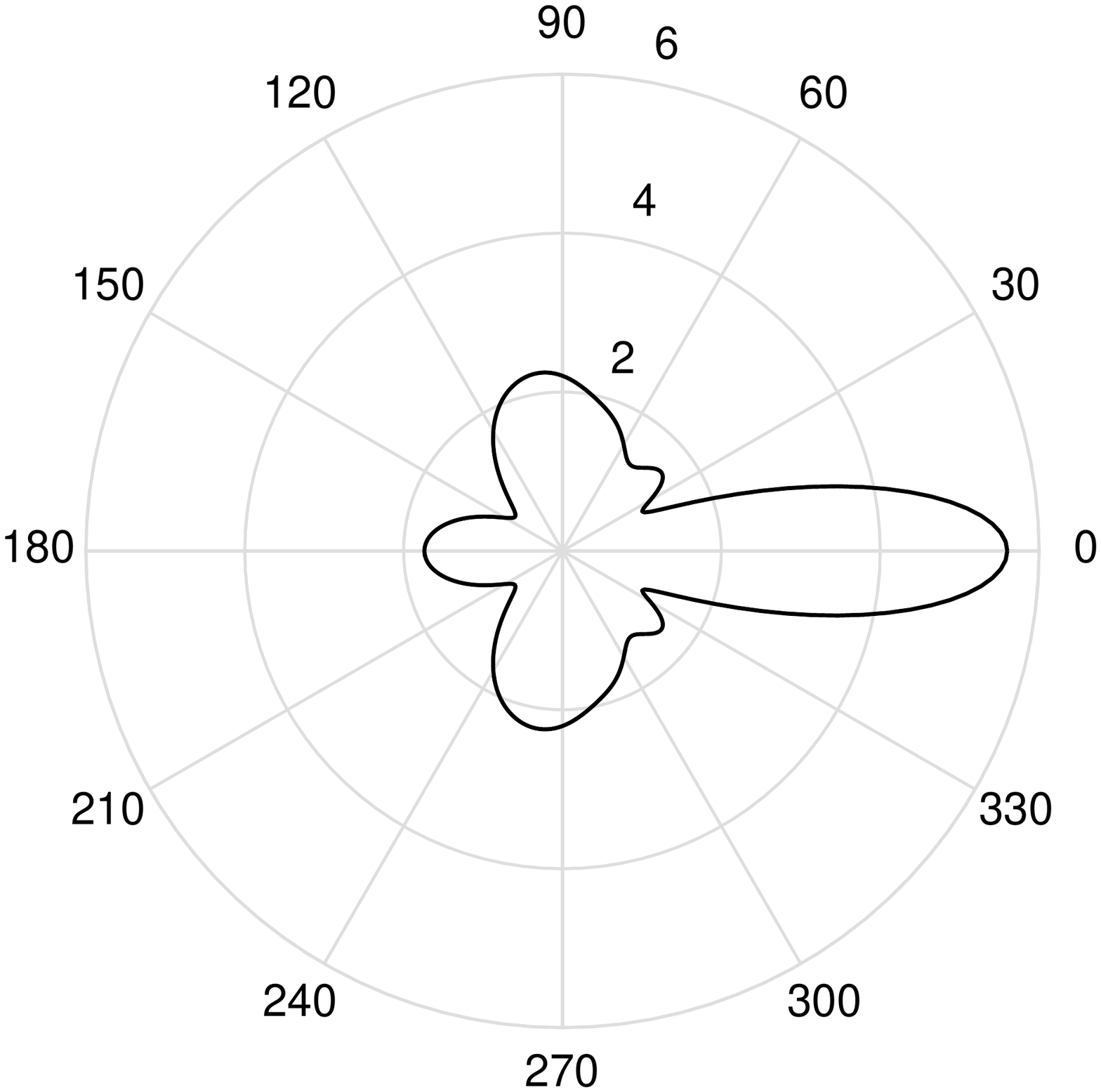}
\end{subfigure}
\begin{subfigure}{0.5\textwidth}
\includegraphics[height=0.7\textwidth]{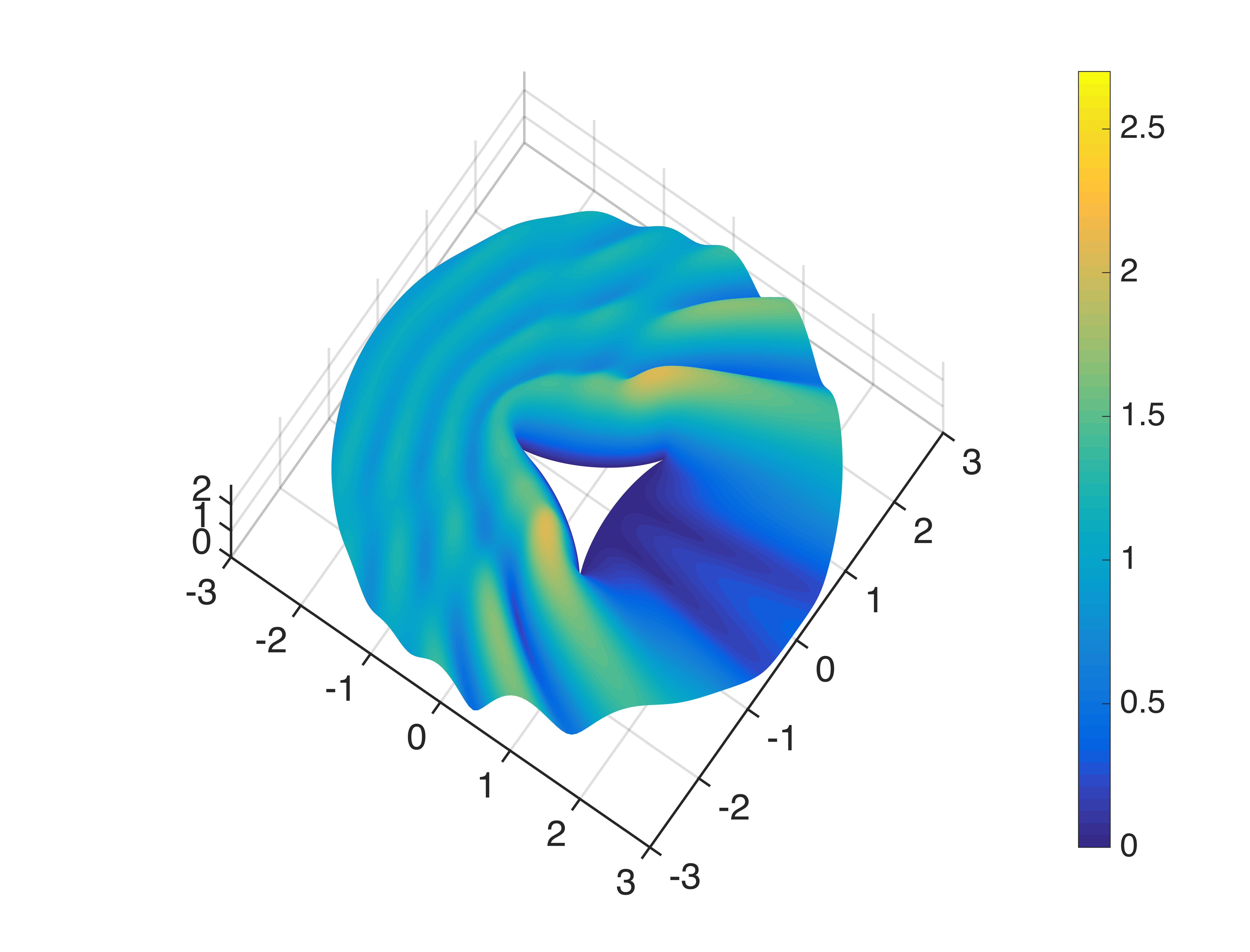}
\end{subfigure}
\begin{subfigure}{0.5\textwidth}
\includegraphics[height=0.7\textwidth]{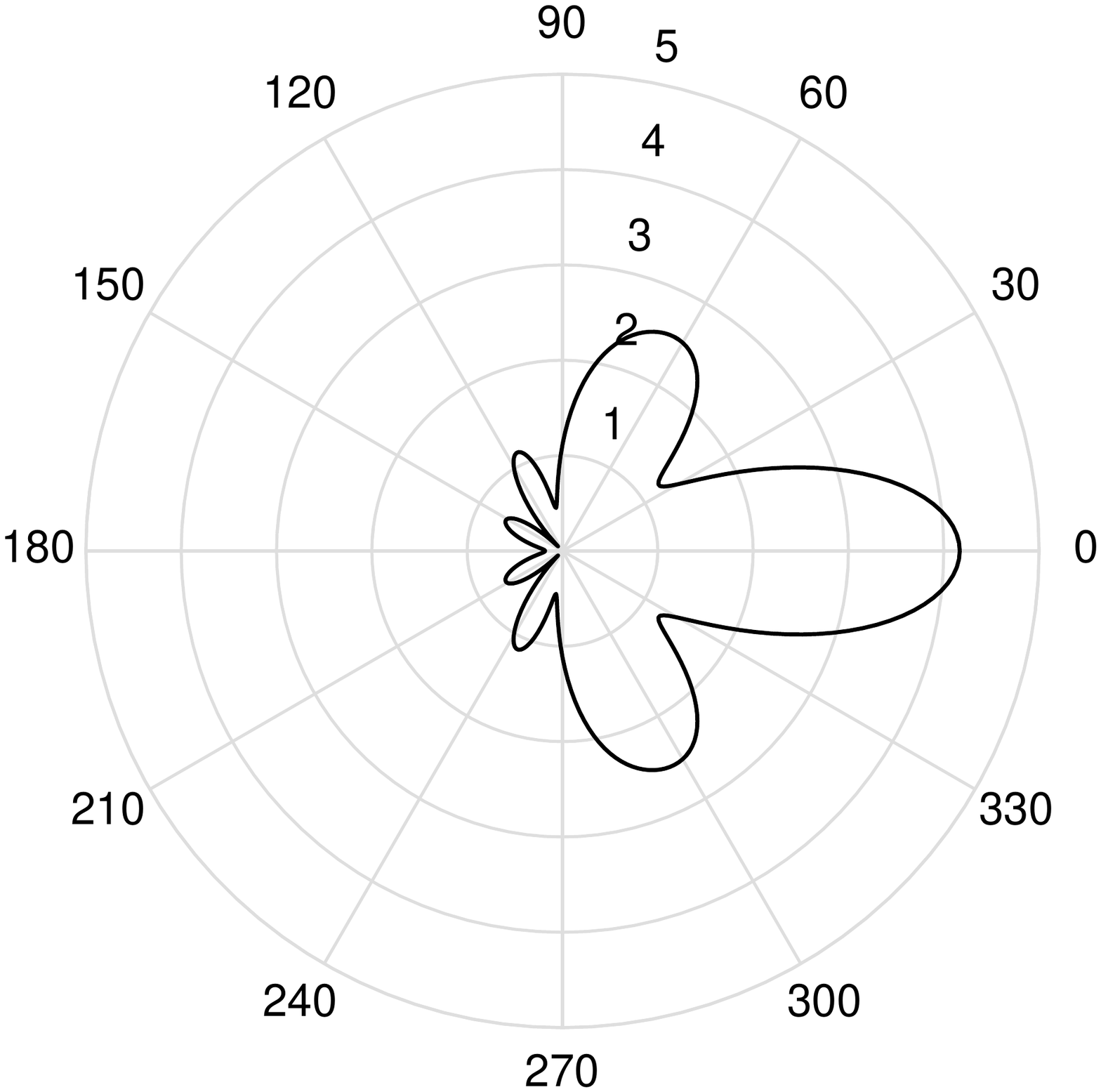}
\end{subfigure}
\caption{Total field and corresponding FFP for scattering from complexly shaped obstacles on elliptic-polar grids using KSFE$_5$-ABC with $k=2\pi$, $R=3$, and PPW=50.} 
\label{fig:ComplexObst}
\end{figure}
In this section, we take advantage of this fact to numerically solve more realistic scattering problems. In fact, we find numerical solutions for acoustic scattering problems from obstacles with complexly shaped bounding curves such as a star, epicycloid, and astroid. We choose as the artificial boundary a circle of radius R~=~3 and the frequency $k =2\pi$.
As described in Section \ref{ScattComplex}, the differential equations defining these BVPs are written in terms of generalized curvilinear coordinates that Acosta and Villamizar derived in 
\cite{JCP2010}. The corresponding grids for these curvilinear coordinates  were obtained from an elliptic grid generator and they were named elliptic-polar grids. Following the circular scatterer case, we use a second order centered finite difference method as our numerical technique for the interior points.  A detailed account of the discretized equations in curvilinear coordinates are also found at \cite{JCP2010}. We employ KSFE$_5$ as our farfield expansion combined with PPW =50 (points per wavelength).  The results are illustrated in Fig. \ref{fig:ComplexObst} where the total field and its corresponding FFP are shown for each one of these obstacles.  The parametric equations of these bounding curves are given by
\begin{eqnarray}
\mbox{Star:}\quad x(\theta)=
0.2(4+\cos(5\theta))\cos(\theta)\qquad
y(\theta)=0.2(4+\cos(5\theta))\sin(\theta),\qquad 0\leq \theta\leq
2\pi. \label{star}
\end{eqnarray}
\begin{eqnarray}
\label{epicycloid}
\mbox{Epicycloid:}\quad& &x(\theta)= ((5\sin(-(\theta+5\pi/4))-
     \sin(-5(\theta+5\pi/4)))\cos(\pi/4)-\nonumber\\\nonumber
     & &\qquad (5\cos(-(\theta+5\pi/4))-
     \cos(-5(\theta+5\pi/4)))\sin(\pi/4))1/6\\
& &y(\theta)= ((5\sin(-(\theta+5\pi/4))-
     \sin(-5(\theta+5\pi/4)))\sin(\pi/4)-\\\nonumber
     & & \qquad(5\cos(-(\theta+5\pi/4))-
     \cos(-5(\theta+5\pi/4)))\cos(\pi/4))1/6,\nonumber
     \qquad 0\leq \theta\leq 2\pi. 
\end{eqnarray}
\begin{eqnarray}
\mbox{Astroid:}\quad& &x(\theta)= \left(2\cos(\theta-\pi/3)+
     \cos(2(\theta-\pi/3))\right)\cos(\pi/3)/3-\nonumber\\
     & &\qquad\left(2\sin(\theta-\pi/3)+
     \sin(2(\theta-\pi/3))\right)\sin(\pi/3)/3\\
& &y(\theta)= \left(2\cos(\theta-\pi/3)+
     \cos(2(\theta-\pi/3))\right)\sin(\pi/3)/3-\nonumber\\\nonumber
     & &\qquad\left(2\sin(\theta-\pi/3)+
     \sin(2(\theta-\pi/3))\right)\cos(\pi/3)/3,
     \qquad 0\leq \theta\leq 2\pi. \label{astroid}
\end{eqnarray}
For the experiments corresponding to the graphs shown in Fig. \ref{fig:ComplexObst}, 
using  relatively fine grids with $PPW=50$, we did not find significant changes in the numerical solution by increasing the number of terms in the KSFE$_L$ condition up to $L=12$ terms. 

Next, we discuss the numerical results for radiating problems defined in the exterior region $\Omega$ bounded internally by an arbitrary simple closed curve $\Gamma$.
These BVPs consist of Helmholtz equation, Sommerfeldt radiation condition, and a Dirichlet condition on the complexly shaped bounding curve $\Gamma$. By imposing an appropriate boundary condition on $\Gamma$, we can easily prescribe a solution for each one of these BVPs. In fact, consider the function $u$ defined in the exterior region $\Omega$ from the superposition of two sources which are located inside the closed region bounded by  $\Gamma$. More precisely, $u$ is given in terms of Hankel functions of first kind of order zero as
\begin{equation}
u({\bf{x}}) = H_0^{(1)}(k r_{1}({\bf x})) + H_0^{(1)}(k r_{2}({\bf x}) ), \qquad {\bf x}\in \Omega\label{SourceSoln}
\end{equation}
where $r_1 = |{\bf{x}}-{\bf x}_1|$, and $r_2 = |{\bf x}-{\bf x}_2|$ with ${\bf x}_1$ and ${\bf x}_2$ inside the region bounded by $\Gamma$. 
Clearly, the function $u$ satisfies Helmholtz equation in $\Omega$ since $H_0^{(1)}(k r_{i})$ does for $i=1,2$. 
It also satisfies the Sommerfeld radiation condition. Thus if we also impose the values of $u$ at the boundary $\Gamma$ (superposition of the two sources) as the boundary condition on $\Gamma$, the function $u$ defined by (\ref{SourceSoln}) satisfies the radiating problem just defined, regardless of the shape of the bounding curve $\Gamma$.

Starting with the previously superimposed boundary condition on the bounding curve $\Gamma$, it is possible to obtain a numerical solution. First, we transform the unbounded radiating problem into a bounded one by introducing the KDFE-ABC or KSFE-ABC on a circular artificial boundary ($r=R$). Then, we apply the proposed numerical technique in generalized curvilinear coordinates in the region $\Omega^{-},$ bounded internally by $\Gamma$ and externally by the circle of radius $R$ to obtain the numerical solution sought.

\begin{figure}[!h]
\begin{subfigure}{0.35\textwidth}
\includegraphics[width=\linewidth, height=5.5cm ]{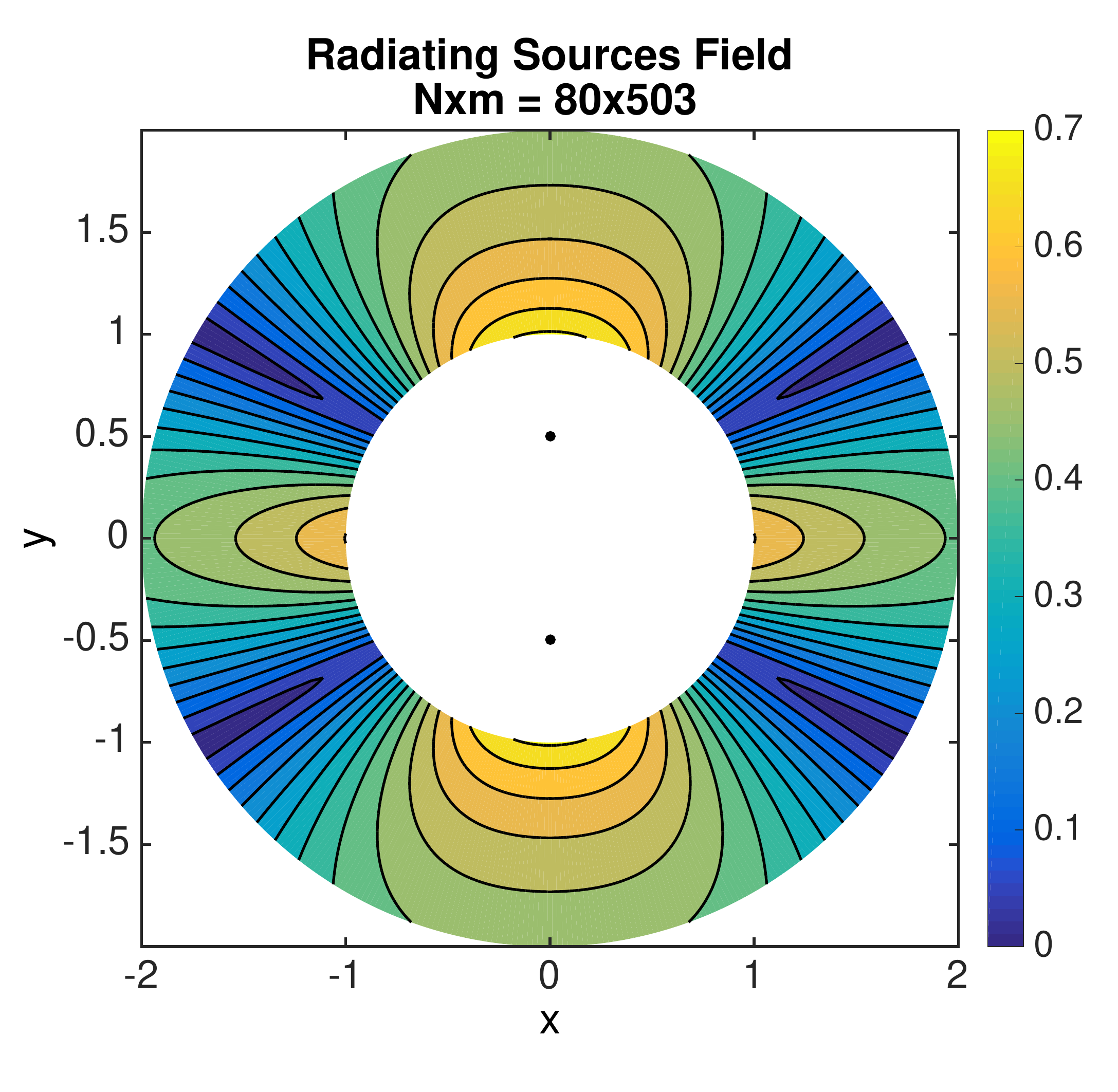} 
\end{subfigure}
\begin{subfigure}{0.3\textwidth}
\includegraphics[width=\linewidth, height=5.2cm]{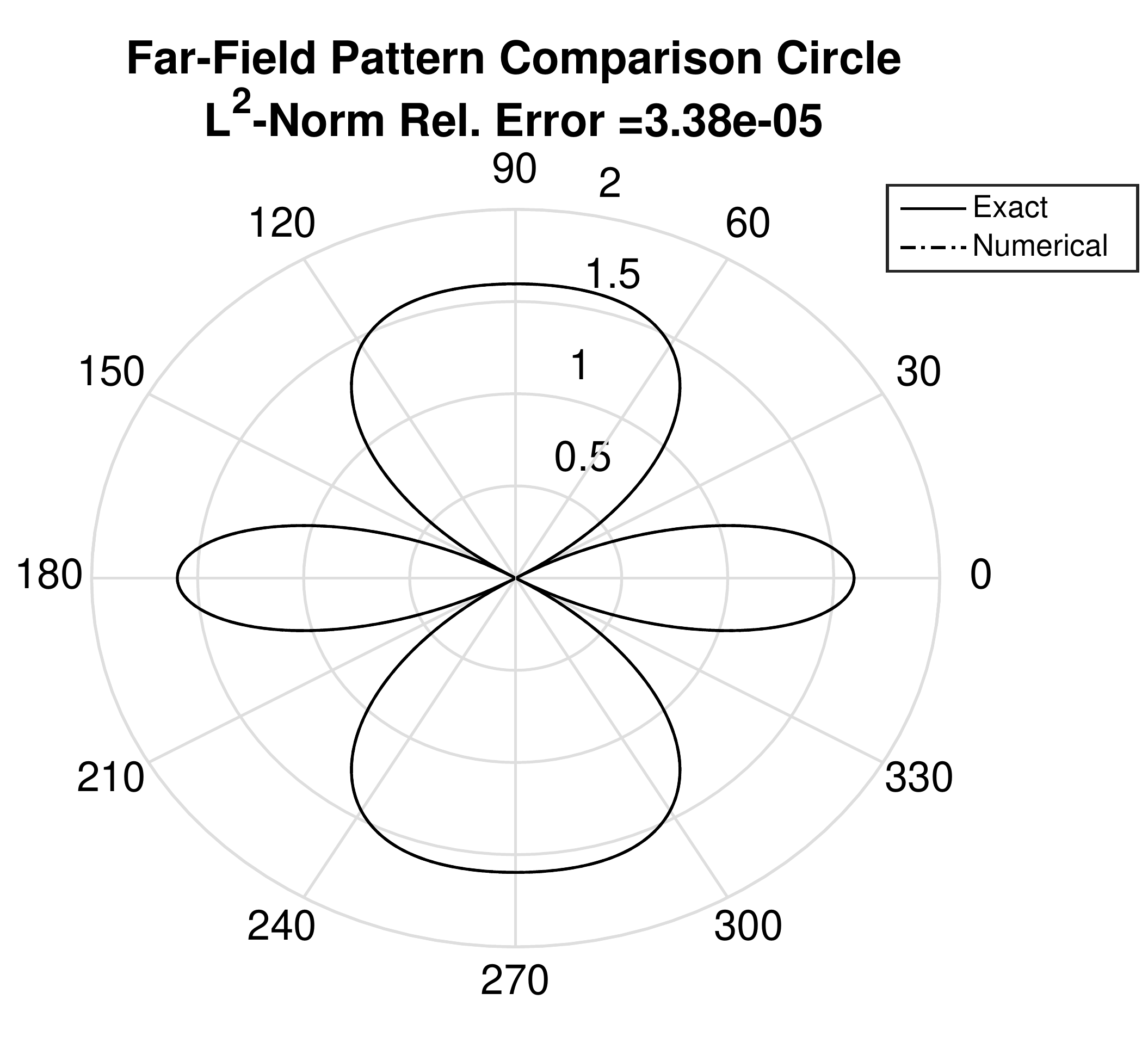}
\end{subfigure}
\begin{subfigure}{0.3\textwidth}
\includegraphics[width=\linewidth, height=5cm ]{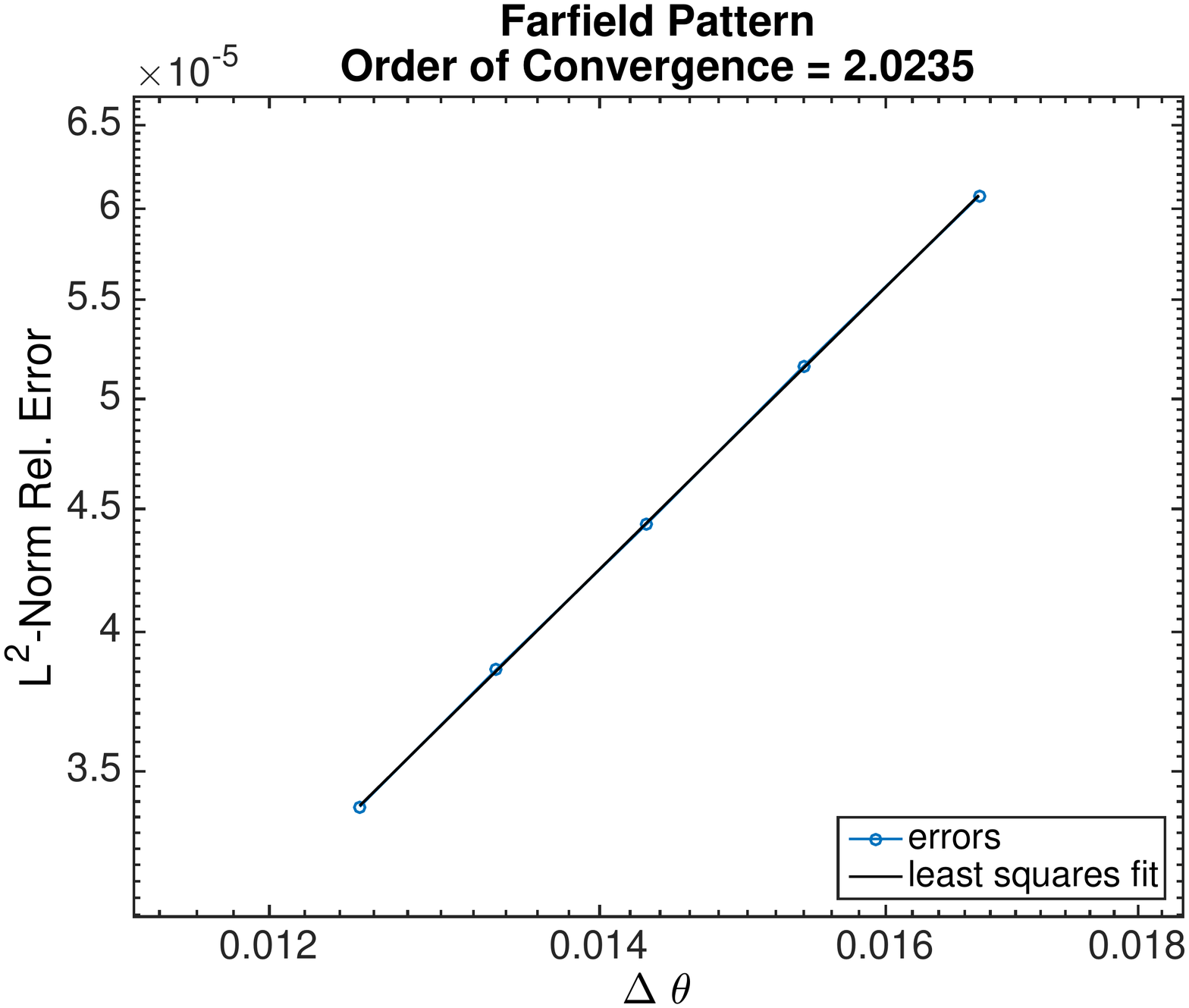}
\end{subfigure}
\begin{subfigure}{0.35\textwidth}
\includegraphics[width=\linewidth, height=5.5cm ]{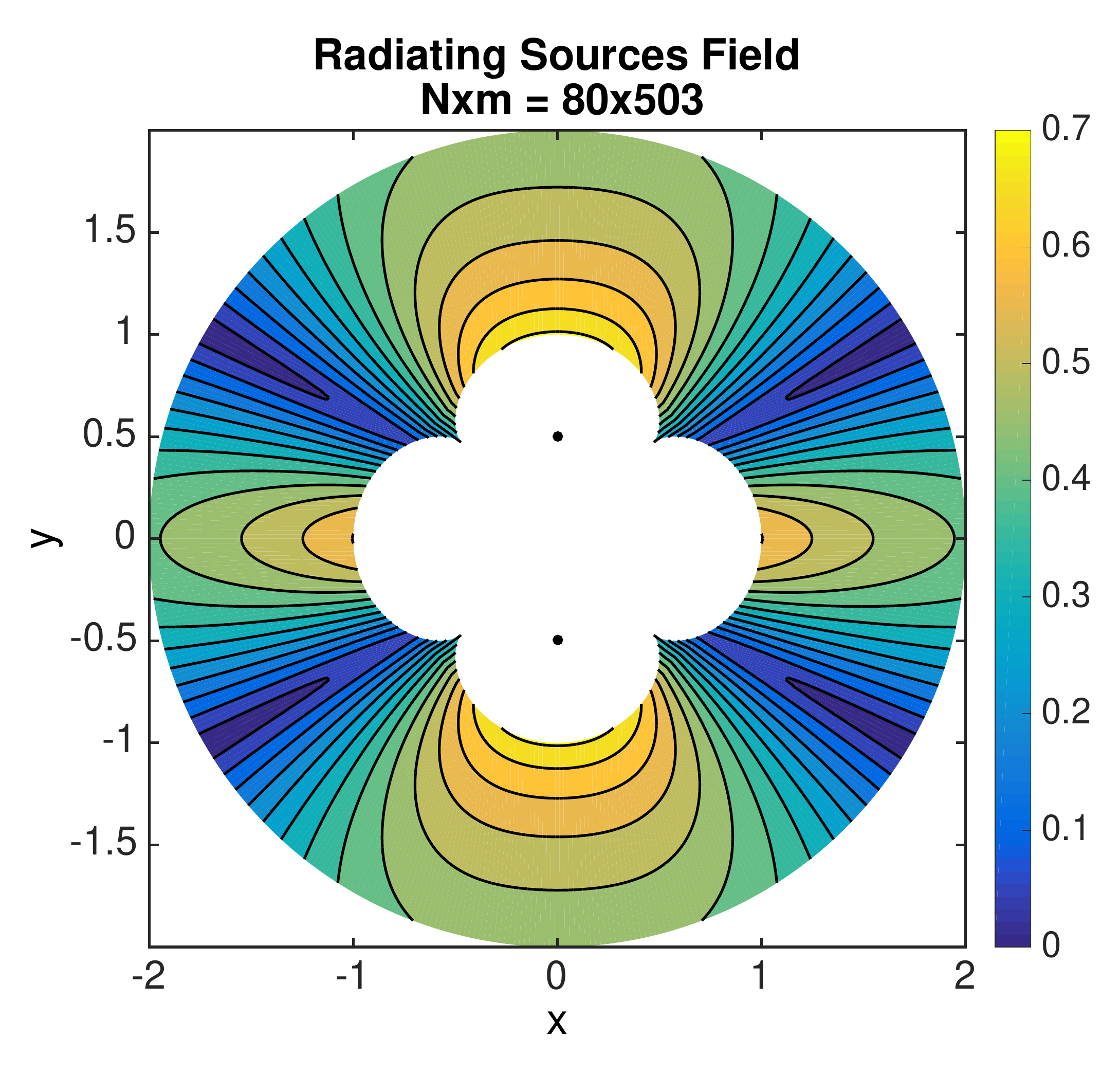} 
\end{subfigure}
\begin{subfigure}{0.3\textwidth}
\includegraphics[width=\linewidth, height=5.2cm]{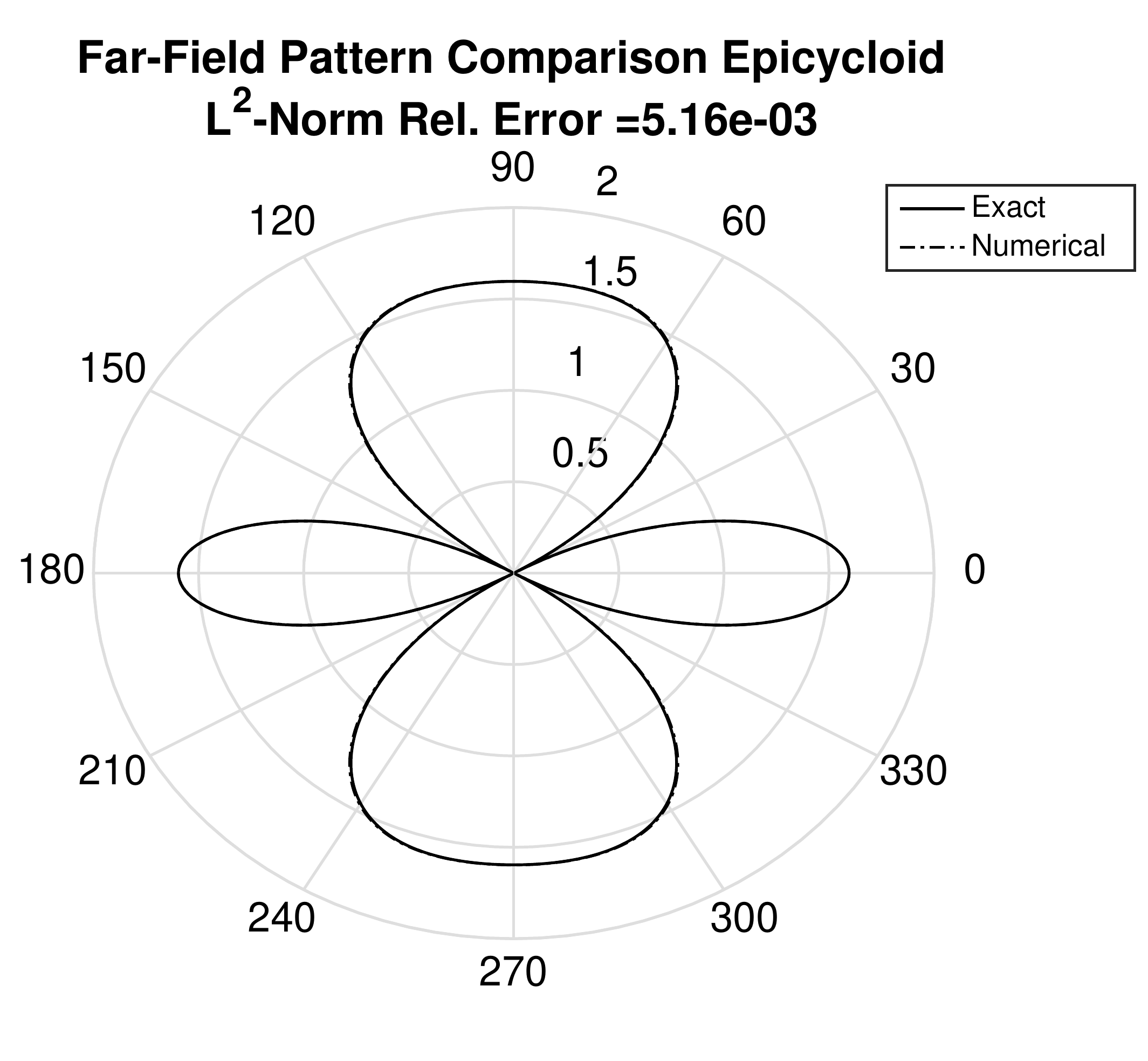}
\end{subfigure}
\hspace{0.1cm}
\begin{subfigure}{0.3\textwidth}
\includegraphics[width=\linewidth, height=5cm ]{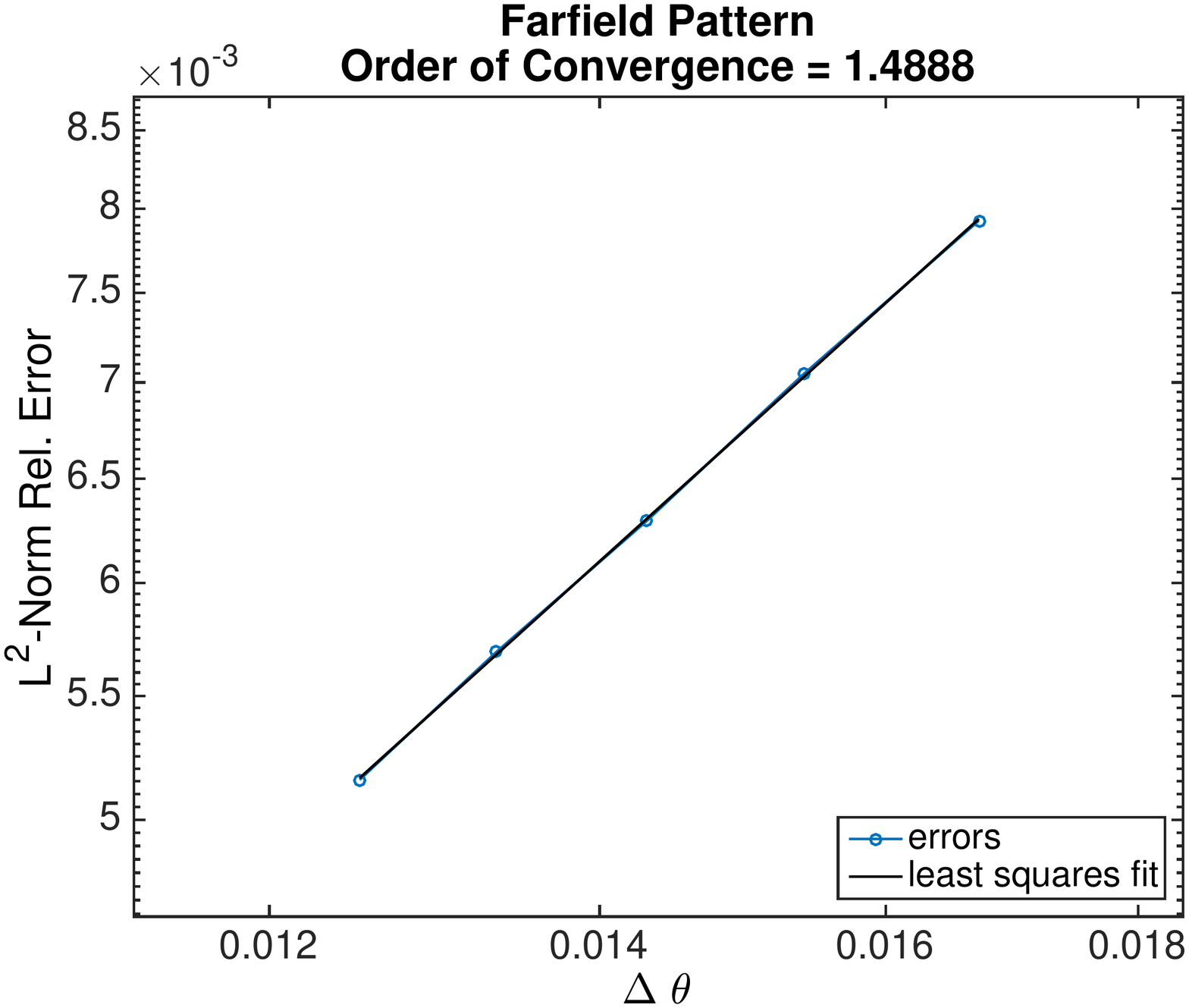}
\end{subfigure}
\begin{subfigure}{0.35\textwidth}
\includegraphics[width=\linewidth, height=5.5cm ]{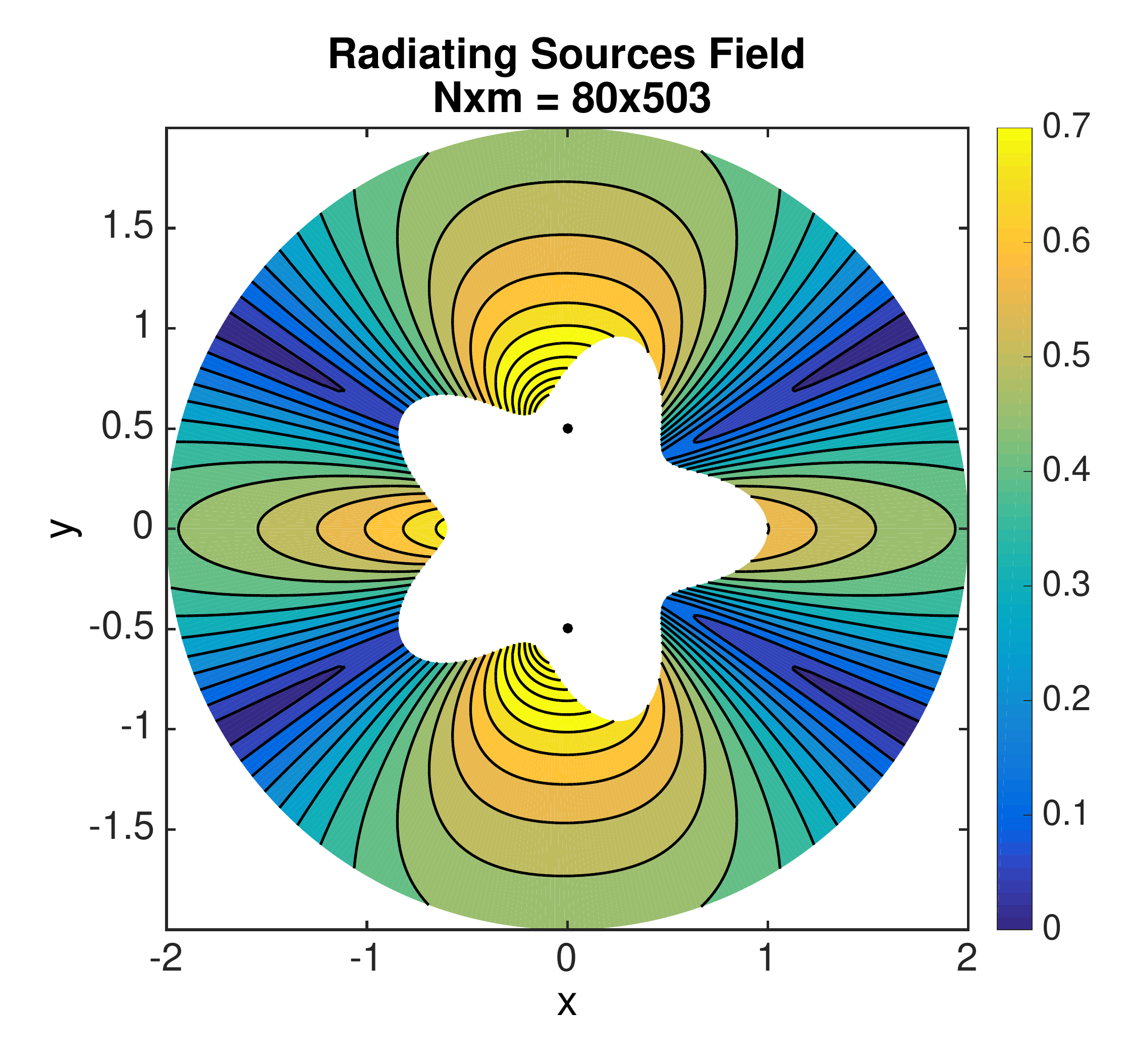} 
\end{subfigure}
\hspace{0.2cm}
\begin{subfigure}{0.3\textwidth}
\includegraphics[width=\linewidth, height=5.2cm]{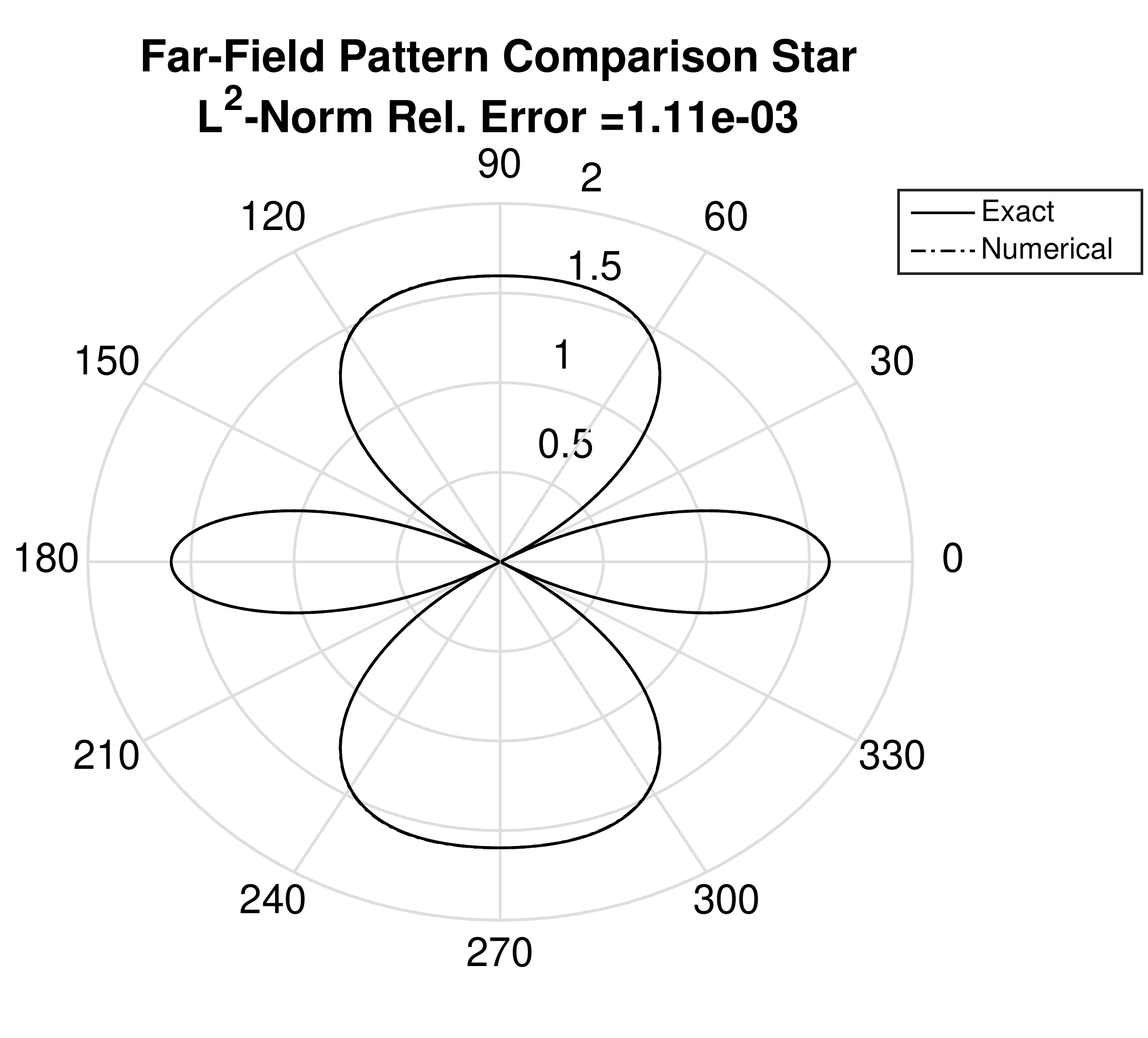}
\end{subfigure}
\hspace{0.2cm}
\begin{subfigure}{0.3\textwidth}
\includegraphics[width=\linewidth, height=5cm ]{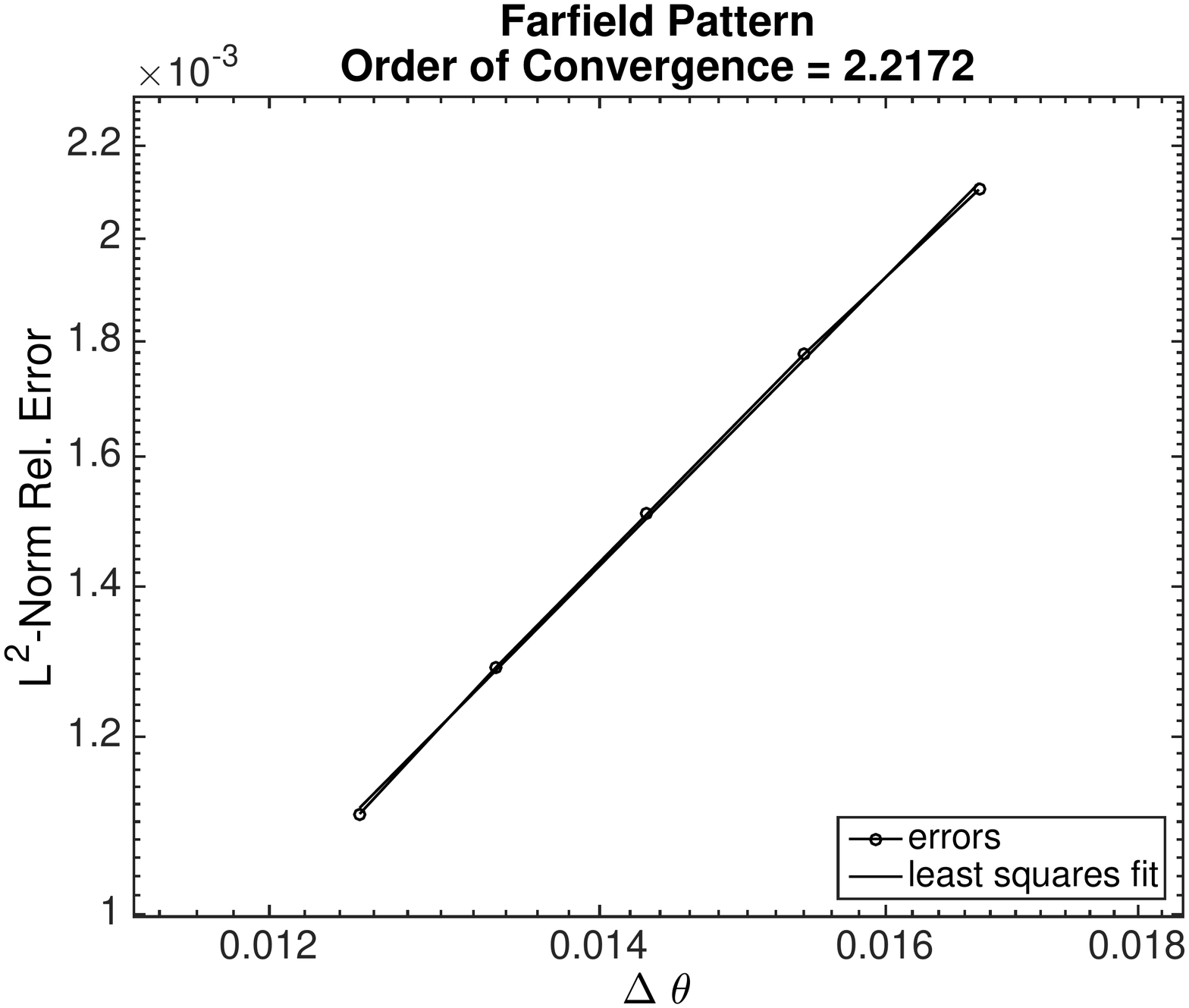}
\end{subfigure}
\vspace{-.7cm}
\caption{Numerical computation of a radiating field from two sources using KSFE$_{10}$-ABC, $k=2\pi$, $R=2$, and PPW=80. Order of convergence of FFP approximation for PPW = 60,65,70,75,80 for complex bounding curves.}
\label{fig:ComplexObstSources}
\end{figure}

The relevant data employed in our numerical experiments is the following: artificial boundary $R=2$, frequency $k=2\pi$, number of terms in the KSFE expansion $L=10$, location of sources ${\bf x}_1 =(0,1/2)$ and ${\bf x}_2 =(0,-1/2)$, set of grid points PPW $= 60, 65, 70, 75, 80$.
We show that these numerical solutions indeed converge to the exact prescribed solution (\ref{SourceSoln}) of the original radiating BVP. 
This is illustrated in Fig. \ref{fig:ComplexObstSources}
where the known radiating field from the two sources is numerically approximated in three different regions $\Omega^{-}$ which are internally bounded by three different curves. They are a circle of radius $r_0=1$, the epicycloid boundary curve defined in (\ref{epicycloid}), and the star curve defined in (\ref{star}).  The relative $L^2$-norm error between the FFP of the prescribed solution and the approximated solution is computed for each different grid. Then, the order of convergence is estimated based on these errors.
As seen in Fig. \ref{fig:ComplexObstSources}, we are able to prove quadratic convergence for the circle and for the star bounding curves. However, for the epicycloid we can only get 1.5 as order of convergence for the same set of grid points and number of terms in the farfield expansion $L$. This is due to the difficulty of generating conforming smooth grids in the neighborhood of the epicycloid singularities.

\subsection{Scattering from a spherical obstacle. The axisymmetric case. Numerical approximation and order of convergence}
\label{orderconverg3D}

In this section, we discuss the results for the scattering from a spherical obstacle modeled by WFE$_L$-BVP as described in Section \ref{ScattSphere}. 
\begin{figure}[!h]
\begin{subfigure}{0.37\textwidth}
\includegraphics[width=\linewidth, height=5.5cm ]{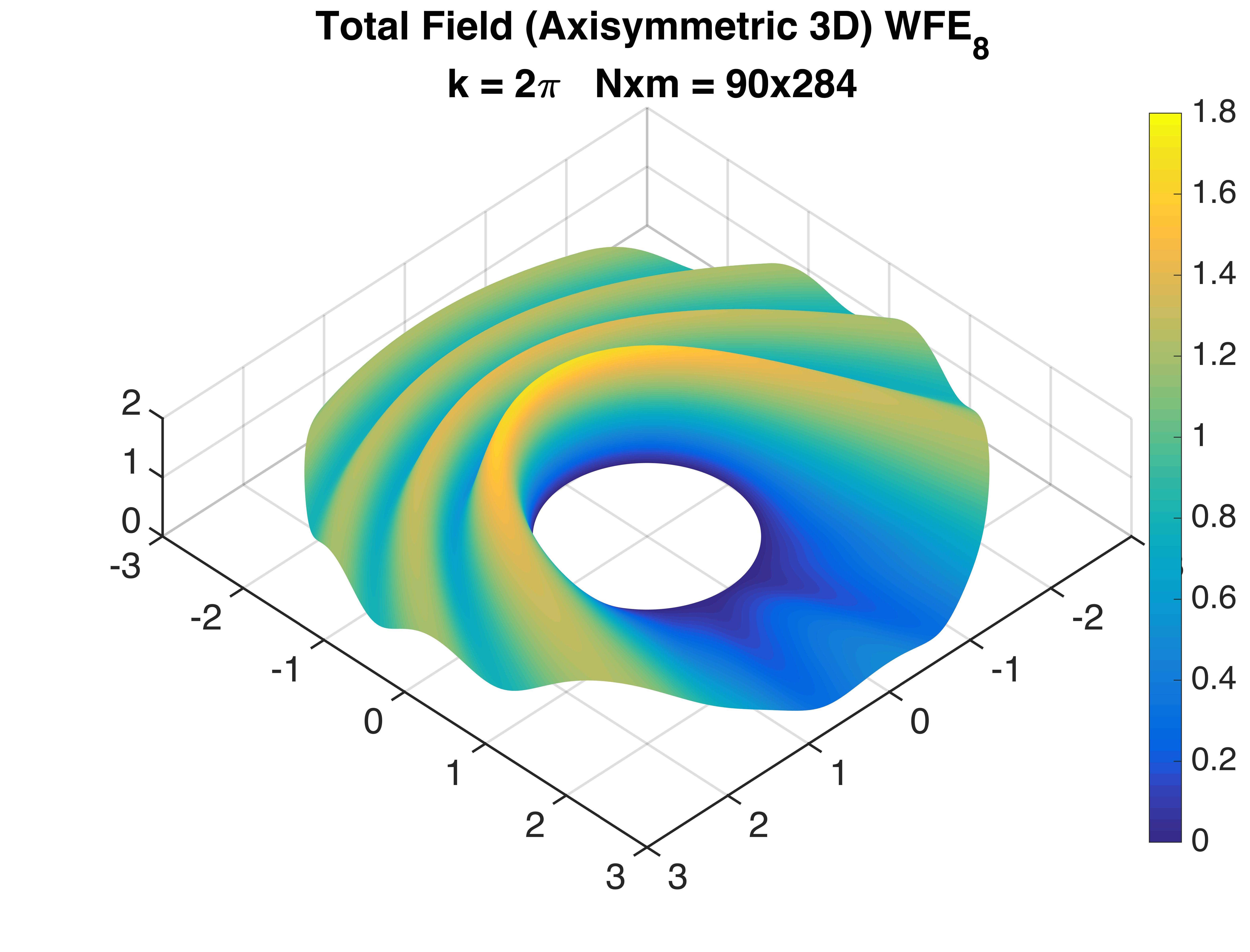} 
\end{subfigure}
\begin{subfigure}{0.37\textwidth}
\includegraphics[width=\linewidth, height=5.5cm]{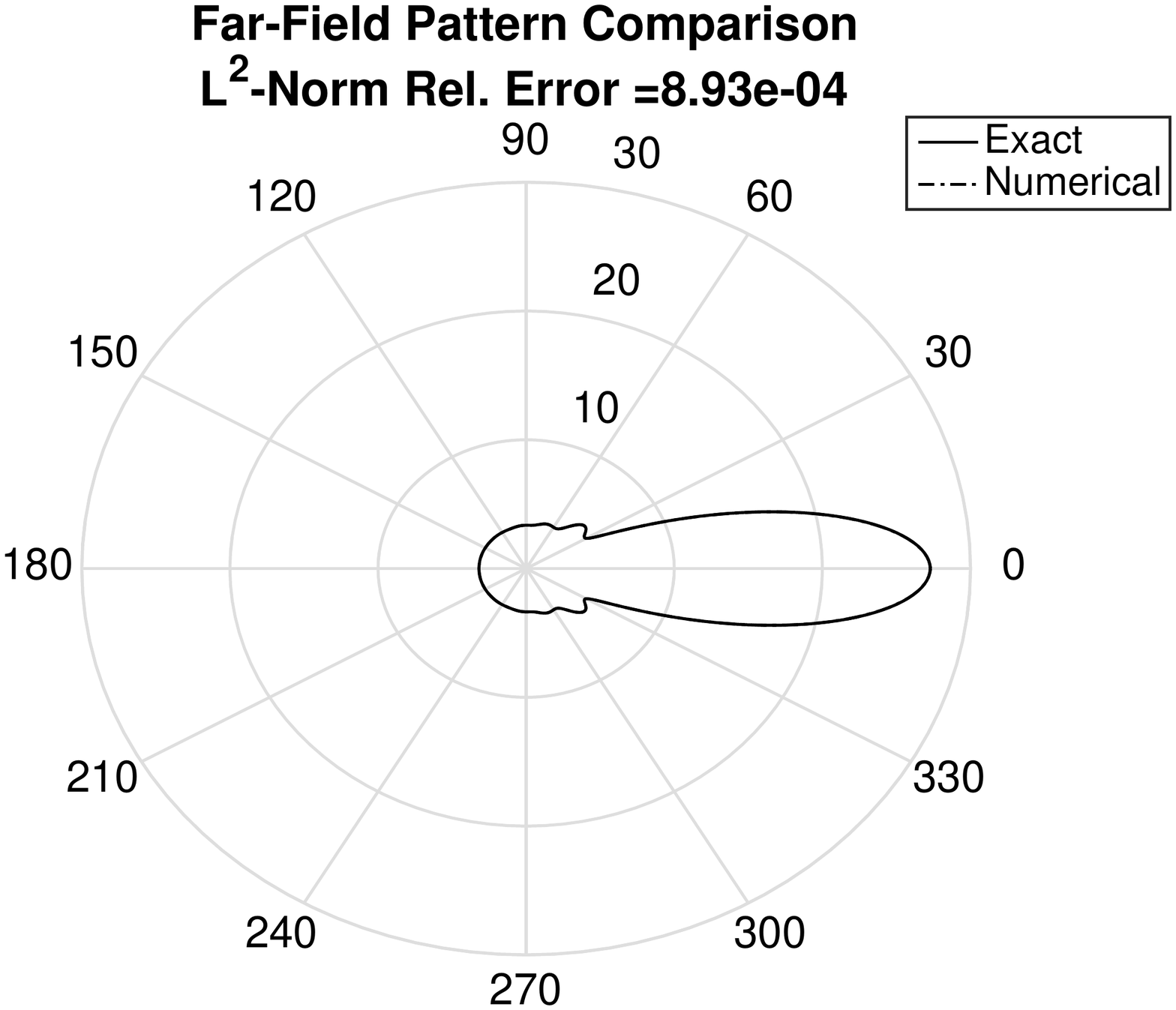}
\end{subfigure}
\begin{subfigure}{0.25\textwidth}
\includegraphics[width=\linewidth, height=4.cm ]{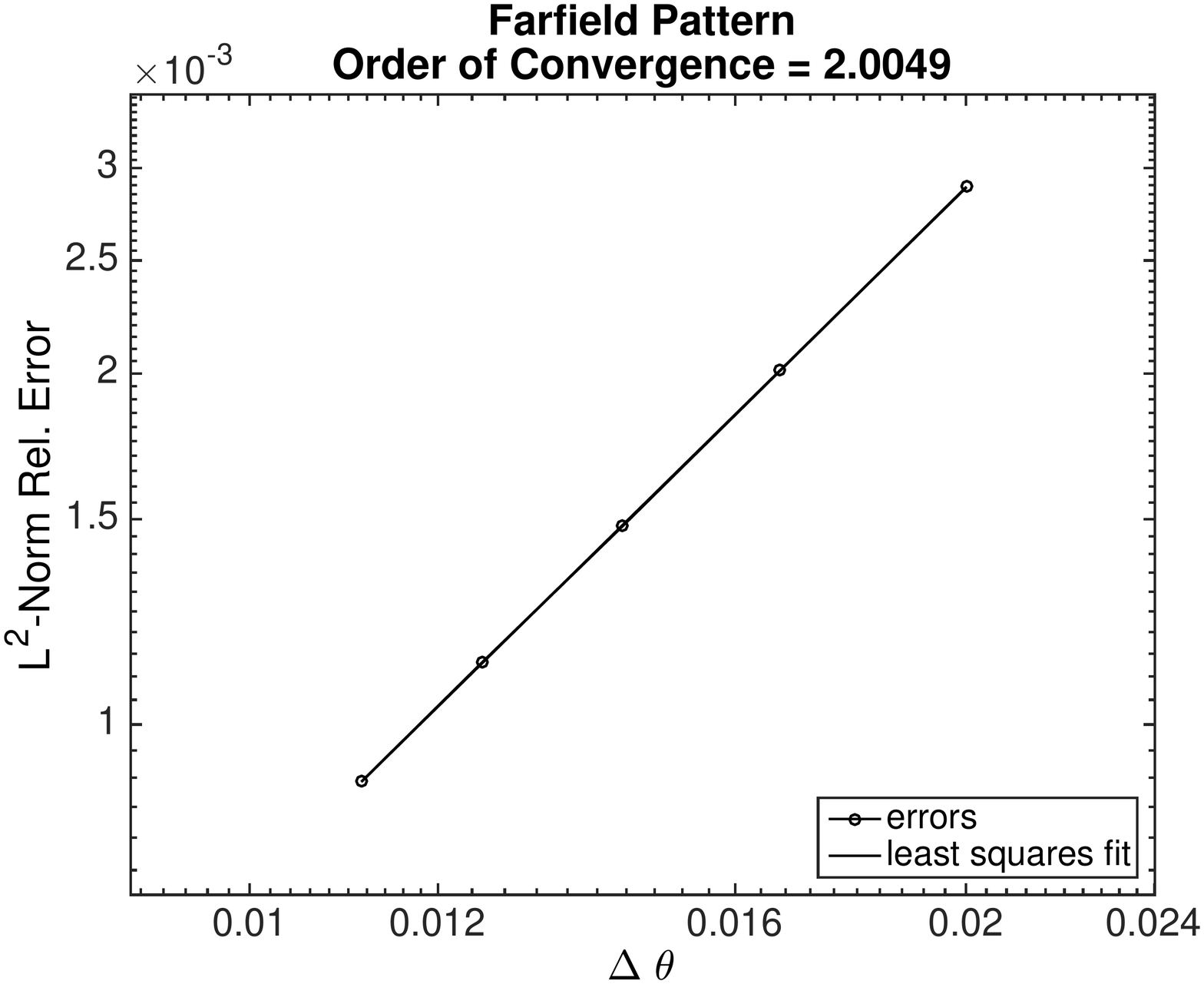}
\end{subfigure}
\begin{subfigure}{0.37\textwidth}
\includegraphics[width=\linewidth, height=5.5cm ]{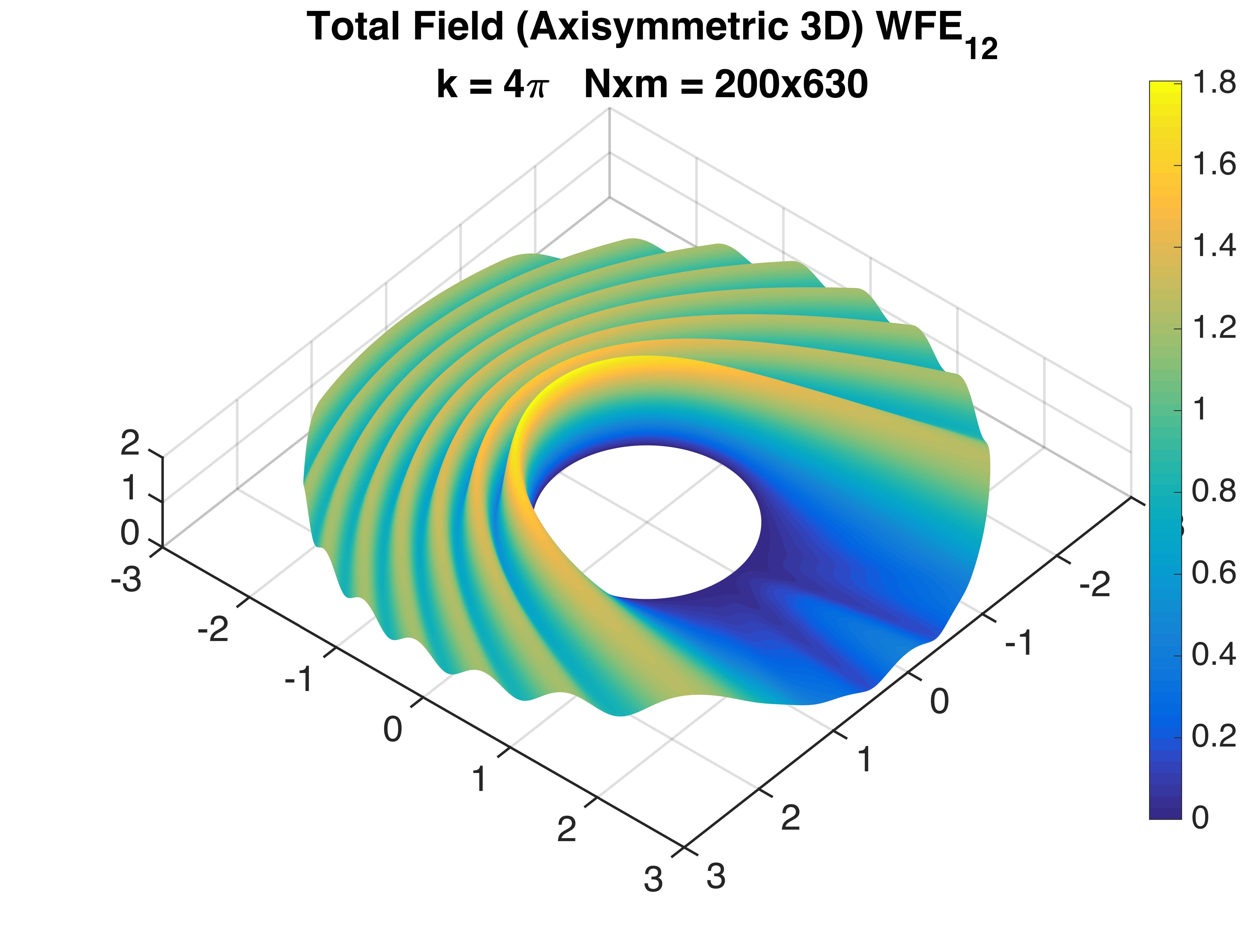} 
\end{subfigure}
\begin{subfigure}{0.37\textwidth}
\includegraphics[width=\linewidth, height=5.5cm]{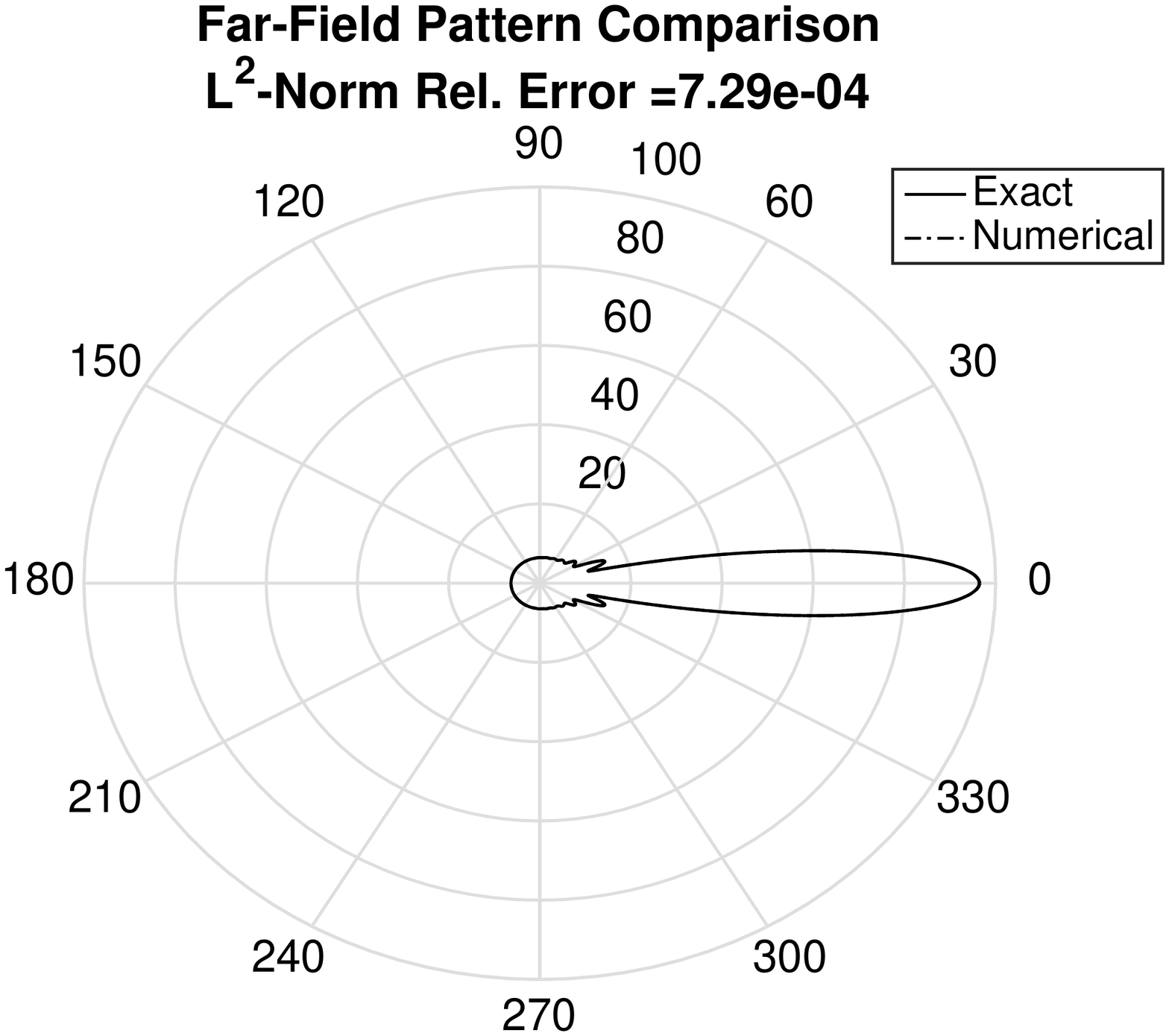}
\end{subfigure}
\begin{subfigure}{0.25\textwidth}
\includegraphics[width=\linewidth, height=4.cm ]{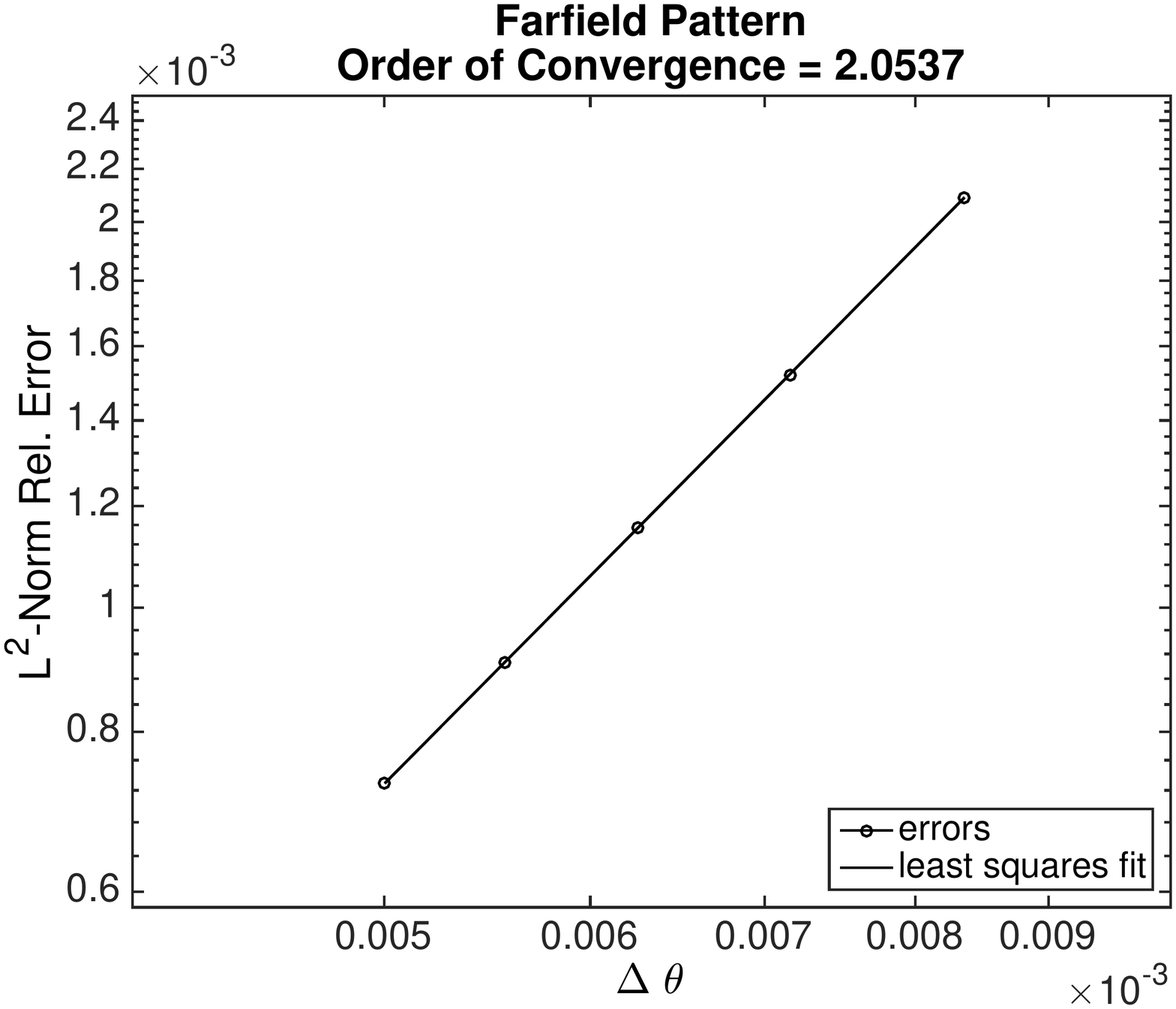}
\end{subfigure}
\vspace{-.7cm}
\caption{Numerical results for scattering from a spherical scatterer using Wilcox farfield ABC: cross-sections of the total field for arbitrary $\theta$, farfield pattern, and order of convergence for two different frequencies $k=2\pi,4\pi$.} 
\label{fig:3DSpericalScattering}
\end{figure}

In Fig. \ref{fig:3DSpericalScattering},  cross-sections of the total field for an arbitrary angle $\theta$ are depicted.  The middle graphs corresponds to the approximation of the farfield pattern of this scattering problem. These  graphs were extended to the interval $[0,2\pi]$ by taking the mirror image of the solution in $[0,\pi]$. Finally, the rightmost graphs show the second order convergence of the numerical solution to the exact solution when Wilcox farfield expansions ABC are employed. The data employed to generate the graphs in the top row of Fig.  \ref{fig:3DSpericalScattering} is: $k=2\pi$, $R=3$,
terms in WFE$_L$,  $L = 8$, and set of grid points used to achieve the second order convergence, PPW = 25, 30, 35, 40, 45. Similarly, the bottom row graphs were obtained using: $k=4\pi$, $R=3$, terms in WFE$_L$,  $L = 8$, and set of grid points used to achieve the second order convergence, PPW = 30, 35, 40, 45, 50.

These results reveals the high accuracy that can be achieved using the exact Wilcox farfield expansions in the 3D case. As we showed in the 2D case, the accuracy of the numerical solutions depends only on the order of approximation of the numerical method employed in the interior of $\Omega^{-}$ when enough terms in the exact farfield expansions ABCs are used.

\section{Concluding remarks}
\label{Section.Conclusions}
We have derived exact local ABCs for acoustic waves in two-dimensions (KDFE), and in three-dimensions (WFE).  We have constructed them directly from Karp's and Atkinson-Wilcox's farfield expansions, respectively. A previous attempt by Zarmi and Turkel \cite{Zarmi-Turkel} to derive a high order local ABC from Karp's expansion was partially successful. However, they were able to obtain other high order local conditions using an annihilating technique more general than the procedure used to obtain HH-ABC.

\begin{figure}[!h]
\begin{subfigure}{0.33\textwidth}
\includegraphics[width=\linewidth, trim = {20 20 20 20}]{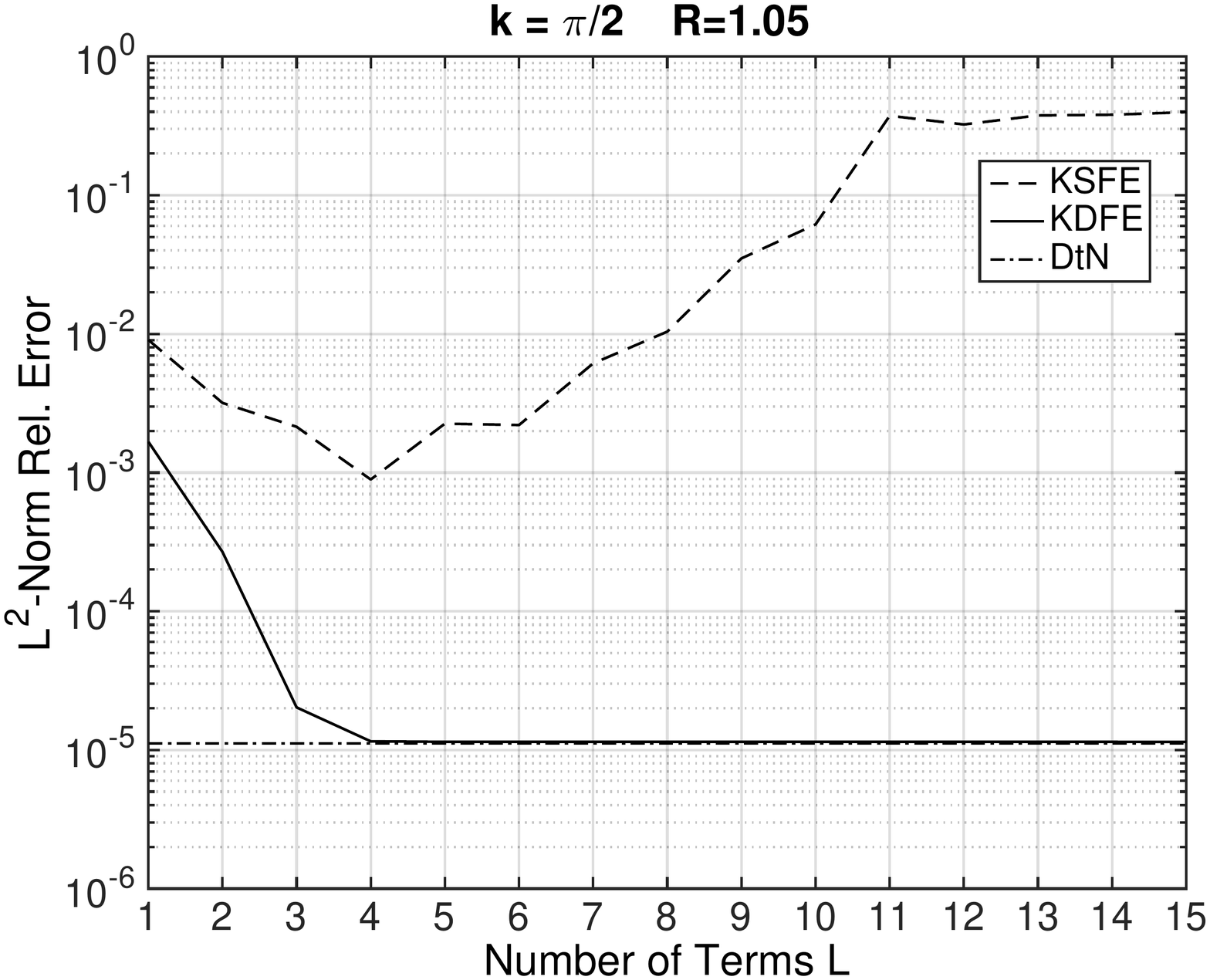} 
\end{subfigure}
\begin{subfigure}{0.33\textwidth}
\includegraphics[width=\linewidth, trim = {20 20 20 20}]{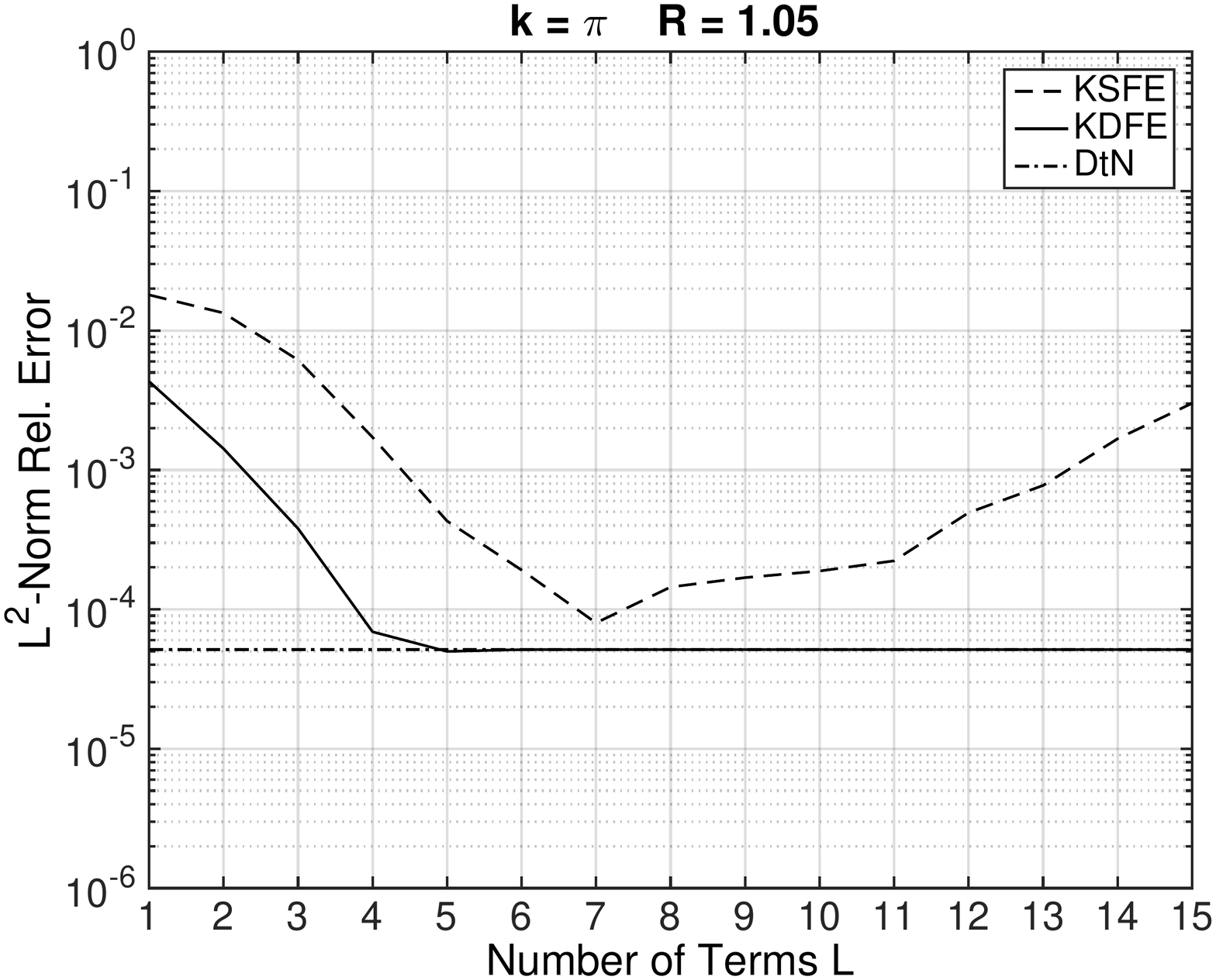}
\end{subfigure}
\begin{subfigure}{0.33\textwidth}
\includegraphics[width=\linewidth, trim = {20 20 20 20}]{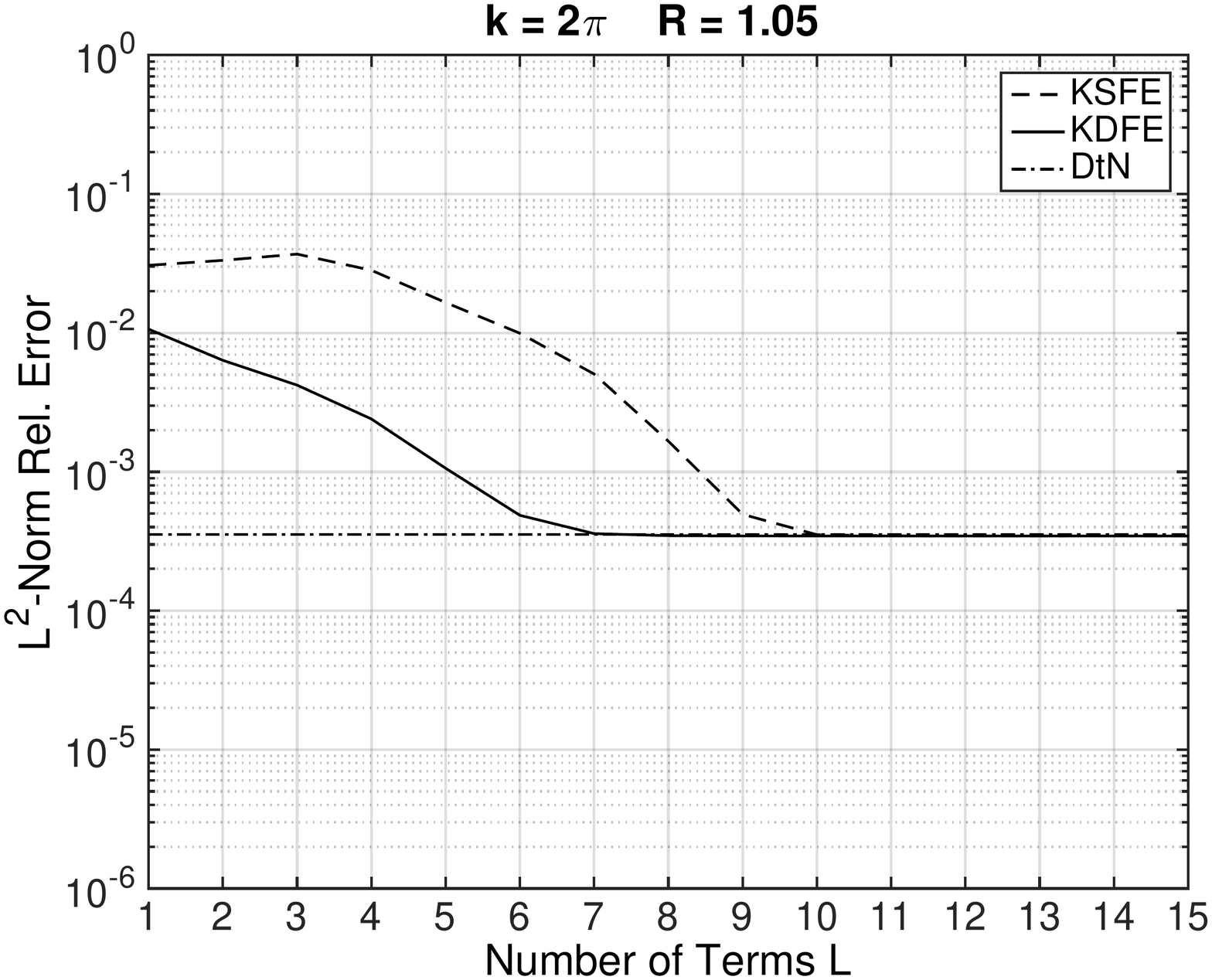}
\end{subfigure}
\caption{Convergence properties of the numerical approximation of the FFP, obtained from KSFE$_L$-BVP and KDFE$_L$-BVP, when $L$ increases for various $kR$ products.} 
\label{FFEConvergence3}
\end{figure}

Some of the attributes of the novel farfield ABCs have been highlighted in various sections of this article. 
Among the most relevant attributes we find the exact character of these absorbing conditions according to \cite{GivoliReview2}. This means the error between the solutions of KDFE$_L$-BVP  and WFE$_L$-BVP, and the solutions of their corresponding original unbounded problems approaches zero when $L\rightarrow\infty$ and the radius $R$ of the artificial boundary is held fixed. Although, it is not possible to prove this exact property merely from numerical experiments, it is still possible to determine this behavior for moderately large values of $L$. A discussion on this convergence properties follows in the next paragraphs. 

As we pointed out earlier, possibly the most well-known higher order local absorbing boundary condition in two-dimensions is due to Hagstrom and Hariharan \cite{Hagstrom98} which we denote as HH-ABC. The advantage of KDFE sequence of ABCs over the HH counterpart is that the former leads to convergence of the numerical approximation to the exact solution for a fixed value of $R$, while the HH-ABC only converges asymptotically as $R$ increases. In addition to KDFE and WFE farfield ABCs, we also derived KSFE in Section \ref{Section.ABC2DAsymp}, which is a farfield expansion  obtained from a classical asymptotic expansion of Karp's series. This asymptotic expansion is the same 
employed in the derivation of the BGT and HH absorbing conditions in two-dimensions. 

In Fig. \ref{FFEConvergence3} the convergence properties of the numerical FFP obtained form KSFE-BVP and KDFE-BVP are compared. The physical problem is the same scattering problem studied in Section \ref{orderconverg2D} and illustrated in Fig. 
\ref{fig:ScattFields}. However, instead of describing the approximation of the outgoing wave at the artificial boundary, we describe the approximation of the farfield pattern for different values of $kR$ which are obtained for a fixed $R=1.05$ combined with appropriate values of $k$. 

Notice that the convergence of the solutions of KSFE$_L$-BVP is conditioned by the value of the frequency $k$ and radius $R$ of the artificial boundary. More precisely,
for $kR=\pi/2$ and $kR=\pi$, the FFP approximation of KSFE$_L$-BVP begins to diverge from the exact FFP for $L\ge 4$ and $L\ge 7$, respectively. However for $kR=2\pi$, the FFP of KSFE$_L$-BVP converges to the exact FFP when $L$ increases, for $1\le L\le 15$. Furthermore, this approximation is as good as the one obtained using the exact DtN
boundary condition for $10\le L\le 15$.
However, as we continue increasing the number of terms $L$, the solutions of KSFE$_L$-BVP will eventually diverge. This behavior parallels the one established by a rigorous proof given by Schmidt and Heier \cite{Kersten} for the convergence properties of the solution obtained using Feng's absorbing boundary conditions. Feng's condition arises from an asymptotic expansion of the exact DtN boundary condition for large $R$. 

In practical terms, the use of KSFE-ABC is advisable only if the product $kR$ is large enough which is also applicable to any absorbing condition obtained from an asymptotic expansion of series representation of the outgoing waves.  The application of KSFE$_L$  is still useful in many physical problems where $kR$ is sufficiently large
since it takes only a few terms to reach the same order of convergence than the one obtained from DtN-ABC.
On the other hand, the exact character of KDFE-BVP is clearly shown in Fig. \ref{FFEConvergence3}. In fact, it only takes four terms of Karp's expansion ($L=4$) to reach the same level of convergence of the solution of the DtN-BVP when $kr=\pi/2$. This level is maintained until $L=15$. Similar behavior is observed for the other two values of the product $kR$. In all these experiments $R=1.05$ and the frequency $k$ was chosen according to the desired value of $kR$. We also employed the same grid in all these experiments. 

A non-asymptotic version of BGT$_2$ can be obtained by constructing  a second order operator that annihilates 
the terms of $O \left( 1\right)$  in 
Karp's expansion. This was the approach followed by Grote and Keller \cite{Grote-Keller01}  to obtain the second order differential operator,
\begin{equation}
L_0u = \partial_r u - k \left( \frac{H_0'(kr)}{H_0(kr)}u - \left(\frac{H_0'(kr)}{H_0(kr)} -\frac{H_1'(kr)}{H_1(kr)}\right)\partial_{\theta}^2 u\right)\label{BGTH2}
\end{equation}
 An alternative derivation of (\ref{BGTH2}) was given by Li and Cendes \cite{Li-Cendes} by requiring that the first two terms of the exact solution of normal modes of Helmholtz equation in cylindrical coordinates were annihilated.  All these authors and more recently Turkel, Farhat, and Hetmaniuk \cite{Turkel-Farhat2004} used the differential operator (\ref{BGTH2}) at the artificial boundary as an ABC for the scattering of a plane wave from a circular obstacle. They noticed the superior accuracy of the solution obtained with this condition compared with the one obtained from the absorbing boundary conditions  BGT$_L$  (for $L=1,2$),
  for low values of the frequency $k$. In particular,  Turkel et al. \cite{Turkel-Farhat2004} showed that for a frequency $k=0.01$ and radius $R=5$ (artificial boundary)  the $L^2$-norm relative error at the artificial boundary  is about 50 times better using (\ref{BGTH2}) over BGT$_2$.
  These results can be considered as a low order version of the results illustrated in Fig. \ref{FFEConvergence3} for the high order local KSFE and KDFE absorbing boundary conditions.  Zarmi and Turkel \cite{Zarmi-Turkel} also arrived to the same conclusion by comparing their higher order version of Li and Cendes' operator in 2D  with the higher order versions of HH operators obtained from the asymptotic Karp's expansion.

We would like to highlight two other valuable attributes of the farfield ABCs. First, the farfield pattern is the coefficient of the leading term of the farfield expansion. This leading coefficient (angular function) is 
one of the unknowns of the linear system to be solved to obtain the approximation of the exact solution. So, there is no additional computation afterward to obtain the FFP.  In most of our experiments, we decided to use the FFP approximation formula obtained in Section \ref{NumericalFFP} for comparison purposes.
Secondly, by increasing the parameter $L$ (number of terms in the expansion), the error introduced by KDFE$_L$ and WFE$_L$ can easily be reduced and made negligible compared with the error from the numerical method in the interior domain $\Omega^{-}$. 

There are numerous directions in which the application of farfield ABCs can be extended. Some  of those on which we are currently working or plan to work are the following:
\begin{enumerate}
\item  {The combined formulation of high order finite difference (or finite element), for the discretization of the Helmholtz equation in $\Omega^{-}$, with the novel exact local farfield ABCs. This will show the high accuracy that can be achieved by simply increasing the number of terms in the ABCs expansion, using relatively coarse grids. For this purpose, we plan to explore several high order compact finite difference schemes that have been recently developed  \cite{Britt-Tsynkov-Turkel2010,Turkel-Gordon2013} and others well-established found in \cite{Lele1992}.}

\item  The extension of the formulation of our ABC to the wave equation (time-domain). This extension is clearly feasible in 3D since the time-domain analogue of the Wilcox expansion is available \cite{Bayliss-TurkelWave1980,Grote03} due to the Fourier duality between $\partial_{t}$ and $ik$, and between $e^{ikr}$ and time shift. This is also valid for the KSFE-ABC in 2D. However, for the KDFE absorbing condition in 2D, Karp's expansion has no closed-form transformation to the time domain due to the complexity of the terms $H_{0}(kr)$ and $H_{1}(kr)$. Such a transformation would lead to nonlocal operators in the time variable similar to the ones discussed in \cite{Alpert-Greengard-Hagstrom2000}.   

\item Construction of exact local farfield ABCs for multiple scattering of time-harmonic waves. The farfield expansions of Wilcox and Karp allow the evaluation of the scattered field semi-analytically at any point outside the artificial boundary. This property is fundamental in the multiple scattering setting for the introduction of artificial sub-boundaries enclosing obstacles disjointly \textit{\`{a} la} Grote-Kirsch \cite{Grote01}.

\end{enumerate}

%


\section*{Acknowledgments}
The first and third authors acknowledge the support provided by the Office of Research and Creative Activities (ORCA) of Brigham Young University.
Thanks are also due to the referees for their useful suggestions.


\bibliographystyle{elsarticle-num}
\bibliography{AcoBib}

\end{document}